%% file: main.tex
\documentclass[reflite,authorcolumns,numberwithinsect,a4paper]{no-lipics}
\usepackage{booktabs}
\usepackage{multirow}
\errorstopmode

\input{macros}

\title{The Parametrised Complexity of Counting~Small~Sub-Hypergraphs}

\author{Marco Bressan}{Department of Computer Science\\University of
Milan, Italy}{marco.bressan@unimi.it}{}{}
\author{Julian Christoph Brinkmann}{Department of Computer Science\\Goethe University Frankfurt,
Germany}{J.Brinkmann@em.uni-frankfurt.de}{https://orcid.org/0009-0000-0332-4543}{}
\author{Holger Dell}{Basic Algorithms Research Copenhagen,
  Denmark\\IT University of Copenhagen,
Denmark}{hold@itu.dk}{https://orcid.org/0000-0001-8955-0786}{}
\author{Marc Roth}{School of Electronic Engineering and Computer
  Science\\ Queen Mary University of
London}{m.roth@qmul.ac.uk}{https://orcid.org/0000-0003-3159-9418}{}
\author{Philip Wellnitz}{National Institute of Informatics\\The
  Graduate University for Advanced Studies, SOKENDAI\\Tokyo,
Japan}{wellnitz@nii.ac.jp}{https://orcid.org/0000-0002-6482-8478}{}

\acknowledgements{Parts of this work have been done while the authors
  were visiting Dagstuhl
Seminar 22482 ``Counting and Sampling: Algorithms and Complexity''.}

\authorrunning{M. Bressan, J. C. Brinkmann, H. Dell, M. Roth, and P. Wellnitz}

\begin{document}

\maketitle
\pagenumbering{roman}
\input{abstract}

\thispagestyle{plain}
\tinytoc

\clearpage
\pagenumbering{arabic}

\input{s1_intro}
\input{s2_techo}

\input{s3_prelims}
\input{s4_hyper_basis}

\input{s5_sub_counting}

\input{s6_indsub_counting}
\input{s7_trim}

\bibliographystyle{alphaurl-new}
\bibliography{main}

\newpage

\appendix
\input{sa_dedekind}

\end{document}

%% file: macros.tex
\mathchardef\mhyphen="2D
\newcommand{\Oh}{\mathcal{O}}

\def\emptyset{\varnothing}

\newcommand\hypersub[2]{\hyperref[#1]{\vphantom{#2}\smash{#2_{}\kern-\scriptspace}}}

\newcommand{\gaifman}{\mathsf{G}}

\newcommand{\trimsub}[2]{#1\langle{#2}\rangle}
\newcommand{\trim}[2]{\trimsub{#1}{#2}}
\newcommand{\trimhoms}[2]{\mathsf{TrimHom}(#1\to#2)}
\newcommand{\trimhomsprob}{\text{\ref{prob:trimhoms}}}
\newcommand{\cptrimhoms}[2]{\mathsf{cpTrimHom}(#1\to#2)}
\newcommand{\cftrimhoms}[2]{\mathsf{cfTrimHom}(#1\to#2)}
\newcommand{\cptrimhomsprob}{\textupsc{CpTrimHom}}
\newcommand{\cftrimhomsprob}{\textupsc{CfTrimHom}}
\newcommand{\trimembs}[2]{\mathsf{TrimEmb}(#1\to#2)}
\newcommand{\strtrimembs}[2]{\mathsf{StrTrimEmb}(#1\to#2)}
\newcommand{\trimstrembs}[2]{\mathsf{StrTrimEmb}(#1\to#2)}
\newcommand{\trimsubs}[2]{\mathsf{TrimSub}(#1\to#2)}
\newcommand{\trimsubsprob}{\#\textupsc{TrimSub}}
\newcommand{\indtrimsubs}[2]{\mathsf{IndTrimSub}(#1\to#2)}
\newcommand{\indtrimsubsprob}{\text{\ref{prob:trimindsub}}}

\def\tensor{\operatorname{\hypersub{def:hyp_tensor}{\otimes}}}

\def\fredgeco{\hypersub{def:fredgeco}{\rho^\ast}}
\def\fredgecoS#1{\hypersub{def:fredgeco}{{\rho^\ast_{#1}}}}
\def\indno{\hypersub{def:indno}{\alpha}}
\def\frindno{\hypersub{def:frindno}{\alpha^\ast}}

\def\frcoindno{\hypersub{def:frcoindno}{\sigma^\ast}}
\def\frcoindnoS#1{\hypersub{def:frcoindno}{{\sigma^\ast_{#1}}}}

\def\fhtw{\hypersub{def:widthmeasures}{\operatorname{fhw}}}
\def\aw{\hypersub{def:widthmeasures}{\operatorname{aw}}}
\def\tw{\hypersub{def:widthmeasures}{\operatorname{tw}}}
\newcommand\subw{\operatorname{subw}}

\newcommand{\setminusnotrim}[2]{#1 - #2}
\newcommand{\setminustrim}[2]{#1 \setminus #2}

\newcommand{\textupsc}[1]{\textup{\textsc{#1}}}

\newcommand{\dotcup}{\mathbin{\dot{\cup}}}
\newcommand{\poly}{\operatorname{poly}}
\newcommand{\supp}{\operatorname{supp}}

\newcommand{\subsprob}{\text{\ref{prob:sub}}}
\newcommand{\indsubsprob}{\text{\ref{prob:indsub}}}
\newcommand{\homsprob}{\text{\ref{prob:homs}}}

\newcommand{\cphomsprob}{\text{\ref{prob:cphoms}}}
\newcommand{\cfhomsprob}{\text{\ref{prob:cfhoms}}}
\newcommand{\img}{\ensuremath{\mathsf{img}}}

\newcommand{\subs}[2]{\mathsf{Sub}(#1 \to #2)}
\newcommand{\indsubs}[2]{\mathsf{IndSub}(#1 \to #2)}

\newcommand{\id}{\mathsf{id}}

\newcommand{\W}{\mathrm{W}}

\newcommand{\N}{\mathbb{N}}
\newcommand{\Q}{\mathbb{Q}}

\newcommand{\R}{\mathbb{R}}
\newcommand{\Rp}{\mathbb{R}_{\ge 0}}
\newcommand{\scH}{\mathcal{H}}

\newcommand{\scN}{\mathcal{N}}
\newcommand{\scG}{\mathcal{G}}
\newcommand{\scF}{\mathcal{F}}

\newcommand{\scB}{\mathcal{B}}

\newcommand{\scT}{\mathcal{T}}
\newcommand{\scQ}{\mathcal{Q}}
\newcommand{\scS}{\mathcal{S}}

\newcommand{\hyperME}{\textupsc{\ref{prob:hyperme}}}

\newcommand{\ccStyle}[1]{\ensuremath{\mathrm{#1}}\xspace}

\newcommand{\ccSharpP}{\#\ccStyle{P}}
\newcommand{\ccNP}{\ccStyle{NP}}
\newcommand{\ccFPT}{\ccStyle{FPT}}

\newcommand{\ccSharpW}[1]{\#\ccStyle{W}[#1]}

\newcommand{\clique}{\textupsc{Clique}}

\newcommand{\fptred}{\leq^{\scalebox{.6}{\ensuremath{\mathsf{FPT}}}}}
\newcommand{\fpteq}{\equiv^{\scalebox{.6}{\ensuremath{\mathsf{FPT}}}}}
\newcommand{\fptredlin}{\leq^{\scalebox{.6}{\ensuremath{\mathsf{FPT\text{-}LIN}}}}}
\newcommand{\fpteqlin}{\equiv^{\scalebox{.6}{\ensuremath{\mathsf{FPT\text{-}LIN}}}}}

\newcommand{\homs}[2]{\mathsf{Hom}(#1 \to #2)}

\newcommand{\cphoms}[2]{\mathsf{cpHom}(#1 \to #2)}
\newcommand{\cfhoms}[2]{\mathsf{cfHom}(#1 \to #2)}

\newcommand{\auts}[1]{\mathsf{Aut}(#1)}

\renewcommand{\phi}{\varphi}
\newcommand{\CfCommonNeighProb}{\text{\ref{prob:common_neigh}}}

\DeclareMathOperator{\rank}{rank}

\newcommand{\embs}[2]{\mathsf{Emb}(#1 \to #2)}
\newcommand{\strembs}[2]{\mathsf{StrEmb}(#1 \to #2)}

\newcommand{\G}{\mathcal{G}}

\newcommand{\proj}{\ensuremath{\operatorname{proj}}}
\newcommand{\join}{\ensuremath{\operatorname{join}}}

\newcommand{\WeakMinors}{\mathcal{T}}
\newcommand{\fwidth}[1]{{#1}\mhyphen\mathsf{width}}

\newcommand{\size}[1]{\lVert{#1}\rVert}

\newcommand{\ignore}[1]{}

%% file: abstract.tex
\begin{abstract}
    Subgraph counting is a fundamental and well-studied problem whose
    computational complexity is well understood.
    Quite surprisingly, the \emph{hypergraph} version of subgraph
    counting has been almost ignored.
    In this work, we address this gap by investigating the most basic
    sub-hypergraph counting
    problem: given a (small) hypergraph $H$ and a (large) hypergraph
    $G$, compute the number of sub-hypergraphs of $G$ isomorphic to $H$.
    Formally, for a family~$\mathcal{H}$ of hypergraphs, let $\#\textsc{Sub}(\mathcal{H})$ be
    the restriction of the problem to $H\in \mathcal{H}$; the induced variant
    $\#\textsc{IndSub}(\mathcal{H})$ is defined analogously.
    Our main contribution is a complete classification of the complexity of
    these problems.
    Assuming the Exponential Time Hypothesis, we prove that $\#\textsc{Sub}(\mathcal{H})$ is
    fixed-parameter tractable if and only if $\mathcal{H}$ has bounded
    \emph{fractional co-independent edge-cover number}, a novel hypergraph parameter we introduce.
    Moreover, $\#\textsc{IndSub}(\mathcal{H})$ is fixed-parameter
    tractable if and only if $\mathcal{H}$ has bounded \emph{fractional edge-cover number}.
    Both results subsume pre-existing results for graphs as special cases.
    We also show that the fixed-parameter tractable cases of
    $\#\textsc{Sub}(\mathcal{H})$ and $\#\textsc{IndSub}(\mathcal{H})$
    are unlikely to be in polynomial time, unless respectively
    $\#\mathrm{P}=\mathrm{P}$ and $\textsc{Graph Isomorphism} \in \mathrm{P}$.
    This shows a separation with the special case of graphs, where the
    fixed-parameter tractable cases are known to actually be in polynomial time.

    From a technical standpoint, we turn to the hypergraph homomorphism
    basis and lift the
    complexity monotonicity principle due to Curticapean, Dell, and
    Marx [STOC 2017] from
    graphs to hypergraphs of unbounded rank.
    Moreover, we crucially rely on the integrality gap for fractional
    independent sets based
    on adaptive width due to Bressan, Lanzinger, and Roth [STOC 2023].
    Our proofs rely on a careful investigation of the adaptive width of the patterns that
    survive in the hypergraph homomorphism basis.
    We also consider a natural variant of sub-hypergraphs where edges
    are trimmed to the
    vertex subset; we show that, surprisingly, in this case complexity
    monotonicity fails.
\end{abstract}

%% file: s1_intro.tex
\section{Introduction}\label{sec:intro}

In the subgraph counting problem, the task is to count the copies of a
graph $H$ (the \emph{pattern}) in a graph $G$ (the \emph{host}).
As one of the most studied problems in computer science,
a multitude of results give insights into the computational complexity of counting
subgraphs.
We understand the problem tractability along the
class of allowed patterns, both for the induced case~\cite{ChenTW08}
and the not-necessarily induced case~\cite{CurticapeanM14}.
In addition, refined complexity characterizations exist for sparse
graphs~\cite{BR21}, for linear-vs.-non-linear
tractability~\cite{BeraGLSS22}, for dense
graphs~\cite{bressan_counting_2024}, for counting subgraphs with
hereditary and edge-monotone
properties~\cite{FockeR22,doring_counting_2024,DMW25}, and more.
Besides the results themselves, powerful conceptual tools, such
as the ``homomorphism basis''~\cite{CurticapeanDM17} and Fourier
analysis~\cite{CN25-FourierAnalysis}, provide a unifying view
of subgraph counting.

In stark contrast, our understanding of counting substructures of
\emph{hypergraphs} is
almost non-existent:
besides partial bounds for counting homomorphisms, very little is
known.
Yet, a growing body of literature highlights the many applications of
sub-hypergraph counting, ranging from computational
biology~\cite{hwang_learning_2008},
computer
vision~\cite{huang_image_2010,jun_yu_adaptive_2012,agarwal_beyond_2005}, link
prediction~\cite{benson_simplicial_2018,hwang_ahp_2022,li_link_2013}, and
clustering~\cite{amburg_clustering_2020,zhou_learning_2006}, to name
just a few examples.
This is not surprising: counting subgraphs has itself several
applications, notably in
social network analysis and computational biology, and hypergraphs
are strictly more
expressive than graphs.
Thus, one would expect sub-hypergraph counting to receive a
comparable amount of attention
as subgraph counting.
In this work, we make a first and substantial step towards
understanding the complexity of
counting substructures in hypergraphs.

A hypergraph is a pair $H=(V,E)$ where $V$ is a finite set
and $E \subseteq 2^V \setminus \{\emptyset\}$.
Given two hypergraphs $H$ and $G$, the number of copies of $H$ in $G$
is the number of
sub-hypergraphs of $G$ that are isomorphic to $H$.%
\footnote{By sub-hypergraph of $G$ we mean a hypergraph $G'$ with
  $V(G') \subseteq V(G)$
and $E(G')\subseteq E(G)$.}
We may phrase the basic subgraph counting problem for hypergraphs as
follows: for a
family of hypergraphs $\scH$, given a hypergraph $H \in \scH$ and a
hypergraph $G$, count
the number of copies of $H$ in $G$.
We denote this problem by $\subsprob(\scH)$, and we define similarly
the \emph{induced}
variant~$\indsubsprob(\scH)$ of the problem.
The goal of this paper is to understand the complexity of these
problems when parametrised by the size of the pattern hypergraph~$H$.
\begin{center}
    For which hypergraph families $\scH$ are the problems
    $\subsprob(\scH)$ and
    $\indsubsprob(\scH)$ fixed-parameter tractable,\linebreak that is,
    solvable in time \(f(|H|)\cdot |G|^{O(1)}\) where \(f\) is some
    computable function?
\end{center}
We answer this question with full complexity dichotomies expressed
through two invariants.
The first one, the \emph{fractional co-independent edge-cover number}
$\frcoindno(\scH)$, is introduced in this work and characterises the
tractability of $\subsprob(\scH)$.
The second one, the \emph{fractional edge-cover number}
$\fredgeco(\scH)$ of Grohe and Marx~\cite{GroheM14}, characterises
the tractability of $\indsubsprob(\scH)$.
Both invariants are defined formally in \cref{sec:overview-notions}.
Our complexity classifications are as follows.

\begin{restatable*}
    {mtheorem}{rstmthmone}
    \dglabel{thm:classification_subsprob}[lem:classification_subsprob_1,lem:classification_subsprob_2]
    Fix a recursively enumerable family $\scH$ of hypergraphs.
    \begin{itemize}
        \item If $\scH$ has bounded fractional co-independent edge-cover number,
            $\frcoindno(\scH)<\infty$, then $\subsprob(\scH)$ parametrised by~$|H|$ is
            FPT and solvable in time $f(H) \cdot
            \size{G}^{\frcoindno(\scH)+O(1)}$ for some function $f$.
        \item If $\scH$ has unbounded fractional co-independent edge-cover number,
            $\frcoindno(\scH)=\infty$, then $\subsprob(\scH)$ parametrised
            by~$|H|$ is not
            FPT and not solvable in time $f(H) \cdot
            \size{G}^{o\big(\sqrt[4]{\frcoindno(H)}\big)}$ for any function
            $f$, unless ETH fails. \qedhere
    \end{itemize}
\end{restatable*}

\begin{restatable*}
    {mtheorem}{rstmthmtwo}
    \dglabel{thm:classification_indsubsprob}[lem:classification_indsubsprob_1,lem:classification_indsubsprob_2]
    Fix a recursively enumerable family $\scH$ of hypergraphs.
    \begin{itemize}
        \item If $\scH$ has bounded fractional edge-cover number,
            $\fredgeco(\scH)<\infty$, then $\indsubsprob(\scH)$
            parametrised by~$|H|$ is
            FPT and solvable in time $f(H) \cdot
            \size{G}^{\fredgeco(\scH)+O(1)}$ for some function $f$.
        \item If $\scH$ has unbounded fractional edge-cover number,
            $\fredgeco(\scH)=\infty$, then $\indsubsprob(\scH)$
            parametrised by~$|H|$ is
            not FPT and not solvable in time $f(H) \cdot
            \size{G}^{o\big(\sqrt[4]{\fredgeco(H)}\big)}$ for any function
            $f$, unless ETH fails.
            \qedhere
    \end{itemize}
\end{restatable*}

\Cref{sec:technical_overview} contains formal definitions of the
notation\footnote{One may as well use the provided hyperlinks to jump
to their definitions.} of \cref{thm:classification_subsprob} and
\cref{thm:classification_indsubsprob}, as well as an overview of their proofs.
\Cref{fig:sunflowers} and \cref{ghzqkxmeqg} show examples of
a hard and an easy hypergraph family for the two problems.
Next, we discuss the challenges that make the criteria of
\cref{thm:classification_subsprob,thm:classification_indsubsprob}
different, and more involved to analyse, compared to their equivalents
for graphs.

\begin{figure}[t]
    \renewcommand\tabularxcolumn[1]{m{#1}}
    \centering
    \begin{tabularx}{\linewidth}{*{2}{>{\centering\arraybackslash}X}}
        \includegraphics[scale=1.5]{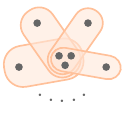}
    &
    \includegraphics[scale=1.5]{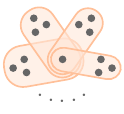}
    \\
    \begin{subfigure}[t]{\linewidth}
        \caption{The family of $4$-uniform sunflowers with core size three.
            By \cref{thm:classification_subsprob}, counting
        the sunflowers is easy for this family.}
    \end{subfigure}
    &
    \begin{subfigure}[t]{\linewidth}
        \caption{The family of $4$-uniform sunflowers with core size one.
            By \cref{thm:classification_subsprob}, counting
        the sunflowers is hard for this family.}
    \end{subfigure}
    \end{tabularx}
    \caption{Two families of \emph{sunflowers}, that is, hypergraphs
        whose edges pairwise
        intersect in the same set, called \emph{core}.
        The example can be easily extended to unbounded rank by taking
        the families of $\ell$-uniform sunflowers with cores of sizes $1$ and
    $\ell-1$, respectively.}
    \label{fig:sunflowers}
\end{figure}

\begin{figure}[t]
    \renewcommand\tabularxcolumn[1]{m{#1}}
    \centering
    \begin{tabularx}{\linewidth}{*{2}{>{\centering\arraybackslash}X}}
        \includegraphics[scale=.9]{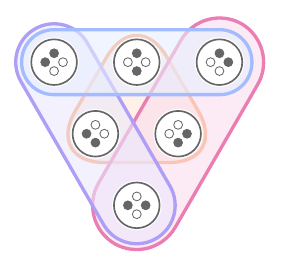}
    &
    \includegraphics[scale=.9]{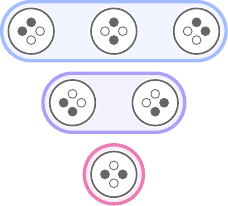}
    \\
    \begin{subfigure}[t]{\linewidth}
        \caption{The hypergraph family from
            \cite[Example~4.2]{GroheM14}:
            for an integer \(n > 0\), the hypergraph has a vertex for each
            size-\(n\) subset of \(\{1,\dots,2n\}\) (depicted: \(n=2\)) and for
            each \(1 \le i \le 2n\) has an edge containing all vertices
            that include \(i\).
            We obtain a fractional edge cover of \(2\) by assigning
            \(1/n\) to each edge.
            By \cref{thm:classification_indsubsprob}, counting
        is easy for this family.}
    \end{subfigure}
    &
    \begin{subfigure}[t]{\linewidth}
        \caption{Consider the family of all hypergraphs whose edges are
            pairwise disjoint.
            Clearly, this family has an unbounded fractional edge cover.
            By \cref{thm:classification_indsubsprob}, counting
        is hard for this family.}
    \end{subfigure}
    \end{tabularx}
    \caption{Two families of \emph{hypergraphs} that fall into either
        of the two cases of
    \cref{thm:classification_indsubsprob}.}
    \label{ghzqkxmeqg}
\end{figure}

Consider the problem of counting \emph{homomorphisms}.
If $H$ and $G$ are hypergraphs, a homomorphism from $H$ to $G$ is a
map $\varphi \colon
V(H) \to V(G)$ that preserves edges, that is, such that $\varphi(e)
\in E(G)$ for every $e
\in E(H)$.
The problem $\homsprob(\scH)$ asks, given $H \in \scH$ and a
hypergraph~$G$ as input, to
count all homomorphisms from $H$ to $G$.
When restricted to graphs, the complexity of counting homomorphism
has been fully
understood for over 20 years~\cite{DalmauJ04}.
Assuming ETH, $\homsprob(\scH)$ is polynomial-time computable if and only if the
\emph{treewidth}%
\footnote{A measure of how close a graph is to being a tree: paths
    and trees have constant
treewidth, whereas grids and complete graphs have unbounded treewidth.}
of $\scH$ is bounded.

In stark contrast, a complete dichotomy for counting homomorphisms
between general
hypergraphs is still an open problem; only partial results are known.
Loosely speaking, this is because hypergraphs have \emph{unbounded
rank}---that is, edges
of arbitrarily large cardinality---and this changes the nature of the
problem radically, so that many of the tractability criteria from the graph version do not
carry over.
For instance, it is well-known that the treewidth of hypergraphs
(which is the treewidth
of the underlying Gaifman graphs) is \emph{not} the correct criterion
for the tractability
of $\homsprob(\scH)$.\footnote{Let $\scH$ be the family of
  hypergraphs consisting of one
  edge that contains all vertices. While the Gaifman graphs of $\scH$
  have unbounded
  treewidth, the problem $\homsprob(\scH)$ essentially boils down to
  counting edges of
size~$=|H|$.}
Instead, one needs to turn to more complex generalizations of
treewidth, such as the
fractional hypertree width~\cite{DBLP:journals/jcss/GottlobLS02}, the submodular
width~\cite{Marx13}, or the fractional edge-cover number~\cite{GroheM14}.

Similarly to the case of $\homsprob(\scH)$, it is not hard to see
that the tractability
criteria for $\subsprob(\scH)$ and $\indsubsprob(\scH)$ do not
directly carry over from
graphs to hypergraphs either.
For instance, in the case of graphs $\indsubsprob(\scH)$ is tractable
if and only if
$\scH$ is finite, but in the case of hypergraphs $\indsubsprob(\scH)$
can be tractable
even for infinite classes.
We overcome these challenges in the present paper: curiously, we
obtain our \emph{full}
classification results for $\subsprob(\scH)$ and $\indsubsprob(\scH)$
from the known,
\emph{partial} results for $\homsprob(\scH)$.

We close with a remark on the lower bounds of
\cref{thm:classification_subsprob,thm:classification_indsubsprob};
in particular their dependence on $G$.
The terms $\size{G}^{o(\sqrt[4]{\frcoindno(H)})}$ and
$\size{G}^{o(\sqrt[4]{\fredgeco(H)})}$, respectively,
mirror precisely the result by Marx on the
complexity of counting homomorphisms, which under ETH cannot be done
in time $f(H)\cdot \size{G}^{o(\sqrt[4]{\aw(H)})}$, see
\cref{lem:unbounded_aw_LB}.
Understanding whether Marx's bound is tight or not is a widely open
problem in the area of subgraph counting.
The proofs of our lower bounds are constructed carefully so as
to inherit precisely that bound, without introducing further gaps.
Therefore, our lower bounds are the best one could prove without
making substantial progress on said open problem.
In fact, our proofs show that any improvement on
our lower bounds implies an improvement on Marx's bounds.

\paragraph*{Discussion of our result for $\subsprob$
(\cref{thm:classification_subsprob})}
For graphs, the fractional co-independent edge-cover
number~$\frcoindno$ is asymptotically
equivalent to the vertex-cover number, which is the tractability
criterion for subgraph
counting proven in~\cite{CurticapeanM14}.
This means that \cref{thm:classification_subsprob} captures the graph
setting as a
special case, as expected.
For hypergraphs, $\frcoindno$ is \emph{not} equivalent to the
vertex-cover number of the
hypergraph for two reasons: First, a vertex-cover (or ``hitting
set'') of a hypergraph
intersects every edge, whereas we need a \emph{co-independent set},
which intersects every
edge in all but at most one vertex.
Secondly, we truly need the fractional version, since we may design
families $\scH$ where
$\frcoindno(\scH)$ is bounded but its non-fractional version is unbounded.

In the case of graphs, the fixed-parameter tractable cases of
$\subsprob(\scH)$ are in
fact polynomial-time computable, see~\cite{CurticapeanM14}.
In contrast, for hypergraphs, the same is most likely not true.
\begin{restatable*}{theorem}{rstmtonelemone}\label{rem:4-3-1}
    There are hypergraph families $\scH$ with $\frcoindno(\scH)<\infty$ where
    $\subsprob(\scH)$ is $\ccSharpP$-hard.
\end{restatable*}

\paragraph*{Discussion of our results for
$\indsubsprob$ (\cref{thm:classification_indsubsprob})}
For graphs, the fractional edge-cover number~$\fredgeco$ specialises
to essentially the
order of the graph,%
\footnote{This is true if every vertex has its own singleton edge,
    which we may assume
without loss of generality, see \cref{sec:prelim}.}
which is the tractability criterion for $\indsubsprob(\scH)$ if
$\scH$ is a class of
graphs \cite{CurticapeanDM17}.
Thus, \cref{thm:classification_indsubsprob} captures the graph
setting as a special case.
Again, it is easy to see that $\fredgeco(H)$ is not equivalent to
$|H|$ for hypergraphs,
and that there are classes with unbounded~$|H|$ but bounded $\fredgeco(H)$.

As above, one may wonder whether the case $\fredgeco(\scH)< \infty$
is not just $\ccFPT$ but even polynomial-time computable.
Interestingly, we seem to get intermediate problems in this case.
\begin{restatable*}{theorem}{rstmtlemone}
    \dglabel{lem:quasiP_indsubs}(If $\fredgeco(\scH)<\infty$, then
    $\indsubsprob(\scH)$ can be solved in quasi-polynomial time)
    If $\fredgeco(\scH)<\infty$, then $\indsubsprob(\scH)$ can be solved in
    quasi-polynomial time in the size of the input hypergraphs; that is, in time
    $(\size{H}+\size{G})^{(\ln \size{H})^{\Oh(1)}}$.
\end{restatable*}
Let us briefly discuss the quasi-polynomial running time in
\cref{lem:quasiP_indsubs}.
First, we solve \(\indsubsprob(\scH)\) by a quasi-polynomial number
of hypergraph isomorphism subproblems, each requiring quasi-polynomial
time.
Due to the quasi-polynomial running time, if
$\indsubsprob(\scH)$ was $\ccSharpP$-hard for
$\fredgeco(\scH)<\infty$, then all problems in
$\ccSharpP$ and $\ccNP$ had quasipolynomial-time algorithms---a rather
unlikely scenario.
Second, we show that improving \cref{lem:quasiP_indsubs} to polynomial time
would imply the same for $\textupsc{Graph Isomorphism}$.
\begin{restatable*}{theorem}{rstmtlemtwo}
    \dglabel{lem:GI_hard_indsubs}(There are hypergraph families $\scH$ with
    $\fredgeco(\scH)<\infty$ that admit a polynomial-time reduction
    from $\textupsc{Graph Isomorphism}$ to
    $\indsubsprob(\scH)$.)
    There are hypergraph families $\scH$ with $\fredgeco(\scH)<\infty$ that admit
    a polynomial-time reduction from $\textupsc{Graph Isomorphism}$ to
    $\indsubsprob(\scH)$.
\end{restatable*}

\paragraph*{Trimmed sub-hypergraphs, and an impossibility result}
Compared to graphs, when considering potential definitions for an induced sub-hypergraph
we have multiple options as
to how we should treat edges from which some, but not all vertices are deleted.
Indeed, while in the notions considered so far, edges
that contain a deleted vertices are dropped entirely,
a perhaps more natural option would be to instead \emph{trim} such edges,
that is, to keep the edge but with the deleted vertices removed---giving rise to
\emph{trimmed sub-hypergraphs}.
Formally, a {trimmed sub-hypergraph} of $G$ is any hypergraph $G'$ with $V(G') \subseteq
V(G)$ such that every $e' \in E(G')$ has the form $e' = e \cap V(G')$ for some $e \in
E(G)$.
The trimmed sub-hypergraph of $G$ \emph{induced} by $X \subseteq V(G)$ is
the hypergraph
$\trimsub{G}{X} \coloneqq \big(X, \trim{E}{X}\big)$, where
$\trim{E}{X} = \{e \cap X : e
\in E(G), e \cap X \ne \emptyset\}$.
Consult \cref{fig:subgraphs} for an example; we call the resulting problems
$\trimsubsprob(\scH)$ and $\indtrimsubsprob(\scH)$.

\begin{figure}[t]
  \renewcommand\tabularxcolumn[1]{m{#1}}
  \centering
  \begin{tabularx}{\linewidth}{*{4}{>{\centering\arraybackslash}X}}
    \includegraphics[scale=1.39]{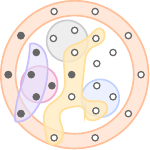}
    &
    \includegraphics[scale=1.52]{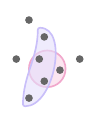}
    &
    \includegraphics[scale=1.52]{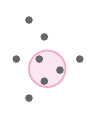}
    &
    \includegraphics[scale=1.52]{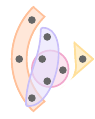}
    \\
    \begin{subfigure}[t]{\linewidth}
      \caption{
        A hypergraph in the shape of a pizza.
        Depicted are three salami (\emph{circular region with three vertices}),
        a ham (\emph{region with four vertices}), a blob of cheese (\emph{yellow
        blob}), and the crust (\emph{orange annulus}).
      }\label{fig:pizza-a}
    \end{subfigure}
    &
    \begin{subfigure}[t]{\linewidth}
      \caption{The sub-hypergraph induced by the black vertices of
        \cref{fig:pizza-a}.
        Observe that only one salami and one ham are entirely contained in the
        \emph{black} vertices---all other edges, such as the crust are dropped.
      }\label{4-4-to}
    \end{subfigure}
    &
    \begin{subfigure}[t]{\linewidth}
      \caption{Deleting (eating) a further edge compared to
        \cref{4-4-to} (in this
      case one ham), we obtain a sub-hypergraph of \cref{fig:pizza-a}.}
    \end{subfigure}
    &
    \begin{subfigure}[t]{\linewidth}
      \caption{
        The induced trimmed sub-hypergraph of the pizza slice formed by the
        \emph{black} vertices of \cref{fig:pizza-a}, parts of the crust and the
        cheese of the pizza are trimmed and survive.
      }
    \end{subfigure}
  \end{tabularx}
  \caption{A hypergraph and the various notions of subgraphs that we
    use in this work,
    with vertices (\emph{black} and \emph{white bullets}) and edges
  (\emph{shaded regions}).}
  \label{fig:subgraphs}
\end{figure}

As it turns out, switching to \emph{trimming edges}
drastically changes both the complexity of the corresponding pattern counting problems and their
inherent algebraic structure.
Tracing the steps of our previous approach, we observe that while the ideas of
Lovász~\cite{Lovasz12} still apply,
the complexity monotonicity principle fails
for the trimmed version of homomorphisms.
Indeed, we may express $\#\indtrimsubs{H}{\star}$\footnote{ $f(x,\star)$ denotes the function $y \mapsto f(x,y)$,
that is, $\#\indtrimsubs{H}{\star}$ maps a hypergraph $G$ to $\#\indtrimsubs{H}{G}$;
$\#\trimsubs{F}{\star}$ and $\#\trimhoms{F}{\star}$ are defined likewise.}
as a linear combination of~$\#\trimsubs{F}{\star}$,
which we may in turn express as a linear combination of $\#\trimhoms{F}{\star}$.
However, we prove the existence of a hypergraph family $\scH$ such that
\begin{itemize}[noitemsep]
  \item $\indtrimsubsprob(\scH)$ is easy, but
  \item some term in the corresponding linear combination of
    $\#\trimhoms{F}{\star}$ is
    hard.
\end{itemize}
See \cref{thm:main_trimmed} for a precise statement of our result.
This implies that one \emph{cannot} prove hardness of
$\indtrimsubsprob$ by leveraging the
hardness of its cousin problem $\trimhomsprob$ in the standard way.
This is somewhat surprising: usually, the natural homomorphism basis
for the counting
problem at hand yields complexity monotonicity, for instance, the
homomorphism basis for
subgraphs or sub-hypergraphs.
Yet, for trimmed sub-hypergraphs this is not true.

We are able to pinpoint exactly where the complexity monotonicity
principle fails: we
prove that, unless all problems in $\#\W[2]$ are fixed-parameter
tractable, there is no
efficiently computable associative hypergraph operator $\boxplus$ that behaves
multiplicatively with respect to \emph{trimmed} homomorphism counts
(as the tensor product
does for graphs).
We also show that replacing the trimmed homomorphism basis with the standard homomorphism basis does not help.
In fact, we show that standard homomorphism counts and trimmed
homomorphism counts are
\emph{linearly independent} functions.
As a consequence, trimmed sub-hypergraph counts \emph{cannot} be
expressed as linear
combinations of standard hypergraph homomorphisms.
These results suggest that alternative notions of sub-hypergraph may
be more delicate to
deal with, and may require novel analytical tools.

\subsection*{Organisation of the paper}
In \cref{sec:technical_overview}, we sketch the techniques and proofs
required to
establish \cref{thm:classification_subsprob,thm:classification_indsubsprob}.
\Cref{sec:prelim} introduces additional background material on hypergraphs, as well as
parametrised and fine-grained complexity theory.
We formally set up the hypergraph homomorphism basis in
\cref{sub:homsprob}, and formally
prove
\cref{thm:classification_subsprob,thm:classification_indsubsprob} as well as
\cref{rem:4-3-1,lem:quasiP_indsubs,lem:GI_hard_indsubs} in \cref{sec:sub}.
Finally, our observations about trimmed sub-hypergraphs are contained in
\cref{sec:trimmed_sec}.

\subsection*{Related work}
\subparagraph*{Pattern counting in graphs.}
In addition to the complexity classifications for counting small
subgraphs~\cite{CurticapeanM14} and induced subgraphs~\cite{ChenTW08}
mentioned in the
opening paragraphs, a tremendous amount of work has been done on related
pattern counting
problems in graphs. Most complexity results on small pattern counting
that were published
after 2017 use the ``homomorphism basis'' in some way. Examples
include results on
counting homomorphisms with inequality constraints due to
Roth~\cite{Roth17}, and on the
generalised induced subgraph counting problem introduced by Jerrum and
Meeks~\cite{JerrumM15}, which was almost completely resolved by Focke and
Roth~\cite{FockeR22}, and by D\"oring, Marx, and
Wellnitz~\cite{doring_counting_2024,DMW25}.

In the context of these and many related results, novel tools for the
analysis of the
homomorphism basis have been developed, which rely on matroid
theory~\cite{Roth17},
simplicial topology~\cite{RothS20}, and Fourier
analysis~\cite{CN25-FourierAnalysis}, to name a few.
However, none of these tools seem to easily transfer to hypergraphs of
unbounded rank, which is what we study in this paper.

\subparagraph*{Pattern counting in hypergraphs and relational structures.}
Pattern counting in hypergraphs has been studied almost exclusively
in the context of
constraint satisfaction problems and database theory. The problems
studied in those
contexts turn out to be versions of homomorphism counting: Given a small pattern
hypergraph $H$ and a large host hypergraph $G$, count the number of
homomorphisms
(edge-preserving mappings) from $H$ to $G$.

Dalmau and Jonsson~\cite{DalmauJ04} establish a parametrised
complexity dichotomy for
counting homomorphisms between hypergraphs of bounded rank. In the
much more challenging
case of unbounded rank, Grohe and Marx~\cite{GroheM14}\footnote{The
  algorithm presented
  in~\cite{GroheM14} is stated for the decision version of the
  homomorphism problem, but it
easily extends to counting.} construct a polynomial-time algorithm for counting
homomorphisms from hypergraphs of bounded fractional hypertree width.
Marx~\cite{Marx13}
partially complements this with a hardness result for counting
homomorphisms from
hypergraphs of unbounded adaptive width. The complexity for unbounded
fractional hypertree
width and bounded adaptive width is still unknown. Nevertheless, the
known complexity
results for the homomorphism counting problem in hypergraphs of
unbounded rank form the
cornerstones of our main results; we discuss this in detail in
\cref{sec:technical_overview}.

In the context of database theory, answers of select-project-join
queries are modelled as
(partial) homomorphisms from a small relational structure (associated
with the query) to a
large relational structure (the database).
Relational structures may be viewed as edge-coloured hypergraphs with
ordered hyperedges.
The problem of counting answers to database queries has been
thoroughly studied by Chen,
Durand, and Mengel~\cite{DurandM15,ChenM15,ChenM16} and by Dell,
Roth, and Wellnitz~\cite{DellRW19}.
The latter work also considers (in the bounded rank setting)
select-project-join queries
with inequalities between the variables of the query.
If all inequalities are present, this class of queries may be used to
count embeddings
into relational structures; however the algorithmic results
in~\cite{DellRW19} on the
complexity of counting answers to such queries do not yield an
explicit tractability
criterion, and they do not apply to the unbounded rank setting.

%% file: s2_techo.tex
\section{Technical overview}\label{sec:technical_overview}
We give a detailed proof overview for
\cref{thm:classification_subsprob,thm:classification_indsubsprob}.
On a high level, we prove both theorems by a careful analysis that
employs the complexity
monotonicity principle, which we lift from graphs to hypergraphs.
This means that we express sub-hypergraph counts as linear
combinations of hypergraph
homomorphism counts, and analyse these linear combinations to
identify terms that are hard
to evaluate.

We begin in \cref{sec:overview-FPT,sec:overview-notions} by more
formally defining the
notions in the statements of these results.
\Cref{sub:overview_basis} introduces our main technical tools: the
hypergraph homomorphism
basis, hypergraph motif parameters, and Dedekind interpolation.
\Cref{sub:overview_subsprob,sub:overview_indsubsprob} show how to
leverage these tools to obtain our classifications for $\subsprob$ and $\indsubsprob$.

\subsection*{Fixed-parameter tractability}\label{sec:overview-FPT}
In line with previous work, our notion of tractability is \emph{fixed-parameter
tractability}.
Loosely speaking, a problem is fixed-parameter tractable if it can be
solved in time
$f(k)\cdot \poly(n)$ for some computable function $f$, where $k \in
\N$ is some suitable
parameter and $n$ is the input size.
In this work, we let $k=|V(H)|$, which is the standard
parametrisation used in previous
work.
The idea is that~$k$ is often much smaller than $n$, hence we may allow for a
superpolynomial dependence on $k$ alone.
One typical example is $H$ representing a query typed by a user and
$G$ representing a
large database.
The class of all fixed-parameter tractable problems is denoted by $\ccFPT$.

To provide evidence that a problem is not fixed-parameter tractable, we use the
Exponential Time Hypothesis (ETH) by Impagliazzo and
Paturi~\cite{ImpagliazzoP01}, which
states that 3-SAT cannot be solved in time $\exp(o(n))$, where~$n$ is
the number of
variables in the input formula.
Moreover, we use hardness for the classes $\ccSharpW{1}$ and
$\ccSharpW{2}$, which are two
analogues of $\ccSharpP$ in parametrised complexity.
ETH implies that $\ccSharpW{1}$-hard and $\ccSharpW{2}$-hard problems are not
fixed-parameter tractable.
For a more thorough introduction to fixed-parameter tractability, see
\Cref{sec:prelim}.

\subsection*{Fractional (co-independent) edge-cover number}\label{sec:overview-notions}
\Cref{thm:classification_subsprob,thm:classification_indsubsprob} refer to the
fractional co-independent edge-cover number and to the fractional
edge-cover number.
The latter notion was introduced by Grohe and Marx~\cite{GroheM14},
but the former appears
to be new.
We introduce both notions here.

Let $H=(V,E)$ be a hypergraph, and fix an $X \subseteq V$.
We say $X$ is \emph{co-independent} if its complement $\overline{X} =
V \setminus X$ is strongly independent, that is, if $|e \cap
\overline{X}| \le 1$ holds for every $e \in E$.
We only consider strongly independent sets in this work and
refer to them simply as independent sets from now on.
A function $\xi\colon E \rightarrow \Rp$ is a \emph{fractional
edge-cover} of $X$ if
$\sum_{e: v\in e} \xi(e) \geq 1$ holds for all $v\in X$.
The \emph{weight} of $\xi$ is $\sum_{e \in E} \xi(e)$.
The fractional edge-cover number
of $X$ in $H$, denoted by $\rho^\ast_H(X)$, is the minimum weight of
a fractional
edge-cover of $X$.
The \emph{fractional edge-cover number} of $H$ is $\fredgeco(H)
\coloneqq\rho^\ast_H(V)$.
Finally, the fractional co-independent edge-cover number of $H$, denoted by
$\frcoindno(H)$, is the minimum fractional edge-cover number
$\rho^\ast_H(X)$ over all
co-independent sets $X \subseteq V$.

\subsection{The hypergraph homomorphism basis}\label{sub:overview_basis}
The hardest part in obtaining complexity classifications is proving lower bounds.
In our case, we prove lower bounds by reduction from the problem of counting
\emph{homomorphisms} between hypergraphs, for which fine-grained lower bounds based on ETH
are already known~\cite{Marx13}.
To perform our reduction, we need to establish a hypergraph version of what is often
called the \emph{complexity monotonicity} principle.
This is a framework that provides a unified view on the complexity of subgraph counting
problems by expressing subgraph counts in the so-called homomorphism basis.

The complexity monotonicity principle for graphs, proved in a landmark paper by
Curticapean, Dell, and Marx~\cite{CurticapeanDM17}, works as follows.
A \emph{graph motif parameter} is any function~$\zeta$ that maps an input graph~$G$ to a
finite linear combination of homomorphism counts; that is,
\begin{equation}\label{eq:zeta_G}
    \zeta(G) = \sum_F \gamma(F)\cdot \#\homs{F}{G},
\end{equation}
where the sum is over all (isomorphism types of) graphs, and
$\gamma(F) \ne 0$ only for finitely many graphs.
The complexity monotonicity principle then says that computing $\zeta$ is  at least as
hard as computing every single function $\#\homs{F}{\star}$ such that $\gamma(F)\ne 0$.
It is not hard to see that subgraph counts are graph motif parameters. That is,
$\#\subs{H}{\star}=\zeta$ for some $\zeta = \zeta_H$ that depends on $H$ through the
function $\gamma = \gamma_H$.
Thus, to obtain a lower bound for $\subsprob$, it suffices to find graphs $F$ with
$\gamma_H(F)\ne 0$ such that $\#\homs{F}{\star}$ is hard.
The same holds for $\#\indsubs{H}{\star}$.

Since the homomorphism counting problem \emph{on graphs} is fully
understood~\cite{DalmauJ04,Marx10}, this principle induces a general
framework for
analysing counting problems.
If a function from graphs to rationals can be cast as a linear
combination of homomorphism
counts, then the complexity analysis of the function reduces to the
purely combinatorial
task of understanding which graphs have a non-zero coefficient in the
homomorphism basis.
This framework has seen tremendous success, resolving various open problems in
parametrised and fine-grained counting complexity theory and establishing novel
classifications for a wide range of pattern counting
problems~\cite{CurticapeanDM17,DellRW19,RothSW20,BeraGLSS22,BLR23}.

For our complexity lower bounds we need a version of
the homomorphism basis for hypergraphs---which does not seem to have been formulated yet.
While researchers obtained such a result for the closely related \emph{relational structures}
in the context of database query evaluation~\cite{ChenM16,DellRW19},
unfortunately,
as the edges in hypergraphs are unordered compared to ordered relational tuples,
we are unable to directly apply the complexity
monotonicity principle for relational structures to hypergraphs.
As our first technical contribution, we thus introduce a
hypergraph version of the homomorphism basis.
In particular, we define \emph{hypergraph motif parameters}.
\begin{restatable}[Hypergraph Motif Parameter]{definition}{hyperparam}
    \dglabel{def:hypmotifparam}
    Let $\gamma$ be a function from hypergraphs to $\Q$ with finite support. The
    hypergraph motif parameter defined by $\gamma$, denoted by $\zeta_\gamma$, is the
    function from hypergraphs to $\Q$ that for every hypergraph $G$ satisfies
    \begin{equation}
        \zeta_\gamma(G) = \sum_F \gamma(F)\cdot \#\homs{F}{G}, \label{eq:zeta_gamma}
    \end{equation}
    where the sum is over all isomorphism types of hypergraphs.
\end{restatable}
\Cref{def:hypmotifparam} is the natural hypergraph generalization of the notion of graph
motif parameter introduced by~\cite{CurticapeanDM17}.

It is not hard to show that $\#\subs{H}{\star}$ is a hypergraph motif parameter, and the
same for $\#\indsubs{H}{\star}$.
Thus, $\#\subs{H}{\star}=\zeta_\gamma$ where $\gamma=\gamma_H$ depends on $H$.
In other words, for every hypergraph $H$ the ``subhypergraph count'' function
$\#\subs{H}{\star}$ is identified by the corresponding
function $\gamma_H$.
As a consequence, every family of hypergraphs $\scH$ can be associated with a family of
functions $\Gamma_\scH := \{\gamma_H : H \in \scH\}$.
Therefore, the input $H \in \scH$ can be encoded equivalently as a function $\gamma \in
\Gamma_{\scH}$.
This leads to the following general version of the sub-hypergraph counting problem, stated
in terms of hypergraph motif parameters.
\begin{problem}[HyperME]{Hypergraph Motif Evaluation (HyperME($\Gamma$))}
    \label{prob:hyperme}
    \PInput{A hypergraph $G$ and a finitely supported function
        $\gamma \in \Gamma$ from
    hypergraphs to $\Q$.}
    \POutput{\(\zeta_\gamma(G)\).}
    \PParameter{$\lVert\gamma\rVert$, the description length of \(\gamma\).}
\end{problem}
This (somewhat laborious) abstraction allows us to elegantly state our main technical tool---the hypergraph
generalization of the complexity monotonicity principle of~\cite{CurticapeanDM17}.
For every family $\Gamma$ of hypergraph motif parameters let $\scH_\Gamma
\coloneqq\bigcup_{\gamma \in \Gamma}\supp(\gamma)$.
\begin{theoremq}[Complexity Monotonicity for Hypergraph Motif Parameters]
    \dglabel{lem:monotonicity_intro}
    For every recursively enumerable family $\Gamma$ of hypergraph motif parameters,
    \begin{equation*}
        \hyperME(\Gamma) \fpteqlin \homsprob(\scH_\Gamma).\qedhere
    \end{equation*}
\end{theoremq}
Intuitively, \cref{lem:monotonicity_intro} says that evaluating \(\zeta_\gamma(G)\) is as
hard as evaluating~$\#\homs{H}{G}$ for every hypergraph $H$ in the support of $\gamma$.

In the next two sections we describe how we instantiate \cref{lem:monotonicity_intro} to
the hypergraph motif parameters $\#\subs{\scH}{\star}$ and $\#\indsubs{\scH}{\star}$ and
obtain the first crucial step in our lower bounds.

We conclude with a technical discussion on the proof of \cref{lem:monotonicity_intro}.
First, the reduction from $\hyperME(\Gamma)$ to
$\homsprob(\scH_\Gamma)$ is easy, as
by \cref{eq:zeta_gamma}, $\zeta_\gamma(G)$ reduces to evaluating
$\#\homs{H}{G}$ for
every $H \in \supp(\gamma)$.
The non-trivial direction is the reduction from
$\homsprob(\scH_\Gamma)$ to $\hyperME(\Gamma)$.
We rely on ``Dedekind Interpolation'', an abstract method for isolating terms in finite
linear combinations of semigroup\footnote{Recall that a semigroup consists of a set and an
associative operation on said set. In contrast to groups, the existence of neutral and
inverse elements is not necessary.}
homomorphisms introduced by Bressan, Lanzinger, and Roth~\cite{BLR23}.
Here we state a simplified version of the method, and defer the full
version to \cref{sub:homsprob}.

\begin{theoremq}[Efficiently isolating terms in finite linear combinations of a semigroup,
    ``Dedekind interpolation'', {\cite[simplified]{BLR23}}]\dglabel^{thm:dedekind_intro}
    Let $(\mathrm{G},\ast)$ be a semigroup.
    Let $(\varphi_i)_{i\in[k]}$ with $\varphi_i\colon \mathrm{G} \to
    \Q$ be pairwise distinct and non-zero semigroup homomorphisms from
    $(\mathrm{G},\ast)$ into $(\Q,\cdot)$, that is,
    $\varphi_i(g_1\ast g_2)= \varphi_i(g_1)\cdot \varphi_i(g_2)$ for
    all $i\in[k]$ and $g_1,g_2\in \mathrm{G}$.
    Let $\phi \colon \mathrm{G} \to \Q$ be any function of the form
    \begin{equation}\label{eq:dedekind_intro}
        \phi \colon g \mapsto \sum_{i=1}^k a_i \cdot \varphi_i(g),
    \end{equation}
    where the $a_i$ are rational numbers. Then there is an efficient
    algorithm $\hat{\mathbb{A}}$ which is equipped with oracle access
    to $\phi$ and which computes the coefficients $a_1,\dots, a_k$.
\end{theoremq}

In our application of Dedekind interpolation, the semigroup
homomorphisms are the
functions $G \mapsto \#\homs{H}{G}$.
To this end, we need to define an operation $\tensor$ on hypergraphs
that satisfies $\#\homs{H}{F\tensor G} = \#\homs{H}{F} \cdot\#\homs{H}{G}$.
To no surprise to readers familiar with this area, a version
of the tensor product turns out to be the right answer here.
For hypergraphs of bounded rank, this is fairly straightforward;
however, as we are dealing with hypergraphs of \emph{unbounded}
rank~$r$, we need to define the tensor $F\tensor G$ very carefully in
order to make sure that its
size is not too large; for example, we cannot just consider
\emph{all} size-$r$ subsets of each edge of $G$ when constructing the
tensor product, as this may make the product size explode with $|G|^r$,
which is not fixed-parameter tractable.

Nevertheless, we are able to define a notion of the tensor product
that is suitable for
Dedekind interpolation.
Then, for the reduction in \cref{lem:monotonicity_intro}, given a
hypergraph $G$ and an
oracle for a hypergraph motif parameter $\zeta_\gamma$ with support
$H_1,\dots, H_k$, we compute, for each hypergraph $F$,
\begin{equation}
    F \mapsto \sum_{i=1}^k \gamma(H_i)\cdot \#\homs{H_i}{F\tensor G} =
    \sum_{i=1}^k
    \left(\gamma(H_i)\cdot \#\homs{H_i}{F}\right) \cdot \#\homs{H_i}{G}.
\end{equation}
By setting $a_i \coloneqq \gamma(H_i)\cdot \#\homs{H_i}{G}$, we may
thus apply Dedekind
interpolation to obtain for each $i\in [k]$ the term $\#\homs{H_i}{G}
= a_i/\gamma(H_i)$.

\subsection{Classification for \texorpdfstring{$\subsprob$
    (\Cref{thm:classification_subsprob})}{Sub-Hypergraph
Counting}}\label{sub:overview_subsprob}
As a first step, we cast $\subsprob(\scH)$ as a hypergraph motif evaluation problem.
We show that, for each hypergraph $H$, there is a function $\gamma_H$ such that, for all $G$,
\begin{equation}\label{eq:subs_H=sum_homs_Q}
    \#\subs{H}{G} = \sum_{F \in \scQ(H)} \gamma_H(F) \cdot \#\homs{F}{G},
\end{equation}
where $\scQ(H)$ is the set of all quotients of $H$, see \cref{def:hypergraph_quotient}.
The proof is similar to the one for graphs~\cite{CurticapeanDM17}.
Thus, $\subsprob(\scH)=\hyperME(\Gamma_\scH)$, where $\Gamma_\scH = \{\gamma_H : H \in \scH\}$, see above.
Using our complexity monotonicity principle for hypergraphs,
\cref{lem:monotonicity_intro}, we then obtain the crucial equivalence
\[\homsprob(\scQ(\scH)) \fpteqlin \subsprob(\scH),\]
where $\scQ(\scH) = \cup_{H \in \scH} \,\scQ(H)$.
We then prove \cref{thm:classification_subsprob} as follows.

For the upper bounds we show that, if $\frcoindno(\scH) < \infty$, then $\homsprob(\scQ(\scH))$ is easy.
More precisely, we show that for every hypergraph family $\scH$ the
\emph{fractional hypertree width}~$\fhtw(\scQ(\scH))$ of~$\scQ(\scH)$
is bounded by $\frcoindno(\scH)+1$.
Thus, if $\frcoindno(\scH) < \infty$, then $\fhtw(\scQ(\scH)) \le
\frcoindno(\scH)+1 < \infty$.
We then use the well-known fact that $\homsprob(\scQ(\scH))$ is
solvable in time $f(H) \cdot \size{G}^{\fhtw(H)+O(1)}$,
see~\cite[Corollary~4.11]{GroheM14}.
This gives us the upper bounds of~\cref{thm:classification_subsprob}.

For the lower bounds, things are more involved.
Suppose that for some function $g(n)=o(\sqrt[4]{n})$ one can solve
$\subsprob(\scH)$ in time $f(H) \cdot \size{G}^{g(\frcoindno(H))}$;
we want to prove that ETH fails.
We do this by looking at the family $\WeakMinors(\scQ(\scH))$ of all
hypergraphs obtained by trimming a quotient of $\scH$.
First, we show that counting \emph{homomorphisms} of trimmed
quotients of $\scH$ reduces to $\subsprob(\scH)$, that is,
\[\homsprob(\WeakMinors(\scQ(\scH))) \fptredlin \subsprob(\scH).\]
This implies that one can solve $\homsprob(\WeakMinors(\scQ(\scH)))$
in time $f(H) \cdot \size{G}^{g(\frcoindno(H))}$, too.

We exploit this fact as follows.
Fix any unbounded function $c$ with~${c(n) = o(n)}$.
Set~${g(n) = o\big(\sqrt[4]{c(n)}\big)}$.
Define the following two hypergraph families
\begin{align*}
    \scH_{\aw} & \coloneqq \left\{ H \in \scH : \aw(H) \ge
    c(\frcoindno(H)) \right\} \quad\text{and}
    \\
    \scH_{\tw} & \coloneqq \left\{ H \in \scH : \aw(H) <
    c(\frcoindno(H)) \right\}.
\end{align*}
As $\frcoindno(\scH)=\infty$ and $\scH = \scH_{\aw} \, \cup \,
\scH_{\tw}$, we have $\frcoindno(\scH_{\aw}) = \infty$ or
$\frcoindno(\scH_{\tw})=\infty$.
We consider each case separately.

Suppose first $\frcoindno(\scH_{\aw}) = \infty$.
Since $c$ is unbounded, $\aw(\scH_{\aw})$ is also unbounded.
Moreover, one can easily see that our assumptions imply that one can solve
$\homsprob(\scH_{\aw})$ in time $f(H) \cdot
\size{G}^{o\big(\sqrt[4]{\aw(H)}\big)}$.
A lower bound by Marx then implies that ETH fails, see
\cref{lem:unbounded_aw_LB}.

Suppose now $\frcoindno(\scH_{\tw})=\infty$.
Then, by construction of $\scH_{\tw}$, the ratio
${\frcoindno(H)}/{\aw(H)}$ diverges over~${H \in \scH_{\tw}}$.
In this case we prove a structural result, \cref{lem:hat_H_subs},
that allows us to obtain an infinite subfamily~${\hat \scH_{\tw}
\subseteq \WeakMinors(\scQ(\scH)})$ of bounded rank and unbounded treewidth.
By a careful argument involving the choice of~$g$ as well as the relationship
between $\frcoindno$ and the treewidth of $\hat \scH_{\tw}$, we show
that our assumptions imply that one can solve $\homsprob(\hat \scH_{\tw})$
in time $f(\hat H) \cdot \size{G}^{o(\tw(\hat H) / \ln \tw(\hat H))}$.
Another lower bound of Marx again implies that ETH fails, see
\cref{lem:unbounded_tw_LB}.

\subsection{Classification for \texorpdfstring{$\indsubsprob$
    (\Cref{thm:classification_indsubsprob})}{Induced Sub-Hypergraph
Counting}}\label{sub:overview_indsubsprob}
Similarly to the case of $\subsprob(\scH)$, we begin by casting
$\indsubsprob(\scH)$ as a
hypergraph motif evaluation problem.
We show that for every hypergraph $H$ there is a finitely-supported function $\gamma_H$ such that
\begin{equation}\label{eq:indsub_expansion}
    \#\indsubs{H}{G} = \sum_{F\in \scQ(\scS(H))} \gamma_H(F) \cdot
    \#\homs{F}{G},
\end{equation}
where $\scS(H)$ is the set of edge-super-hypergraphs of $H$.\footnote{The set of
hypergraphs that can be obtained from $H$ by adding zero or more edges.}
Moreover, $\gamma_H(F) \ne 0$ for every $F \in \scS(H)$.
By using again our complexity monotonicity principle, \Cref{lem:monotonicity_intro}, we prove that
\[
    \homsprob(\scS(\scH)) \fptredlin \indsubsprob(\scH) \fptredlin
    \homsprob(\scQ(\scS(\scH))),
\]
where $\scS(\scH) = \cup_{H \in \scH} \,\scS(H).$
We then prove \cref{thm:classification_indsubsprob} as follows.

For the upper bounds, we show that $\fhtw(\scQ(\scS(\scH))) \le \fredgeco(\scH)$.
Similarly to~$\subsprob(\scH)$, we then use the fact that $\homsprob(\scQ(\scS(\scH)))$ is
solvable in time $f(H) \cdot \size{G}^{\fhtw(H)+O(1)}$.
This gives us the claimed running time.
As $\indsubsprob(\scH) \fptredlin \homsprob(\scQ(\scS(\scH)))$, the same running time
applies to $\indsubsprob(\scH)$ as well.

The lower bounds are again the harder part.
We follow the same strategy of $\subsprob(\scH)$, but with the
relevant parameter being $\fredgeco$ instead of $\frcoindno$.
Assume indeed $\scH$ is recursively enumerable with
$\fredgeco(\scH)=\infty$, and yet for some function
$g(n)=o(\sqrt[4]{n})$ one can solve $\indsubsprob(\scH)$ in time
$f(H) \cdot \size{G}^{g(\fredgeco(H))}$.
Mirroring the $\subsprob(\scH)$ case, we choose a suitable function
$c$ and define two hypergraph families
\begin{align*}
    \scH_{\aw} & \coloneqq \left\{ H \in \scH : \aw(H) \ge
    c(\fredgeco(H)) \right\}\quad\text{and}
    \\
        \scH_{\tw} & \coloneqq \left\{ H \in \scH : \aw(H) < c(\fredgeco(H)) \right\}.
\end{align*}
Since $\scH = \scH_{\aw} \, \cup \, \scH_{\tw}$, we have $\fredgeco(\scH_{\aw}) = \infty$ or
$\fredgeco(\scH_{\tw})=\infty$.
As we did for $\subsprob(\scH)$, and using the fact that $\homsprob(\scS(\scH))
\fptredlin \indsubsprob(\scH)$ stated above, we then show that the assumed
algorithm solves~$\homsprob(\scH_{\aw})$ or~$\homsprob(\scH_{\tw})$
in a running time that violates the conditional lower bounds of
\cref{lem:unbounded_aw_LB} or \cref{lem:unbounded_tw_LB}, implying
that ETH fails as well.

%% file: s3_prelims.tex
\section{Preliminaries}\label{sec:prelim}

\subparagraph*{Functions, sets and Partitions.}
Given sets $A$, $B$, and $C$, a function $f:A \times B \to C$, and an
element $a\in A$, we
write $f(a,\star):B \to C$ for the function $b \mapsto f(a,b)$.
Given a positive integer $k$, we set $[k] \coloneqq \{1,\dots,k\}$.
For a finite set \(A\), we use \(|A|\) and \(\#A\) to denote the
cardinality of \(A\).
A \emph{partition} of $A$ is a set $\tau=\{B_1,\dots,B_\ell\}$ of
non-empty subsets,
called \emph{blocks}, of $A$ such that $A=B_1 \cup \dots \cup B_\ell$
and $B_i \cap B_j
=\emptyset$ for all $i\neq j$.

\subparagraph*{Hypergraphs and morphisms between hypergraphs.}
A \emph{hypergraph} is a pair $H=(V,E)$, where $V=V(H)$ is a finite set and $E =
E(H)\subseteq 2^V \setminus\{\emptyset\}$ is a set of non-empty
subsets of $V$.
We write~$|H| \coloneqq |V(H)|$ and $\lVert H \rVert \coloneqq
|H|+\sum_{e \in E(H)}|e|$.
The \emph{rank} of $H$ is defined as $\rank(H)\coloneqq\max_{e\in E(H)} |e|$.
In this work we use capital letters \(F, G, H, \dots\) to denote
hypergraphs; we use script letters $\scF,\scG,\scH,\dots$ to denote families of hypergraphs.
Let $H=(V,E)$ be a hypergraph and let $v \in V$.
The hypergraph obtained from $H$ by \emph{deleting} $v$ has vertex
set $V \setminus \{v\}$
and edge set $\{e \in E : v \notin e\}$.
The hypergraph $\setminusnotrim{H}{X}$ is obtained by deleting all vertices in $X \subseteq V$.
Formally, its vertex set is $V\setminus X$ and its edge set is $\{e \in E : e \cap X = \emptyset\}$.
We also define another deletion operation:
the hypergraph $\setminustrim{H}{X}$ has vertex set $\overline{X} = V \setminus X$ and
edge set~$\{e \cap\overline{X} : e \in E\}\setminus \{\emptyset\}$.
Equivalently, $H \setminus X = \trim{H}{\overline{X}}$, see below.

A \emph{homomorphism} from a hypergraph $H$ to a hypergraph $G$ is a mapping
$\varphi\colon V(H)\to V(G)$ such that~$\varphi(e)\in E(G)$ holds for
all $e\in E(H)$.
We write $\homs{H}{G}$ for the set of all homomorphisms from $H$ to $G$.
An \emph{embedding} from $H$ to $G$ is an injective homomorphism from
$H$ to $G$.
A \emph{strong embedding} is an embedding $\varphi$ which additionally satisfies $S \in
E(H)\Leftrightarrow\varphi(S)\in E(G)$ for every subset $S\subseteq V(H)$.
We write $\embs{H}{G}$ for the set of all embeddings from $H$ to $G$,
and we write
$\strembs{H}{G}$ for the set of all strong embeddings from $H$ to $G$.
An \emph{isomorphism} from \(H\) to \(G\) is a bijective strong
embedding; $H$ and $G$ are
\emph{isomorphic} if there is an isomorphism from~$H$ to~$G$.
In this case, we also write $H\cong G$.
An \emph{automorphism} of \(H\) is an isomorphism from $H$ to itself.
We write $\auts{H}$ for the set of all automorphisms of $H$.

A \emph{sub-hypergraph} of $G$ is any hypergraph $G'$ with
$V(G')\subseteq V(G)$ and $E(G') \subseteq E(G)$.
If $V(G')=V(G)$ then $G'$ is an \emph{edge}-sub-hypergraph.
For a vertex subset $X \subseteq V(G)$, the sub-hypergraph
\emph{induced by $X$} in $G$ is the hypergraph $G[X]$ with $V(G[X])=X$ and
$E(G[X])=\{e \in E(G) \mid e\subseteq X\}$.
We denote by~$\subs{H}{G}$ the set of all sub-hypergraphs of $G$ that
are isomorphic to $H$, and by $\indsubs{H}{G}$ the set of all induced sub-hypergraphs
that are isomorphic to $H$.
One readily verifies the following well-known equalities that relate
sub-hypergraphs and embeddings.
\begin{align}
    \#\embs{H}{G}    & = \#\auts{H}\cdot \#\subs{H}{G}\label{eq:embs=auts/subs} \\
    \#\strembs{H}{G} & = \#\auts{H}\cdot
    \#\indsubs{H}{G}\label{eq:strembs=auts/indsubs}
\end{align}

\subparagraph*{Trimmed homomorphisms and trimmed sub-hypergraphs.}
A \emph{trimmed homomorphism} from a hypergraph $H$ to a hypergraph $G$ is a map
$\varphi\colon V(H)\to V(G)$ such that if $e \in E(H)$ then there
exists $e' \in E(G)$
satisfying $\varphi(e) = e' \cap \img(\varphi)$.
We write $\trimhoms{H}{G}$ for the set of all such maps.
A \emph{trimmed embedding} from $H$ to $G$ is a trimmed homomorphism
from $H$ to $G$ that
is injective.
A \emph{strong trimmed embedding} is a trimmed embedding $\varphi$
where $e \in E(H)$ if
and only if there exists $e' \in E(G)$ satisfying $\varphi(e) = e'
\cap \img(\varphi)$.
We write $\trimembs{H}{G}$ and $\strtrimembs{H}{G}$ for the set of
all trimmed and strong
trimmed embeddings of $H$ in $G$.

A \emph{trimmed sub-hypergraph} of $G$ is any hypergraph $G'$ with
$V(G') \subseteq V(G)$
such that for every $e' \in E(G')$ there exists $e \in E(G)$ such that
$e' = e \cap V(G')$.
The trimmed sub-hypergraph of $G$ induced by $X \subseteq V(G)$ is
the hypergraph
$\trimsub{G}{X} \coloneqq \big(X, \trim{E}{X}\big)$, where
$\trim{E}{X} = \{e \cap X : e \in E(G), e \cap X \ne \emptyset\}$.
We denote by $\trimsubs{H}{G}$ and $\indtrimsubs{H}{G}$ the sets of trimmed
sub-hypergraphs and induced trimmed sub-hypergraphs of $G$ isomorphic to $H$.
One may adapt \cref{eq:embs=auts/subs,eq:strembs=auts/indsubs} to
obtain the following
equalities, see \cref{lem:aut-omatic_lemma}.
\begin{align}
    \#\trimembs{H}{G}    & = \#\auts{H}\cdot
    \#\trimsubs{H}{G}\label{eq:trimembs=auts/trimsubs} \\
    \#\strtrimembs{H}{G} & = \#\auts{H}\cdot
    \#\indtrimsubs{H}{G}\label{eq:strtrimembs=auts/trimindsubs}
\end{align}

\subparagraph*{Parametrised problems, fixed-parameter tractability,
and fine-grained complexity theory.}
We provide a brief introduction to parametrised complexity theory and
refer the reader to
one of the standard textbooks~\cite{FlumG06,CyganFKLMPPS15,DowneyF13}
for a comprehensive treatment.
A \emph{parametrised (counting) problem} is a pair of a function
$P\colon \{0,1\}^\ast \to
\mathbb{Q}$ and a computable parametrisation $\kappa\colon
\{0,1\}^\ast \to \mathbb{N}$. A
parametrised problem $(P,\kappa)$ is called \emph{fixed-parameter
tractable (FPT)} if
there is a computable function $f$ and a deterministic algorithm
$\mathbb{A}$ such that,
on input $x$, the algorithm $\mathbb{A}$ computes $P(x)$ in time $f(\kappa(x))\cdot
|x|^{\Oh(1)}$.
$\mathbb{A}$ is called an \emph{FPT algorithm} for $(P,\kappa)$.
A \emph{parametrised Turing-reduction} from $(P_1,\kappa_1)$ to $(P_2,\kappa_2)$ is an FPT
algorithm for~$(P_1,\kappa_1)$ with oracle access to $P_2$ such that there is a computable
function $g$ such that, on input~$x$, each oracle query~$y$ satisfies $\kappa_2(y)\leq
g(\kappa_1(x))$.
We write $(P_1,\kappa_1)\fptred (P_2,\kappa_2)$ if a parametrised
Turing-reduction from
$(P_1,\kappa_1)$ to $(P_2,\kappa_2)$ exists, and we write $(P_1,\kappa_1)\fpteq
(P_2,\kappa_2)$ if $(P_1,\kappa_1)\fptred (P_2,\kappa_2)$ and~$(P_2,\kappa_2)\fptred(P_1,\kappa_1)$.
If $(P_1,\kappa_1)\fptred (P_2,\kappa_2)$ and moreover the reduction
runs in FPT \emph{linear time}, that is in time $f(\kappa(x)) \cdot
|x|$, then we write $(P_1,\kappa_1)\fptredlin (P_2,\kappa_2)$.
Note that if $(P_1,\kappa_1)\fptredlin (P_2,\kappa_2)$ and~$(P_2,\kappa_2)\fptredlin (P_3,\kappa_3)$ then
$(P_1,\kappa_1)\fptredlin (P_3,\kappa_3)$ too.

The parametrised problem $\#\textsc{Clique}$ gets as input a graph $G$ and a
positive integer $k$, and outputs the number of $k$-cliques in $G$.
The parametrisation of
$\#\textsc{Clique}$ is given by $\kappa(G,k)=k$; we say that $k$ is
the \emph{parameter}
of the problem.

A parametrised problem $(P,\kappa)$ is $\#\W[1]$-\emph{hard} if
$\#\textsc{Clique}\fptred (P,\kappa)$, and it is $\#\W[1]$-\emph{complete}%
\footnote{For parametrised counting problems it is common to define
  completeness via
  parametrised Turing-reductions rather than via parametrised
  parsimonious reductions.
  In particular, this notion of $\#\W[1]$-completeness does not
  necessarily imply that the
  problem is contained in $\#\W[1]$ (see~\cite[Chapter 14]{FlumG01}
    for a structural
definition of $\#\W[1]$).}
if $\#\textsc{Clique}\fpteq (P,\kappa)$.
The standard hardness assumption in parametrised counting complexity theory
``$\mathrm{FPT}\neq \#\W[1]$'' asserts that $\#\W[1]$-hard problems are not
fixed-parameter tractable.
Moreover, ``$\mathrm{FPT}\neq \#\W[1]$'' is implied by the
\emph{Exponential Time Hypothesis}.
\begin{conjecture}[ETH;~\cite{ImpagliazzoP01}]
    $3$-$\textsc{SAT}$ cannot be solved in time $\mathsf{exp}(o(n))$,
    where $n$ is the number of variables in the input formula.
\end{conjecture}

The parametrised problem $\#\textsc{DomSet}$ gets as input a graph
$G$ and a positive integer $k$, and outputs the number of dominating
sets of size
$k$ in $G$. The parameter is $k$.
A parametrised problem $(P,\kappa)$ is $\#\W[2]$-\emph{hard} if
$\#\textsc{DomSet}\fptred (P,\kappa)$.
It is conjectured that $\#\W[2]$-hard problems are \emph{strictly} harder than
$\#\W[1]$-complete problems~\cite[Chapter 14]{FlumG06}.
The final lower bound assumption featured in this work is the
\emph{Strong Exponential
Time Hypothesis}.
\begin{conjecture}[SETH;~\cite{CalabroIP09}]
    For every $\delta>0$, there is a positive integer $k$ such that
    $k$-$\textsc{SAT}$
    cannot be solved in time $\Oh(2^{(1-\delta)n})$, where $n$ is the
    number of variables
    in the input formula.
\end{conjecture}

In this work, we are interested in the parametrised complexity of
counting substructures
in hypergraphs.
Formally, we define below the problems $\homsprob(\scH)$, $\subsprob(\scH)$ and
$\indsubsprob(\scH)$, where $\scH$ denotes any class of hypergraphs,
that is, $\scH$ is
not part of the problem input, but instead each class $\scH$ defines
a unique problem.

\noindent
\begin{minipage}{1\linewidth}%
    \begin{problem}-[$\#\textupsc{Hom}$]{$\#\textupsc{Hom}(\scH)$}
        \label{prob:homs}
        \PInput{A pair of hypergraphs $H, G$ with $H \in \scH$}
        \POutput{$\#\homs{H}{G}$}
        \PParameter{$|H|$}
    \end{problem}

    \begin{problem}+-=[$\#\textupsc{Sub}$]{$\#\textupsc{Sub}(\scH)$}
        \label{prob:sub}
        \PInput{A pair of hypergraphs $H, G$ with $H \in \scH$}
        \POutput{$\#\subs{H}{G}$}
        \PParameter{$|H|$}
    \end{problem}

    \begin{problem}+=[$\#\textupsc{IndSub}$]{$\#\textupsc{IndSub}(\scH)$}
        \label{prob:indsub}
        \PInput{A pair of hypergraphs $H, G$ with $H \in \scH$}
        \POutput{$\#\indsubs{H}{G}$}
        \PParameter{$|H|$}
    \end{problem}
\end{minipage}

\subparagraph*{Invariants and width measures.}\label{sub:invariants}
Our analysis makes use of several hypergraph invariants that we now define.
Note that, for the sole purpose of the definitions, we assume
hypergraphs do not have isolated vertices (vertices that do not
appear in any edge).
More precisely, if $H$ has isolated vertices, then the definition of
any invariant here below is meant on $H$ after adding a singleton
edge $\{v\}$ for every isolated $v$.
This is necessary to guarantee that the definitions of fractional
edge-covers and fractional independent sets admit a solution when
seen as linear programs, which we need to apply standard duality results.
We repeat that this is only for the purpose of the definitions; all
our algorithms work for arbitrary hypergraphs.

\begin{definition}[Edge-cover numbers]
    \dglabel{def:edgeco}
    \label{def:fredgeco}
    Let $H=(V,E)$ be a hypergraph and $X \subseteq V$.

    \begin{itemize}
        \item An edge set $E' \subseteq E$ is an edge cover of $X$ if $X\subseteq\cup_{e'\in E'}e'$.
            The \emph{edge-cover number of $X$ in $H$}, denoted by~$\rho_H(X)$, is the
            smallest cardinality of an edge cover of $X$.
            The edge-cover number of $H$ is $\rho(H)\coloneqq\rho_H(V)$.

        \item A function $\xi\colon E \rightarrow \Rp$ is a
            \emph{fractional edge cover of
            $X$} if $\sum_{e: v\in e} \xi(e) \geq 1$ for all $v\in X$.
            The fractional edge-cover number of $X$ in $H$, denoted by
            $\fredgeco_H(X)$, is the minimum of $\sum_{e\in E} \xi(e)$ over all fractional edge covers~$\xi$ of $X$.
            The fractional edge-cover number of $H$ is $\fredgeco(H) \coloneqq\fredgeco_H(V)$.\qedhere
    \end{itemize}
\end{definition}

\begin{definition}[Independence numbers]
    \dglabel{def:indno}\label{def:frindno}
    Let $H=(V,E)$ be a hypergraph.
    \begin{itemize}
        \item  A subset $X \subseteq V$ is \emph{independent} if $|e \cap X|\le 1$ for all
            $e\in E$.\footnote{This notion of independence is also called strongly independent
            by some authors.}
            The \emph{independence number of $H$}, denoted by~$\alpha(H)$, is the largest
            cardinality of an independent set.

        \item A function $\eta: V \rightarrow \Rp$ is a \emph{fractional
            independent set}
            if
            $\sum_{v\in e} \eta(v) \leq 1$ for all $e\in E$, and its \emph{weight} is
            $\sum_{v
            \in V} \eta(v)$.
            The fractional independence number of $H$, denoted by
            $\alpha^\ast(H)$, is the
            maximum weight of a fractional independent set.
            \qedhere
    \end{itemize}
\end{definition}

\begin{definition}[Fractional co-independent edge-cover numbers]
    \dglabel{def:coindno}\label{def:frcoindno}
    Let $H=(V,E)$ be a hypergraph.
    A set~$X \subseteq V$ is \emph{co-independent} if its complement $\overline{X} = V \setminus X$ is independent.
    The \emph{fractional co-independent edge-cover number} of a
    hypergraph $H=(V,E)$ is \[\sigma^\ast(H) \coloneqq \min
        \left\{\fredgecoS{H}(X) \,: X\text{ is co-independent in }H
    \right\}.\qedhere\]
\end{definition}

Next, we recall the standard notion of tree decomposition of a
hypergraph, and the several
related notions of \emph{width} of a hypergraph, which plays a key
role in our complexity
dichotomies.
\begin{definition}[Tree decomposition]
    \dglabel{def:treewidth}
    A \emph{tree decomposition} of a hypergraph $H$ is a pair
    $(T,\scB)$, where $T$ is a
    tree and $\scB=\{B_t\}_{t\in V(T)}$ is a family of subsets of
    $V(H)$ indexed by
    $V(T)$, such that the following hold.
    \begin{enumerate}
        \item $\bigcup_{t\in V(T)} B_t = V(H)$.
        \item For every $e\in E(H)$ there exists $t \in V(T)$ such that
            $e\subseteq B_t$.
        \item For every $v\in V(H)$ the subgraph $T[\{t~|~v\in B_t\}]$ is
            connected. \qedhere
    \end{enumerate}
\end{definition}

\begin{definition}[$f$-width]
    \dglabel{7-21-1}[def:treewidth]
    Let $H$ be a hypergraph and let $f : 2^{V(H)}\rightarrow \mathbb{R}_{\ge 0}$ be a
    function.
    The $f$-\emph{width} of a tree decomposition $(T,\scB)$ of $H$ is
    \[\fwidth{f}(T,\scB)
    \coloneqq\max_{t\in V(T)}f(B_t).\]
    The $f$-width of $H$ is the minimum $f$-width of any tree
    decomposition of $H$.
\end{definition}

\begin{definition}[Treewidth, fractional hypertree width,
    adaptive width]
    \dglabel{def:widthmeasures}[7-21-1]
    Let $H$ be a hypergraph.
    \begin{itemize}
        \item The treewidth of $H$, denoted by $\tw(H)$, is $\fwidth{f}(H)$ where $f$
            is given by $f(X) \coloneqq\max(0,|X|-1)$.
        \item The fractional hypertree width of $H$ is $\fhtw(H) \coloneq
            \fwidth{\fredgecoS{H}}(H)$.
        \item The adaptive width of $H$, denoted by $\aw(H)$, is the supremum of
            $\fwidth{f}(H)$ over all $f$ that are fractional independent sets of
            $H$.\qedhere
    \end{itemize}
\end{definition}
The Gaifman graph of a hypergraph $H=(V,E)$ is the simple graph
$\gaifman(H)$ with vertex
set $V$ which contains an edge $\{u,v\}$ for all vertices $u$ and $v$
that are adjacent in $H$.
The equality $\tw(\gaifman(H))=\tw(H)$ is well-known (e.g., see
\cite[Section~2]{DBLP:conf/cp/FichteHLS18}).

For a function \(w\) from hypergraphs to $\R_{\ge 0}$ and a
hypergraph family $\scH$, we
set \(w(\scH) \coloneqq \sup_{H \in \scH} w(H)\).
We say that \(w\) is \emph{bounded} for \(\scH\), denoted by
$w(\scH)<\infty$, if there is
a $c \ge 0$ such that $w(H) \le c$ for all~$H \in \scH$.
Otherwise, we say that \(w\) is \emph{unbounded} for \(\scH\), which
we denote by $w(\scH)=\infty$.

We often use the following relationships (e.g., see
\cite[Proposition~3.7f]{Marx13} and \cite[Section~2]{DBLP:conf/cp/FichteHLS18}).
\begin{lemmaq}
    \dglabel{lem:width_measures}(Every hypergraph $H$ satisfies $\aw(H)
    \leq \fhtw(H)
    \leq\tw(H)$)
    If $H$ is a hypergraph, then $\aw(H) \leq \fhtw(H) \leq\tw(H)$.
\end{lemmaq}

\begin{lemmaq}[LP duality, {\cite{fractionalGraphTheory}}]
    \dglabel{lem:LP-duality}(Every hypergraph $H$ satisfies
    $\frindno(H)=\fredgeco(H)$, LP duality, {\cite{fractionalGraphTheory}})
    If $H$ is a hypergraph, then $\frindno(H)=\fredgeco(H)$.
\end{lemmaq}

\begin{lemmaq}[{\cite[Theorem 17]{BLR23}}]
    \dglabel{lem:int_gap_aw}(Every hypergraph $H$ satisfies
    $\indno(H)\geq {1}/{2} +
    {\frindno(H)}/{4\, \aw(H)}$, \cite[Theorem 17]{BLR23})
    If $H$ is a hypergraph, then $\indno(H)\geq {1}/{2} +
    {\frindno(H)}/{4\, \aw(H)}$.
\end{lemmaq}

Note that the presented width measures are not asymptotically equivalent.
There exists a class of hypergraphs \(\scH_1\) such that
\(\aw(\scH_1)<\infty\) and \(\fhtw(\scH_1)=\infty\), see
\cite[Corollaries~22 and~25]{DBLP:journals/mst/Marx11}.
Similarly, there is a class of hypergraphs \(\scH_2\) such that
\(\fhtw(\scH_2)<\infty\) and \(\tw(\scH_2)=\infty\), see
\cite[Example~2.1]{GroheM14}.
However, in the case of bounded rank, all these notions are
asymptotically equivalent \cite[Observation~34]{FockeGRZ25}, that is, for a class of
hypergraphs \(\scH_3\) the property
\(\rank(\scH_3)<\infty\) implies that \(\aw(\scH_3)<\infty\) if and
only if \(\tw(\scH_3)<\infty\).

We recall two conditional lower bounds from \cite{Marx13}, which have been slightly adapted.%
\footnote{In \cite{Marx13}, the first bound uses the \emph{submodular
    width} $\subw(H)$ in place of $\aw(H)$; however, $\aw(H) \le
    \subw(H)$ as shown in the same work. Moreover, in \cite{Marx13} both
    bounds are for the decision version, so they hold in particular for
the counting version.}
\begin{lemmaq}[{see \cite[Theorem 7.1]{Marx13}}]
    \dglabel{lem:unbounded_aw_LB}
    Let $\scH$ be a recursively enumerable family of hypergraphs with~$\aw(\scH)=\infty$.
    If there exist an algorithm $\mathrm{A}$ and a function $f$ such that
    $\mathrm{A}$ solves $\homsprob(\scH)$ in time~${f(H) \cdot
    \size{G}^{o\left(\sqrt[4]{\aw(H)}\right)}}$, then ETH fails.
\end{lemmaq}

\begin{lemmaq}[{see \cite[Theorem 1.2]{Marx13}}]
    \dglabel{lem:unbounded_tw_LB}
    Let $\scH$ be a recursively enumerable family of hypergraphs with~$\tw(\scH)=\infty$ and $\rank(\scH)<\infty$.
    If there exist an algorithm $\mathrm{A}$ and a function $f$ such that
    $\mathrm{A}$ solves~$\homsprob(\scH)$ in time $f(H) \cdot
    \size{G}^{o\left(\tw(H)/\log\tw(H)\right)}$, then ETH fails.
\end{lemmaq}

Finally, we need the following monotonicity property of adaptive width under trimmings.
\begin{lemma}
    \dglabel{lem:aw_monotone}(Trimmed induced sub-hypergraphs preserve
    adaptive width, $\aw(\trim{H}{X}) \le \aw(H)$)
    Let $H=(V,E)$ be a hypergraph.
    For every $X \subseteq V$ it holds that $\aw(\trim{H}{X}) \le \aw(H)$.
\end{lemma}
\begin{proof}
    Let $\scF$ be the family of all fractional independent sets over
    $H$, and $\scF_{X}^0
    \subseteq \scF$ the subfamily of those functions with value zero
    over $V\setminus X$.
    Obviously $\scF\supseteq \scF_{X}^0$, thus
    \begin{align}
        \aw(H) = \sup_{f \in \scF} \left(\fwidth{f}(H)\right) \ge \sup_{f
        \in \scF_{X}^0}
        \left(\fwidth{f}(H)\right).
    \end{align}
    Now, let $\scF_{X}$ be the family of all fractional independent sets
    of $\trim{H}{X}$, and
    for every $f_X \in \scF_{X}$ let $f_X^0 \in \scF_X^0$ be its
    extension to $V$ that
    assigns $0$ to $V\setminus X$.
    Observe that
    \begin{align}
        \fwidth{f_X}(\trim{H}{X}) \le \fwidth{f_X^0}(H).
    \end{align}
    Consider indeed any tree decomposition $(T,\scB)$ of $H$.
    The pair $(T,\scB_{X})$
    obtained by intersecting every bag of $\scB$ with $X$ is then a
    tree decomposition of
    $\trim{H}{X}$, and clearly $\fwidth{f_X}((T,\scB_X)) \le
    \fwidth{f_X^0}((T,\scB))$.
    Taking the infimum over all tree decompositions then yields
    $\fwidth{f_X}(\trim{H}{X})
    \le \fwidth{f_X^0}(H)$, as claimed.
    Finally, by taking the supremum over all $f_X \in \scF_{X}$, we
    conclude that
    \begin{align}
        \aw(\trim{H}{X}) = \sup_{f_X \in \scF_X}\left(\fwidth{f_X}(\trim{H}{X})\right)
        \le \sup_{f_X \in \scF_X} \left(\fwidth{f_X^0}(H)\right)
        \le \sup_{f\in \scF_X^0} \left(\fwidth{f}(H)\right).
    \end{align}
    Together with the inequality above, this proves that
    $\aw(\trim{H}{X}) \le \aw(H)$.
\end{proof}

%% file: s4_hyper_basis.tex
\section{The hypergraph homomorphism basis}
\label{sub:homsprob}
The goal of this section is to relate the complexity of $\homsprob$ and of
$\subsprob$/$\indsubsprob$; this is one of the key steps in proving our main theorems (see
\cref{sec:sub,sec:indsub}).
To this end we prove a hypergraph version of the so-called \emph{complexity monotonicity}
principle, which requires several technical ingredients.
First, \cref{sub:homs_embs_strembs} introduces \emph{hypergraph motif parameters}, which
are linear combinations of hypergraph homomorphisms counts, and shows that this captures
both sub-hypergraph and induced sub-hypergraph counts.
Next, \cref{sub:homs_and_tensors} introduces a notion of \emph{tensor product} between
hypergraphs that satisfies certain good properties with respect to fixed-parameter
tractability; this is needed afterwards.
Then, \cref{sub:hyper_motifs} develops a hypergraph version of the so-called ``Dedekind
interpolation'' technique of~\cite{BLR23}; this yields in particular
\cref{lem:monotonicity_intro} and, in turn, the desired complexity inequalities between
$\homsprob$ and of $\subsprob$/$\indsubsprob$.
Finally, \cref{sub:homs_and_colhoms} completes the picture by proving complexity results
for the \emph{colour-prescribed} and \emph{colourful} variants of hypergraph homomorphism
counting, both of which are also crucial to our bounds of \cref{sec:sub,sec:indsub}.

\subsection{Hypergraph motif parameters}\label{sub:homs_embs_strembs}
The starting point of our complexity results is the extension to hypergraphs of the notion
of \emph{graph motif parameter} introduced by \cite{CurticapeanDM17}. Essentially, a
hypergraph motif parameter is any function that can be written as a finite linear
combination of hypergraph homomorphism counts.

\hyperparam*
The main result of this section reads as follows.
\begin{lemmaq}\label{lem:subs_and_indsubs_are_motifs}
    For every hypergraph $H$, the functions $\subs{H}{\star}$ and $\indsubs{H}{\star}$ are
    hypergraph motif parameters.
    \begin{enumerate}
        \item There exists a rational function $\gamma_H$ with finite support such that
            \begin{align}
                \label{eq:subs_to_homcomb_reminder}
                \#\subs{H}{\star} = \sum_F \gamma_H(F) \cdot \#\homs{F}{\star}
            \end{align}
            and $\gamma_H(F)\neq 0$ if and only if $F \in \scQ(H)$.
        \item There exists a rational function $\gamma_H'$ with finite support such that
            \begin{align}
                \label{eq:indsubs_to_homcomb_reminder}
                \#\indsubs{H}{\star} = \sum_F \gamma_H'(F) \cdot \#\homs{F}{\star}
            \end{align}
            and $\gamma_H'(F) \ne 0$ if $F \in \scS(H)$.
    \end{enumerate}
    As a consequence, for every hypergraph family we have $\subsprob(\scH) =
    \hyperME(\Gamma_\scH)$, where $\Gamma_\scH=\cup_{H \in \scH}\{\gamma_H\}$, and we have
    $\indsubsprob(\scH) = \hyperME(\Gamma'_\scH)$, where $\Gamma'_\scH=\cup_{H \in
    \scH}\{\gamma'_H\}$.
\end{lemmaq}
The first claim of the lemma follows from \cref{eq:embs=auts/subs} and
\cref{lem:embs_to_homs} below for $\#\subs{H}{\star}$, and from
\cref{eq:strembs=auts/indsubs} and \cref{lem:strembs_to_homs} below for
$\#\indsubs{H}{\star}$. The second claim follows from the definitions of the problems.
The rest of this subsection is devoted to prove the mentioned results.

In graphs, homomorphisms counts, embedding counts, and strong embedding counts can be
represented as linear combinations of each other. We need to derive analogous results for
hypergraphs.
We first define the quotient of a hypergraph.

\begin{definition}[The quotient $H/\tau$ of a hypergraph \(H\) with respect to a
    partition \(\tau\) of the vertices \(H\); the set \(\scQ(H)\) of
    all quotients of \(H\)]
    \dglabel{def:hypergraph_quotient}
    Let $H=(V,E)$ be a hypergraph and let $\tau=\{V_1,\ldots,V_\ell\}$
    be a partition of
    $V$.
    For any set $X \subseteq V$, the \emph{quotient of $X$ under
    $\tau$} is the set
    $X/\tau$ defined via
    \(X/\tau \coloneqq \{V_i \in \tau\colon V_i \cap X \ne \emptyset\}.\)

    The \emph{quotient of~$H$ with respect to $\tau$} is the hypergraph
    $H/\tau$ defined
    via $H/\tau\coloneqq \big(V/\tau,\{e/\tau\colon e \in E\}\big)$.
    We write $\scQ(H)$ for the set of all quotients of $H$, that is,
    \(
    \scQ(H) \coloneqq \big\{H/\tau \colon \tau \text{ partition of }V\big\}.
    \)
\end{definition}

Observe that $V/\tau = \tau$ holds for all partitions $\tau$ of~$V$,
and that we may have
$e/\tau=e'/\tau$ even if the edges $e,e' \in E$ are distinct.
Using the Möbius transform over the partition lattice, we now prove
that embedding counts
can be represented as a linear combination of homomorphism counts,
which is analogous to
the case of graphs (see \cite[(5.18)]{Lovasz12}).
\begin{restatable}[Writing hypergraph embedding counts as
    hypergraph homomorphism
    counts]{lemma}{embshom}
    \dglabel{lem:embs_to_homs}[def:hypergraph_quotient]
    For all hypergraphs $H$ there is a function $\gamma: \scQ(H) \to
    \mathbb{Q}$ such
    that for all hypergraphs $G$ we have
    \begin{align*}
        \#\embs{H}{G} = \sum_{F\in\scQ(H)} \gamma(F) \cdot \#\homs{F}{G},
    \end{align*}
    where $\gamma(F) \ne 0$ holds for all terms in the sum.
\end{restatable}
\begin{proof}
    The proof is almost identical to the one for graphs. Observe that
    \begin{align*}
        \#\homs{H}{G} = \sum_{\tau} \#\embs{H/\tau}{G},
    \end{align*}
    where $\tau$ ranges over all partitions of $V(H)$.
    By Möbius inversion over the lattice of partitions of $V(H)$, then,
    \begin{align*}
        \#\embs{H}{G} = \sum_{\tau} \mu(\tau) \cdot \#\homs{H/\tau}{G},
    \end{align*}
    where $\mu(\tau)$ is the Möbius function of the partition lattice,
    see~\cite[Appendix
    A]{Lovasz12}.
    For each $F \in \scQ(H)$, let $\gamma(F) =
    \sum_{\tau:H/\tau \cong F}
    \mu(\tau)$,
    We obtain
    \begin{align*}
        \#\embs{H}{G} = \sum_{F \in \scQ(H)} \gamma(F) \cdot \#\homs{F}{G}.
    \end{align*}
    To conclude, for any given $F \in \scQ(H)$ note that $\mu(\tau)$ is
    either strictly
    positive or strictly negative for all $\tau$ such that $H/\tau \cong F$---see
    again~\cite{Lovasz12}.
    This implies $\gamma(F) \ne 0$ for all $F \in \scQ(H)$, as claimed.
\end{proof}

Using an inclusion--exclusion argument, we now prove that strong
embedding counts can be
represented as a linear combination of embedding counts, which is
analogous to the case of
graphs (see \cite[(5.17)]{Lovasz12}).
\begin{restatable}[Writing hypergraph strong embedding counts as hypergraph
    homomorphism counts]{lemma}{strembshom}
    \dglabel{lem:strembs_to_homs}[lem:embs_to_homs]
    Let $\scS(H)$ be the set of all edge-super-hypergraphs of $H$.
    For every hypergraph $H$ there is a function $\gamma: \scQ(\scS(H)) \to \mathbb{Q}$
    such that, for all hypergraphs $G$, we have
    \begin{align*}
        \#\strembs{H}{G} = \sum_{F \in \scQ(\scS(H))} \gamma(F) \cdot \#\homs{F}{G},
    \end{align*}
    where $\gamma(F) \ne 0$ holds for every $F \in \scS(H)$.
\end{restatable}
\begin{proof}
    The proof is almost identical to the one for graphs
    (\cite{CurticapeanDM17}, see
    also~\cite[Theorem 12]{RothS20} for a detailed presentation of the
    transformation). We provide the details here only for reasons of
    self-containment.

    First, set $\overline{E}(H) \coloneqq (2^{V(H)} \setminus
    \{\emptyset\}) \setminus
    E(H)$; this is the set of non-edges of $H$.
    For any $A \subseteq \overline{E}(H)$ set~${H^A \coloneqq (V(H),E(H) \cup A})$.
    Observe that
    \begin{align*}
        \#\strembs{H}{G} = \#\embs{H}{G} - \#\left(
            \bigcup_{e \in \overline{E}(H)} \{\varphi \in
        \embs{H}{G}\mid\varphi(e)\in E(G) \}\right).
    \end{align*}
    By inclusion-exclusion, we obtain
    \begin{align*}
        \#\strembs{H}{G} & = \#\embs{H}{G} -
        \sum_{\emptyset \neq A \subseteq \overline{E}(H)} (-1)^{|A|+1}\cdot
        \#\left( \bigcap_{e \in A} \{\varphi \in
        \embs{H}{G}\mid\varphi(e)\in E(G) \}\right)
        \\
        ~                & =  \#\embs{H}{G} -
        \sum_{\emptyset \neq A \subseteq \overline{E}(H)}
        (-1)^{|A|+1}\cdot \#\embs{H^A}{G}       \\
        ~                & = \sum_{A \subseteq \overline{E}(H)}
        (-1)^{|A|}\cdot \#\embs{H^A}{G}.
    \end{align*}
    Next, by applying \cref{lem:embs_to_homs} we obtain
    \begin{align*}
        \#\strembs{H}{G} = \sum_{A \subseteq \overline{E}(H)} (-1)^{|A|}
        \sum_{\tau} \mu(\tau)
        \cdot \#\homs{H^A/\tau}{G},
    \end{align*}
    where $\tau$ ranges over all partitions of $V(H^A)=V(H)$.
    If for every $F \in \scQ(\scS(H))$ we collect all terms~${H^A/\tau\cong F}$, we finally obtain
    \begin{align*}
        \#\strembs{H}{G} = \sum_{F \in \scQ(\scS(H))} \gamma(F) \cdot \#\homs{F}{G},
    \end{align*}
    where
    \begin{align}
        \gamma(F) = \sum_{A \subseteq \overline{E}(H)} (-1)^{|A|}  \cdot
        \sum_{\tau \,:\,
        H^A\!/\!\tau = F} \mu(\tau) .\label{eq:strembs_cF}
    \end{align}

    It remains to prove that $\gamma(F) \ne 0$ when $F \in \scS(H)$,
    that is, when $F\cong
    H^A$ for some $A \subseteq \overline{E}(H)$.
    Clearly, in this case the inner sum in \cref{eq:strembs_cF} is only
    over the finest partition $\tau$ of $V(H)$, which satisfies~${\mu(\tau)=1}$, see~\cite[Appendix A]{Lovasz12}.
    Second, we have $|A|=|E(F)|-|E(H)|$ independently of $A$.
    We conclude that
    \begin{align*}
        \gamma(F) \;= \sum_{A \subseteq\overline{E}(H) \,:\, H^A\cong F}
        (-1)^{|E(F)|-|E(H)|} \cdot 1
        \ne 0,
    \end{align*}
    which concludes the proof.
\end{proof}

\subsection{A hypergraph tensor product}\label{sub:homs_and_tensors}
Before moving forward, we need a notion of tensor product of two hypergraphs.
This is the \emph{categorical product} of~\cite[Definition
15.1]{hellmuth_survey_2012}, and is the hypergraph equivalent of the standard tensor
product of graphs.
We need this product in order to obtain our complexity monotonicity result (the standard
graph tensor product is used in the same way in~\cite{CurticapeanDM17}).
For a formal definition we need some additional notation.
Let $G$ and $H$ be hypergraphs and let $S =
\{(u_1,v_1),\ldots,(u_\ell,v_\ell)\}$ be a
subset of the Cartesian product $V(G) \times V(H)$.%
\footnote{Observe that $u_1,\ldots,u_\ell$, as well as
  $v_1,\ldots,v_\ell$, are in general
not distinct.}
The \emph{$G$-projection} is the function $\pi_G\colon V(G)\times
V(H)\to V(G)$ that maps
a pair to its first component, that is, $\pi_G(u,v)\coloneqq u$ for
all $(u,v)\in
V(G)\times V(H)$.
The $H$-projection $\pi_H$ is defined analogously.
We extend the notion of projections to sets $S\subseteq V(G)\times
V(H)$ in the natural
manner by setting
\begin{align*}
    \pi_G(S) \coloneqq \{\pi_G(s) \colon s\in S\}
    \quad\text{and}\quad
    \pi_H(S) \coloneqq \{\pi_H(s) \colon s\in S\}.
\end{align*}

Let $F,G,H$ be hypergraphs and let $\phi \colon V(F) \to V(G) \times
V(H)$ be a map.
The \emph{projection} of $\phi$ is the pair $\proj(\phi)=(\phi_G,\phi_H)$, where
$\phi_G(v)=\pi_G(\phi(v))$ and $\phi_H(v)=\pi_H(\phi(v))$ for all $v \in V(F)$.
Conversely, given two maps $\phi_G : V(F) \to V(G)$ and $\phi_H :
V(F) \to V(H)$, their
\emph{join} $\join(\phi_G,\phi_H)$ is defined as the function~${\phi\colon V(F) \to V(G)\times V(H)}$ that satisfies $\phi(v)=(\phi_G(v),\phi_H(v))$ for
all $v \in V(F)$.

\begin{remark}
    \dglabel{2.22-1}(We have $\join=\proj^{-1}$)
    Clearly $\join=\proj^{-1}$, since $\join(\proj(\phi))=\phi$ for
    every map $\phi$.
\end{remark}

\begin{definition}[The tensor product of hypergraphs, ${G\tensor H}$]
    \dglabel{def:hyp_tensor}
    For two hypergraphs $G$ and $H$, the \emph{tensor product}
    ${G\tensor H}$ is the
    hypergraph with vertex set $V(G) \times V(H)$ where each set $e
    \subseteq V(G) \times
    V(H)$ is an edge if and only if $\pi_G(e) \in E(G)$ and $\pi_H(e) \in E(H)$.
\end{definition}

Similar to the tensor product for graphs, the tensor product for hypergraphs is
multiplicative with respect to homomorphism counts.

\begin{lemma}[$\join$ and $\proj$ are compatible with the hypergraph tensor product]\dglabel{lem:join_prod_hom}
    If $\phi\in\homs{F}{G \tensor H}$, then $\proj(\phi)\in\homs{F}{G}\times\homs{F}{H}$.
    Conversely, if $(\phi_G,\phi_H) \in \homs{F}{G}\times\homs{F}{H}$, then
    $\join(\phi_G,\phi_H) \in \homs{F}{G \tensor H}$.
\end{lemma}
\begin{proof}
    Fix $\phi\in\homs{F}{G \tensor H}$ and set $(\phi_G,\phi_H)=\proj(\phi)$.
    Consider any $e\in E(F)$.
    Since $\phi\in\homs{F}{G \tensor H}$, we have $\phi(e)\in E(G\tensor H)$.
    By definition of~$\tensor$, this implies $\phi_G(e)\in E(G)$ and
    $\phi_H(e)\in E(H)$.
    Thus, $\phi_G \in \homs{F}{G}$ and $\phi_H \in \homs{F}{H}$, that
    is, $\proj(\phi) \in
    \homs{F}{G}\times\homs{F}{H}$.

    Now fix $(\phi_G,\phi_H) \in \homs{F}{G}\times\homs{F}{H}$ and set
    $\phi=\join(\phi_G,\phi_H)$.
    Again, consider any $e \in E(F)$.
    Since $\phi_G \in \homs{F}{G}$ and $\phi_H \in \homs{F}{H}$, we
    have $\phi_G(e)\in
    E(G)$ and $\phi_H(e)\in E(H)$.
    By the definition of $\tensor$, this implies $\phi(e) \in E(G \tensor H)$.
    Thus, $\phi \in \homs{F}{G \tensor H}$.
\end{proof}

\begin{corollary}[Hypergraph homomorphism counts are multiplicative
    with respect to the
    hypergraph tensor product]
    \dglabel{lem:semi_group_homomorphism}[lem:join_prod_hom,2.22-1]
    All hypergraphs $F,G,H$ satisfy \[\#\homs{F}{G \tensor H} =
    \#\homs{F}{G} \cdot \#\homs{F}{H}.\]
\end{corollary}
\begin{proof}
    By \cref{lem:join_prod_hom}, we have
    \begin{align*}
    & \proj(\homs{F}{G \tensor H}) \subseteq \homs{F}{G} \times \homs{F}{H},
    \quad                                                                     \\
    & \join(\homs{F}{G} \times \homs{F}{H}) \subseteq \homs{F}{G \tensor H}.
    \end{align*}
    Together with the fact that $\proj$ and $\join$ are inverses of
    each other (see
    \cref{2.22-1}), this implies that $\proj$ is a bijection from
    $\homs{F}{G \tensor H}$
    to $\homs{F}{G} \times \homs{F}{H}$.
    The claim follows.
\end{proof}

\begin{lemma}
    \dglabel{lem:efficient_Tensoring}(The hypergraph tensor product can
    be computed efficiently)
    The tensor product of two hypergraphs $G$ and $H$ can be computed
    in time $\Oh(2^{r^2}
    \cdot \lVert G \rVert \cdot \lVert H \rVert)$ where
    $r=\max(\rank(G),\rank(H))$.
\end{lemma}
\begin{proof}
    First, compute $V(G\tensor H)= V(G) \times V(H)$.
    Next, for each pair of edges $e_G \in E(G)$ and $e_H \in E(H)$, for
    each $e \subseteq
    e_G \times e_H$, add $e$ to $E(G \tensor H)$ if and only if $\pi_G(e)=e_G$ and
    $\pi_H(e)= e_H$.
    Observe that this produces all and only the edges of $G\tensor H$.
    The first step requires time $\Oh(|V(G)|\cdot|V(H)|) = \Oh(\lVert G
    \rVert \cdot
    \lVert H \rVert)$.
    Since $G$ and $H$ contain respectively at most $\lVert G \rVert$ and
    $\lVert H \rVert$
    edges, and each of them has size at most $r$, the second step requires time
    $\Oh(2^{r^2} \cdot \lVert G \rVert \cdot \lVert H \rVert)$.
\end{proof}

\subsection{Complexity monotonicity for hypergraph motif paramters}
\label{sub:hyper_motifs}
The main goal of this subsection is to prove \cref{lem:monotonicity_intro}.
In fact, we prove a more technical result,
\cref{cor:hypergraph_motif_monotonicity_fine_neq}, from which the direction
$\homsprob(\scH_\Gamma) \fptredlin \hyperME(\Gamma)$ follows immediately.
The other direction, $\hyperME(\Gamma) \fptredlin
\homsprob(\scH_\Gamma)$, is immediate: given a hypergraph $G$ and a function $\gamma \in
\Gamma$, it is
sufficient to compute $\#\homs{H}{G}$ for every $H \in \supp(\gamma) \subseteq
\scH_{\gamma}$, and then compute their linear combination into $\zeta_\gamma(G)$.
Thus, we have $\homsprob(\scH_\Gamma) \fpteqlin \hyperME(\Gamma)$, which proves \cref{lem:monotonicity_intro}.
As a further consequence we obtain the following result, which we extensively use in
\cref{sec:sub,sec:indsub}.
\begin{theorem}\label{clm:explicit_sub_dedekind}
    For every recursively enumerable family of hypergraphs $\scH$, we have
    \begin{align*}
        \homsprob(\scQ(\scH)) &\fpteqlin \subsprob(\scH)\,,\;\text{and}\\
        \homsprob(\scS(\scH)) &\fptredlin \indsubsprob(\scH) \fptredlin \homsprob(\scQ(\scS(\scH))).
    \end{align*}
\end{theorem}
\begin{proof}
    First, by chaining \cref{lem:subs_and_indsubs_are_motifs,lem:monotonicity_intro}, and
    by definition of $\homsprob(\scQ(\scH))$, we have
    \[
        \subsprob(\scH)=\hyperME\left(\Gamma_\scH\right) \fpteqlin
        \homsprob\left(\bigcup_{\gamma \in \Gamma_\scH}\supp(\gamma)\right) =
        \homsprob(\scQ(\scH)).
    \]
    Second, again by chaining \cref{lem:subs_and_indsubs_are_motifs,lem:monotonicity_intro},
    \[
        \subsprob(\scH)=\hyperME(\Gamma'_\scH)\fpteqlin \homsprob\left(\bigcup_{\gamma' \in
        \Gamma'_\scH}\supp(\gamma')\right).
    \]
    Now, again by \cref{lem:subs_and_indsubs_are_motifs},
    \[
        \scS(\scH) \subseteq \bigcup_{\gamma' \in \Gamma'_\scH}\supp(\gamma') \subseteq \scQ(\scS(\scH)),
    \]
    which proves the two reductions of the second claim.
\end{proof}
The rest of the subsection is therefore dedicated to developing the complexity
monotonicity principle for hypergraph motif parameters, and in particular to obtaining
\cref{cor:hypergraph_motif_monotonicity_fine_neq}.
To obtain complexity monotonicity for hypergraphs, we adapt the so-called ``Dedekind
interpolation'' technique of \cite{BLR23}.
Before starting, we observe that the proof of that technique (see~\cite[full version,
Theorem 36]{BLR23}) is not entirely precise, as it does not discuss the complexity of
arithmetic operations over a certain abstract semigroup used in the proof.
For this reason, we provide an alternative version of the statement via circuits.
Thus, although our proofs follow the same argument as  in~\cite{BLR23}, they require an
additional recursive construction of a circuit; they can be found in the appendix.

\begin{restatable*}[A Dedekind circuit
    $\mathbb{D}(\varphi_1,\dots,\varphi_k)$  for semigroup homomorphisms
    $\varphi_1,\dots,\varphi_k$]{definition}{ddcirc}%
    \label{def:dedekind_circuit}%
    For~a~semigroup $(\mathrm{G},\ast)$
    and pairwise distinct semigroup homomorphisms%
    \footnote{That is, we have $\varphi_i(g\ast g')= \varphi_i(g)\cdot
        \varphi_i(g')$ for all $i\in\{1,\dots,k\}$ and
    $g,g'\in \mathrm{G}$.}
    $\varphi_1,\dots,\varphi_k\colon
    \mathrm{G} \to \Q$,
    a \emph{Dedekind Circuit}
    $\mathbb{D}(\varphi_1,\dots,\varphi_k)$ is an arithmetic circuit
    over~$\mathbb{Q}$ with
    \begin{itemize}
        \item $k$ \emph{output gates} $\mathrm{out}_1,\dots,\mathrm{out}_k$; and
        \item for finitely many \(g \in \mathrm{G}\), an \emph{input gate} $\mathrm{in}[g]$
            labeled with \(g\).
        \item Further, for each linear combination
            $F=a_1\varphi_1 + \dots + a_k\varphi_k$
            (with \(a_1,\dots,a_k \in \mathbb{Q}\)),
            the circuit $\mathbb{D}(\varphi_1,\dots,\varphi_k)$
            when each gate \(\mathrm{in}[g]\) is assigned the value \(F(g)\), outputs
            \(\mathrm{out}_i = a_i\) for all \(1 \le i \le k\).
            \qedhere
    \end{itemize}
\end{restatable*}
\begin{restatable*}[Efficiently isolating terms in finite linear
    combinations of semigroup homomorphisms, ``Dedekind Interpolation'',
    {\cite[full version, compare Theorem 36]{BLR23}}]{theorem}{rstdi}
    \dglabel{thm:dedekind_new}
    Let $(\mathrm{G},\ast)$ be a computable semigroup and let
    ${\varphi_1,\dots,\varphi_k\colon \mathrm{G} \to \Q}$ be pairwise
    distinct and computable semigroup homomorphisms from
    $(\mathrm{G},\ast)$ to $(\Q,\cdot)$.
    There exists an algorithm~$\mathbb{A}$ with the following properties.
    \begin{enumerate}
        \item $\mathbb{A}$ receives as input
            \begin{itemize}
                \item semigroup elements $g_1,\dots,g_k \in \mathrm{G}$
                    with $\varphi_i(g_i)\neq 0$ for all $i\in\{1,\dots,k\}$, and
                \item for each $i,j\in\{1,\dots,k\}$ with $i<j$, a
                    semigroup element $g_{i,j}\in \mathrm{G}$ with
                    $\varphi_i(g_{i,j})\neq \varphi_j(g_{i,j})$.
            \end{itemize}
        \item $\mathbb{A}$ computes a Dedekind Circuit
            $\mathbb{D}(\varphi_1,\dots,\varphi_k)$ of depth $O(k)$ and
            constant fan-in.
            \qedhere
    \end{enumerate}
\end{restatable*}

Since the fan-in is constant, there are exactly $k$ output gates,
and the depth of the circuit is $O(k)$, the circuit has at most
$\exp(O(k))$ gates in total.

We are now able to provide an efficient method for isolating terms
in a hypergraph motif parameter; in other words, we show that the
computation of a linear combination of (hypergraph) homomorphism
counts is at least as hard as computing its hardest term.
\begin{corollary}[Efficiently computing single $\#\homs{H}{G}$
    with oracle access to
    $\zeta_\gamma$]%
    \dglabel{cor:hypergraph_motif_monotonicity_fine_neq}%
    [lem:efficient_Tensoring,lem:semi_group_homomorphism,thm:dedekind]
    There is a computable function $f$ and an oracle algorithm~$R$
    with the following properties.
    \begin{enumerate}
        \item The input of $R$ is a hypergraph~$G$ and a finitely
            supported function~$\gamma$ from hypergraphs to $\Q$,
            encoded in binary.
        \item $R$ has access to an oracle for the hypergraph motif
            parameter~$\zeta_\gamma$.
        \item $R$ computes $\#\homs{H}{G}$ for all $H\in\supp(\gamma)$.
        \item The running time of $R$ is $O(f(\gamma)\cdot \lVert G
            \rVert)$ and it makes at most $f(\gamma)$ queries.
    \end{enumerate}
\end{corollary}
\begin{proof}
    Let $\supp(\gamma)=\{H_1,\dots,H_k\}$, and let $\ell$ be the
    maximum size of any edge in any $H_i$.
    One can compute $\ell$ in time $f(\lVert\gamma\Vert)$; and in
    time $\Oh(\lVert G
    \rVert)$ one can then delete from $G$ all edges of size strictly
    larger than $\ell$.
    Note that this does not change $\#\homs{H}{G}$ for any $H \in
    \supp(\gamma)$.
    Thus, without loss of generality we may assume $G$ contains edges
    of size at most $\ell$.

    Recall that our goal is the computation of $\#\homs{H_i}{G}$ for
    all $i\in[k]$. To this end, we rely on Dedekind Interpolation
    which requires us to specify a semigroup together with a sequence
    of semigroup homomorphisms. We choose $(\mathrm{G},\otimes)$
    where $\mathrm{G}$ is the set of all (isomorphism types of)
    hypergraphs, and $\otimes$ denotes the hypergraph tensor product
   ---clearly, the tensor product is associative, that is $G_1
    \otimes (G_2 \otimes G_3) \cong (G_1 \otimes G_2) \otimes G_3$.
    Next, for every $i\in [k]$, we define $\varphi_i: \mathrm{G} \to
    \mathbb{Q}$ as the function $F \mapsto \#\homs{H_i}{F}$.
    Our analysis in~\cref{sub:homs_and_tensors} shows that the
    $\varphi_i$ are all semigroup homomorphisms
    (w.r.t.\ $(\mathrm{G},\otimes)$).

    For the application of \cref{thm:dedekind_new}, we set $g_i= H_i$
    for each $i\in [k]$, as $\varphi_i(H_i)=\#\homs{H_i}{H_i}>0$.
    Moreover, for the $g_{i,j}$, we recall  Lovász' Theorem stating
    that, for graphs, the (infinite) vector $(\#\homs{H}{F})_F$
    determines the isomorphism type
    of $H$ (see~\cite[Chapter 5.3]{Lovasz12}), and we note that the
    proof applies without non-trivial modifications to hypergraphs as
    well. As a consequence, for each pair $i<j \in [k]$, there must
    be a hypergraph $g_{i,j}$ such that $\#\homs{H_i}{g_{i,j}}\neq
    \#\homs{H_j}{g_{i,j}}$ as $H_i$ and $H_j$ are not isomorphic.
    Lastly, since the $H_i$ only depend on $\gamma$ (and not on $G$),
    we may search for the $g_{i,j}$ by brute force in time only depending on $\gamma$.

    We next call the algorithm $\mathbb{A}$
    from \cref{thm:dedekind_new} with the $g_i$ and $g_{i,j}$ as
    specified previously. Let
    $\mathbb{D}=\mathbb{D}(\varphi_1,\dots,\varphi_k)$ be the
    Dedekind Circuit returned by $\mathbb{A}$. Also, we point out
    that the running time of $\mathbb{A}$, the size of $\mathbb{D}$,
    and the hypergraphs associated to the input gates all only depend
    on $\gamma$, and not on $G$.

    Next, we observe that our oracle for $\zeta_\gamma$ may be used
    to evaluate the following function $F:\mathrm{G} \to \mathbb{Q}$; we have
    \begin{align*}
      F(g) = \zeta_\gamma(G \otimes g) & = \sum_{i=1}^k \gamma(H_i)
      \cdot \#\homs{H_i}{G \otimes g}
      \\
      ~                                & = \sum_{i=1}^k \gamma(H_i)
      \cdot \#\homs{H_i}{G} \cdot
      \#\homs{H_i}{g}
      \\
      ~                                & = \sum_{i=1}^k a_i \cdot
      \varphi_i(g) ,
    \end{align*}
    where $a_i = \gamma(H_i)\cdot \#\homs{H_i}{G}$.

    In the next step, we evaluate $\mathbb{D}$ by inserting to
    each input gate labelled with hypergraph $g$ the value
    $F(g)=\zeta_\gamma(G \otimes g)$---this is done via our oracle
    access and we point out that the total number of calls only
    depends on the number of input gates, and thus only on $\gamma$,
    as required. Note that, since $g$ only depends on $\gamma$, the
    time to construct $G \otimes g$, and its size, is bounded by
    $O(f_1(\gamma) \cdot \lVert G\rVert)$ for some computable
    function $f_1$ (see~\Cref{lem:efficient_Tensoring}).

    Observe that $F(g)\in O(f_2(\gamma) \cdot |V(G)|^\kappa)$ for
    some computable function $f_2$, where $\kappa$ is the maximum
    number of vertices over the hypergraphs $H_1,\dots,H_k$.  Since both
    $\kappa$ and the size of $\mathbb{D}$ only depend on $\gamma$ as
    well, the maximum number of bits required for any intermediate
    value in the circuit is bounded by $O(f_3(\gamma) \cdot \log
    (\lVert G\rVert))$ for some computable function $f_3$. Hence, we
    generously bound the time it takes to evaluate $\mathbb{D}$
    by $O(f'(\gamma) \cdot \lVert G \rVert)$ for some computable
    function $f'$. By definition of Dedekind circuits, the output
    gates yield $a_1,\dots,a_k$. In the final step, we output
    $a_i/\gamma(H_i) = \#\homs{H_i}{G}$ for all $i\in [k]$; note that
    $\gamma(H_i)\neq 0$ for all $i\in [k]$.

    The total running time is as follows.
    \begin{enumerate}
        \item Removing edges of $G$ of size larger than $\ell$: $\Oh(\lVert G
            \rVert)$.
        \item Computing the $g_i$ and $g_{i,j}$ for $i,j \in [k]$: only
            depends on $\gamma$.
        \item Construction of $\mathbb{D}$: only depends on $\gamma$.
        \item Constructing $G \otimes g$ and computing $\zeta_\gamma(G
            \otimes g)(=F(g))$ via our oracle for all hypergraphs $g$
            equipped to the input gates of $\mathbb{D}$: $O(f_1(\gamma)
            \cdot \lVert G\rVert)$.
        \item Evaluating $\mathbb{D}$ and returning $a_i/\gamma(H_i) =
            \#\homs{H_i}{G}$ for all $i\in [k]$: $O(f'(\gamma) \cdot
            \lVert G \rVert)$.
    \end{enumerate}
    The final running time is thus bounded by $O(f(\gamma) \cdot
    \lVert G \rVert)$, which concludes the proof.
\end{proof}

\subsection{Colourful and colour-prescribed homomorphisms}
\label{sub:homs_and_colhoms}
In the literature of counting subgraphs and induced subgraphs, coloured variants of these
problems naturally appear.
We discuss a generalization of these notions to hypergraphs.
To this end, let $H$ be a hypergraph.
An $H$\emph{-coloured hypergraph} is a pair $(G,c)$ of a hypergraph
$G$ and a homomorphism
$c\in \homs{G}{H}$.
In this context, $c$ is also called an \emph{$H$-colouring}.

\begin{definition}[Colour-prescribed and colourful homomorphisms
    between hypergraphs]%
    \dglabel{7-21-2}%
    A \emph{colour-prescribed homomorphism} from $H$ to $(G,c)$ is a
    homomorphism $\varphi
    \in \homs{H}{G}$ with $c(\varphi(v))=v$ for all $v\in V(H)$.
    A \emph{colourful} homomorphism from $H$ to $(G,c)$ is a homomorphism
    $\varphi \in
    \homs{H}{G}$ with $c(\varphi(V(H)))=V(H)$.
\end{definition}
We denote by $\cphoms{H}{(G,c)}$ the set of all colour-prescribed
homomorphisms from $H$
to $(G,c)$ and by $\cfhoms{H}{(G,c)}$ the set of all colourful
homomorphisms from $H$ to
$(G,c)$.
As with uncoloured homomorphisms, we also define corresponding
computational problems.

\noindent
\begin{minipage}{1\linewidth}%
  \begin{problem}-[$\#\textupsc{cp-Hom}$]{$\#\textupsc{cp-Hom}(\scH)$}
    \label{prob:cphoms}
    \PInput{A hypergraph $H \in \scH$ and an $H$-coloured hypergraph $G$}
    \POutput{$\#\cphoms{H}{(G,c)}$}
    \PParameter{$|H|$}
  \end{problem}

  \begin{problem}+=[$\#\textupsc{cf-Hom}$]{$\#\textupsc{cf-Hom}(\scH)$}
    \label{prob:cfhoms}
    \PInput{A hypergraph $H \in \scH$ and an $H$-coloured hypergraph $G$}
    \POutput{$\#\cfhoms{H}{(G,c)}$}
    \PParameter{$|H|$}
  \end{problem}
\end{minipage}

The main result of this section reads as follows.
\begin{lemmaq}
    \dglabel{lem:equiv_uncol_cp_cf}%
    [lem:cp_to_cf,lem:cf_to_nocol,lem:homs_to_cphoms]
    ($\homsprob(\scH)\fpteq\!\cphomsprob(\scH)\fpteq\!\cfhomsprob(\scH)$)
    For any family $\scH$ of hypergraphs
    $\homsprob(\scH)\fpteqlin\!\cphomsprob(\scH)\fpteqlin\!\cfhomsprob(\scH)$.
\end{lemmaq}
\Cref{lem:equiv_uncol_cp_cf} follows by the next three results,
\cref{lem:cp_to_cf_new,,lem:cf_to_nocol_new,,lem:homs_to_cphoms_new}.
\begin{lemma}
    \dglabel{lem:cp_to_cf_new}
    For every hypergraph $H$ and $H$-coloured hypergraph $(G,c)$, we
    have $\#\cphoms{H}{(G,c)}=\#\auts{H}\cdot\#\cfhoms{H}{(G,c)}$.
    As a consequence, $\cphomsprob(\scH)\fpteqlin\!\cfhomsprob(\scH)$
    for every hypergraph family \(\scH\).
\end{lemma}
\begin{proof}
    If $c$ is not surjective, both sides of the equation are 0, so assume this is not the case.
    By definition, color-prescribed homomorphisms are injective.
    Observe that if $h\in\cphoms{H}{(G,c)}$ and $a\in\auts{H}$, then
    $h\circ a$ is colourful as $h(V(H))=(h\circ a)(V(H))$ and $h$ is colour-prescribed.
    This shows ``$\leq$''.

    For ``$\geq$'', we show that for every colourful homomorphism
    $h\in\cfhoms{H}{(G,c)}$,
    there is a unique~$a\in\auts{H}$ such that $h\circ a$ is colour-prescribed.
    Observe that $a\coloneqq c\circ h$ is a surjective endomorphism of~$H$ and hence
    bijective as $V(H)$ is finite.
    Thus, $a$ is an automorphism.
    Then,~${h\circ a^{-1}}$ is colour-prescribed as $c\circ h\circ
    a^{-1}=(c\circ h)\circ(c\circ h)^{-1}=\id$ where $\id$ is the
    identity map on~$V(H)$.
    As composition with an automorphism does not change the image
    of $h$ and different colour-prescribed homomorphism have different
    images, $a$ is the unique automorphism that makes $h$ colour-prescribed.
\end{proof}
\begin{lemma}
    \dglabel{lem:cf_to_nocol_new}
    For every hypergraph $H$ and $H$-coloured hypergraph $(G,c)$, we have
    \begin{align*}
        \#\cfhoms{H}{(G,c)}=\sum_{I\subseteq
        V(H)}(-1)^{|I|}\cdot\#\homs{H}{G-I},
    \end{align*}
    where $G-I$ denotes the hypergraph obtained from $G$ by deleting all
    vertices coloured with a colour in $I$.
    As a consequence, $\cfhomsprob(\scH)\fptredlin\homsprob(\scH)$ for
    every hypergraph family \(\scH\).
\end{lemma}
\begin{proof}
    Let $(H,(G,c))$ be any instance of $\cfhomsprob(\scH)$.
    For every $I\subseteq V(H)$ define the set
    \begin{align*}
        A_I\coloneqq\{h\in\homs{H}{G}\mid c(h(V(H)))\cap I=\emptyset\}.
    \end{align*}
    In words, $A_I$ is the set of homomorphisms from $H$ to $G$ that
    do not use any of the
    vertices in $G$ coloured with a colour from $I$.
    For $j\in V(H)$, we write $A_j$ instead of $A_{\{j\}}$.
    Observe that~${A_I\cap A_J=A_{I\cup J}}$ for all~${I,J\subseteq V(H)}$.
    Then, $|A_I|=\#\homs{H}{G-I}$.
    Applying the inclusion-exclusion-principle yields
    \begin{equation*}
        \#\cfhoms{H}{(G,c)}=\#\bigcap_{j\in V(H)}\overline{A_j}
        =\sum_{I\subseteq V(H)}(-1)^{|I|}\cdot|A_I|
        =\sum_{I\subseteq V(H)}(-1)^{|I|}\cdot\#\homs{H}{G-I}.\qedhere
    \end{equation*}
\end{proof}
\begin{lemma}
    \dglabel{lem:homs_to_cphoms_new}
    There is a deterministic algorithm $\mathbb{A}$ with the
    following behaviour.
    \begin{enumerate}
        \item $\mathbb{A}$ expects as input hypergraphs $H$ and $G$.
        \item $\mathbb{A}$ computes an $H$-coloured hypergraph $(F,c)$
            of size at most $\lVert H\rVert\cdot\lVert G\rVert$ such that
            $\#\homs{H}{G}=\#\cphoms{H}{(F,c)}$.
        \item The running time of $\mathbb{A}$ is bounded by
            $\Oh(2^{|H|^2}\cdot\lVert H\rVert\cdot\lVert G\rVert)$.
    \end{enumerate}
    As a consequence, $\homsprob(\scH)\fptredlin\cphomsprob(\scH)$
    for every hypergraph family $\scH$.
\end{lemma}
\begin{proof}
    Let $(H,G')$ be any instance of $\homsprob(\scH)$.
    In order to compute the tensor product in \ccFPT-time, first
    compute a hypergraph $G$
    such that $\homs{H}{G'}=\homs{H}{G}$ and $\rank(G)\leq\rank(H)$.
    Explicitly, $G$ is obtained from $G'$ by removing all edges $e\in
    E(G')$ with
    $|e|>\rank(H)$.

    To see that $\homs{H}{G'}=\homs{H}{G}$, observe that any
    homomorphism $\phi$ from $H$
    to $G'$ has to map edges of $H$ to edges of $G'$.
    Through this mapping, the size of an edge cannot increase.
    In particular, no edge of $G'$ with size strictly greater than
    $\rank(H)$ is ever hit
    by $\phi$ so that the edge may be removed without invalidating
    any element of
    $\homs{H}{G'}$.

    Let $F = G \tensor H$ be the tensor product of $G$ and $H$
    (\cref{def:hyp_tensor}).
    By \cref{lem:efficient_Tensoring}, $F$ can be computed in time
    $\Oh(2^{|H|^2}\cdot\lVert H\rVert\cdot\lVert G\rVert)$.
    Now consider the $H$-coloured graph $(F,c)$ where $c$ is the
    colouring defined by
    $c((u,v))=v$ for all $(u,v) \in V(F) = V(G) \times V(H)$.
    We claim that there is a bijection between $\homs{H}{G}$ and
    $\cphoms{H}{(F,c)}$.

    Fix $\varphi \in \homs{H}{G}$.
    Define $\psi : V(H) \to V(F)$ by setting $\psi(u) =
    (\varphi(u),u)$ for all $u \in
    V(H)$.
    As $\varphi \in \homs{H}{G}$, for every $e \in E(H)$ we have
    $\varphi(e) \in E(G)$.
    By the definition of the tensor product, $\psi(e) =
    \{(\varphi(u),u) : u \in e\}$ is
    an edge of $F$.
    Hence, $\psi \in \homs{H}{F}$.
    Moreover, we have $c(\psi(u))=u$. Hence, $\psi \in \cphoms{H}{(F,c)}$.
    Finally, distinct $\varphi,\varphi' \in \homs{H}{G}$ clearly
    yield distinct
    $\psi,\psi' \in
    \cphoms{H}{(F,c)}$.

    Now, let $\psi \in \cphoms{H}{(F,c)}$, and define $\varphi : V(H)
    \to V(G)$ by setting
    $\varphi(u) = x_u$ for every $u \in V(H)$, where $(x_u,y_u)=\psi(u)$.
    As $\psi\in\homs{H}{F}$ and by the definition of the tensor
    product, for every $e \in
    E(H)$ we have $\{x_u:u \in e\} \in E(G)$.
    Therefore, $\varphi \in \homs{H}{G}$.
    Moreover, according to the definition of $c$ every $\psi \in
    \cphoms{H}{(F,c)}$
    satisfies $\psi(u)=(x_u,u)$ for some $x_u \in V(G)$ for every
    $u \in V(H)$.
    Therefore, distinct $\psi,\psi' \in \cphoms{H}{(F,c)}$ yield distinct
    $\varphi,\varphi'
    \in \homs{H}{G}$.

    Thus, the mappings between $\cphoms{H}{(F,c)}$ and $\homs{H}{G}$
    described above
    define a bijection, as desired.
\end{proof}

We conclude with some results, needed by our proofs, on the
hardness monotonicity of
hypergraph trimming.
First, in \cref{lem:WM_fptred_new} we show that taking induced
trimmed sub-hypergraphs does not make homomorphism counting harder.
Second, in \cref{lem:aw_monotone} we show that adaptive width is
monotone with respect to trimming.

For a family of hypergraphs \(\scH\), write $\WeakMinors(\scH)$ for
the family of all
hypergraphs that are an induced trimmed sub-hypergraph of some
hypergraph in \(\scH\);
formally, $\WeakMinors(\scH)\coloneqq\{\trimsub{G}{X}\mid
    G\in\scH,\emptyset\neq
X\subseteq V(G)\}$.

\begin{lemma}
    \dglabel{lem:WM_fptred_new}
    There exists a computable function $f$ and a deterministic
    algorithm that expects as input a hypergraph $H_0$, an induced
    trimmed subhypergraph $H$ of $H_0$ and an $H$-coloured hypergraph
    $(G,c)$, and that outputs an $H_0$-coloured hypergraph $(G_0,c_0)$ such that
    \[ \#\cphoms{H}{(G,c)} = \#\cphoms{H_0}{(G_0,c_0)}. \]
    The size of $(G_0,c_0)$ and the running time of the algorithm are
    both bounded by $f(|H_0|)\cdot \Oh(\lVert(G,c)\rVert)$.

    As a consequence,
    $\homsprob(\WeakMinors(\scH))\fptredlin\homsprob(\scH)$ for every
    recursively enumerable hypergraph family $\scH$.
\end{lemma}
\begin{proof}
    First, we compute, in time only depending on $H_0$, a set
    $\emptyset\neq X\subseteq V(H_0)$ such that
    $H=\trimsub{H_0}{X}$.
    Note that the induced trimmed subgraph $\trimsub{H_0}{X}$ can be
    obtained from a sequence of induced trimmed subgraphs where in each step only one
    vertex is trimmed away.
    Write $V(H_0)\setminus X=\{v_1,\dots,v_\ell\}$ and set~$H_i\coloneqq\trimsub{H_0}{V(H_0)\setminus\{v_1,\dots,v_i\}}$.
    Then $H=H_\ell$.
    Set $(G_\ell,c_\ell)\coloneqq(G,c)$.
    For \(i = \ell - 1, \dots, 0\), we construct a corresponding
    \(H_i\)-coloured
    hypergraph \((G_i, c_i)\) based on \((G_{i+1}, c_{i+1})\) as follows.
    Set
    \[ E_{i+1}\coloneqq\{e\setminus\{ v_{i+1}\}\mid  v_{i+1}\in
    e\in E(H_i)\} \]
    and for a fresh vertex $w_{i+1} \notin V(G_{i+1})$ set
    \begin{align*}
        V(G_i) & \coloneqq V(G_{i+1})\cup\{w_{i+1}\},
        \\
        E(G_i) & \coloneqq \big\{e\in E(G_{i+1})\mid c_{i+1}(e)\in E(H_i) \big\}
        \cup \big\{ e\cup\{w_{i+1}\}\mid e\in
        E(G_{i+1}), c_{i+1}(e)\in E_{i+1} \big\},\quad \text{and}
        \\
        c_i(v) & \coloneqq
        \begin{cases}
            v_{i+1},    & v = w_{i+1}    \\
            c_{i+1}(v), & v \neq w_{i+1}
        \end{cases}.
    \end{align*}

    For the correctness, first, we inductively argue (for \(i = \ell
    - 1, \dots, 0\)) that
    each \(c_i\) is indeed a homomorphism from \(G_i\) to \(H_i\).
    To that end, we need to verify the condition \(c_i(e) \in
    E(H_i)\) for every edge
    \(e\in E(G_i)\).
    For edges \(w_{i+1} \not\in e \in E(G_{i})\), by construction, we have
    \[c_i(e) = c_{i+1}(e) \in E(H_i).\]
    For edges \(w_{i+1} \in e \in E(G_i)\), by construction, we have
    \[c_i(e) = c_i(e \setminus \left\{ w_{i+1} \right\}) \cup \{c_i(w_{i+1})\}
    \in E_{i+1} + v_{i+1} \subseteq E(H_i);\]
    where we set \(A + b \coloneqq \{ a \cup \{b\}  \mid a \in A \}\).

    Next, we inductively prove (for \(i = \ell - 1, \dots, 0\)) that
    \[\#\cphoms{H_{i}}{(G_{i},c_{i})} =
    \#\cphoms{H_{i+1}}{(G_{i+1},c_{i+1})}.\]
    To this end, fix a homomorphism
    $\phi\in\cphoms{H_{i+1}}{(G_{i+1},c_{i+1})}$ and
    consider the mapping $\psi : V(H_i) \to V(G_i)$ with
    $\psi(v_{i+1})=w_{i+1}$ and
    $\psi(v)=\phi(v)$ for $v\in V(H_{i+1})$.
    Clearly, $\psi$ is colour-prescribed.
    It remains to show that \(\psi\) is a homomorphism.

    For every edge $v_{i + 1}\notin e\in E(H_i)$, the mapping $\psi$
    inherits the
    homomorphism property from \(\phi\).
    For every edge $v_{i+1} \in e\in E(H_i)$, we have $e_{{i+1}}\coloneqq
    e\setminus\{v_{i+1}\}\in E_{i+1}$ and $\phi(e_{i+1})\in E(G_{i+1})$ because
    $\phi\in\cphoms{H_{i+1}}{(G_{i+1},c_{i+1})}$.
    As $\phi$ is colour-prescribed, we have $c_{i+1}(\phi(e_{i+1})) =
    e_{i+1} \in
    E_{i+1}$.
    Hence, $\psi(e)=\phi(e_{i+1}) \cup \{w_{i+1}\} \in E(G_i)$ by
    construction of $G_i$.
    As $w_{i+1}$ is the only vertex in $V(G_{i})$ with the colour
    $v_{i+1}$ and $c_{i+1}$
    agrees with $c_{i}$ on $V(G_{i+1})$, the above (implicit) mapping between
    $\phi\in\cphoms{H_{i+1}}{(G_{i+1},c_{i+1})}$ and
    $\psi\in\cphoms{H_i}{(G_i,c_i)}$ is
    bijective, with the inverse map being the restriction to $V(G_{i+1})$.
    This yields the claimed equation.
    Hence, we also have
    \[\#\cphoms{H}{(G,c)} = \#\cphoms{H_{0}}{(G_{0},c_{0})}.\]

    Finally, we analyse the running time of our algorithm.
    We first recall that $X$ can be computed in time only depending on $H_0$.
    Next, for each \(0 \le i <\ell\) we construct \((G_i, c_i)\)
    from \((G_{i+1},c_{i+1})\) by a single pass over \((G_{i+1},
    c_{i+1})\) incurring a running time of \(\Oh(\lVert(G_{i+1},
    c_{i+1})\rVert)\).
    As the number of edges is potentially doubled in each step, we
    incur a total running time of $2^{\Oh(\ell)}\cdot \Oh(\lVert (G,
    c)\rVert)$, which we bound by~$f(|H_0|)\cdot \Oh(\lVert(G,c)\rVert)$.
    This completes the proof.
\end{proof}

%% file: s5_sub_counting.tex
\section{Complexity classification for
    \texorpdfstring{$\subsprob$}{Sub-Hypergraph
Counting}}\label{sec:subsclassification}
\label{sec:sub}

The next three subsections prove respectively the upper
bounds and the lower bounds of \cref{thm:classification_subsprob},
which we recall for convenience, and \cref{rem:4-3-1}.

\rstmthmone*

\paragraph*{Proof of \cref{thm:classification_subsprob}: upper bounds}

The bounds follow follow from the following facts.
\begin{enumerate}
    \item $\fhtw(\scQ(\scH)) \le \frcoindno(\scH) + 1$, see \cref{2.21-1} below.
    \item $\homsprob(\scQ(\scH))$ is solvable in time $f(H) \cdot
        \size{G}^{\fhtw(\scQ(\scH))+O(1)}$, see \cref{2.21-2} below.
    \item $\subsprob(\scH) \fptredlin \homsprob(\scQ(\scH))$, see~\cref{clm:explicit_sub_dedekind}.
\end{enumerate}

\begin{lemma}
    \dglabel{2.21-1}[def:hypergraph_quotient](Bounded fractional
    co-independent edge-cover number implies bounded fractional hypertree width)
    Let $\scH$ be any hypergraph family with $\frcoindno(\scH)<\infty$.
    Then $\fhtw(\scQ(\scH))\leq\frcoindno(\scH)+1<\infty$.
\end{lemma}
\begin{proof}
    Let $H'=(V',E') \in \scQ(\scH)$, and choose $H=(V,E) \in \scH$ so
    that $H'=H/\tau$ for some
    partition $\tau=\{V_1,\ldots,V_\ell\}$ of $V$.
    \begin{claim}
        \label{2.21-3}
        We have $\frcoindno(H') \le \frcoindno(H)$.
    \end{claim}
    \begin{claimproof}
        Let $S$ be a co-independent set of $H$ such that $\frcoindnoS{H}(S)
        =\frcoindno(H)$, and let $\xi \colon E \to \Rp$ be an optimal fractional
        edge-cover of $S$.
        Set $S' \coloneqq S/\tau$; by \cref{def:hypergraph_quotient},
        \[S' \coloneqq S/\tau \coloneqq \{V_i \in \tau : V_i \cap S \ne
        \emptyset\}.\]
        We claim that $S'$ is co-independent in $H'$.
        To this end, choose any two distinct vertices $V_i, V_j \in V'
        \setminus S'$.
        By definition of $S'$, the fact that $V_i \in V' \setminus S'$
        implies $V_i \cap
        S= \emptyset$ and thus $V_i \subseteq V \setminus S$.
        Similarly, we obtain $V_j \subseteq V \setminus S$.
        We claim that $V_i,V_j$ are not adjacent in $H'$.
        For an indirect proof, suppose that there is an $e' \in E'$ with
        $V_i,V_j \in e'$.
        This implies the existence of an $e \in E$ such that $e \cap V_i
        \ne \emptyset$ and $e \cap V_j \ne \emptyset$.
        As $V_i,V_j \subseteq V \setminus S$, however, this means that $V
        \setminus S$ is
        not independent in $H$, a contradiction.
        We conclude that $V_i,V_j$ are not adjacent in $H'$.
        Therefore, $V'\setminus S'$  is independent in $H'$.
        Thus, $S'$ is co-independent in $H'$.

        Next, we define a function $\xi' : E' \to \Rp$ via
        \begin{align*}
            \xi'(e') = \sum_{\substack{e \in E \\ e/\tau = e'}} \xi(e),
            \qquad \text{for
            all }e' \in E'.
        \end{align*}
        We claim that $\xi'$ is a fractional edge-cover of $S'$ of weight
        at most $\frcoindno(H)$.
        To this end observe that, by construction, we have $\sum_{e' \in
        E'} \xi'(e') \le
        \sum_{e \in E} \xi(e)=\frcoindno(H)$.
        To show that $\xi'$ is a fractional edge-cover of $S'$ in $H'$,
        first note that for every $U \in V'$, we have
        \begin{align*}
            \xi'(U) = \sum_{e' \in E' : U \in e'} \xi'(e') = \sum_{e \in E : U
            \cap e \ne \emptyset} \xi(e).
        \end{align*}
        Now, if $U \in S'$, then by construction we have $U \cap S \ne \emptyset$.
        Thus, there exists a $v \in S \cap U$.
        Since every $e \in E$ with $v \in e$ satisfies in particular $U \cap e
        \ne\emptyset$, we obtain
        \begin{align*}
            \sum_{e \in E : U \cap e \ne \emptyset} \xi(e) \ge \sum_{e \in E : v \in
            e} \xi(e) \ge 1,
        \end{align*}
        where the second inequality follows from $v \in S$ and $\xi$
        being a fractional edge-cover of $S$ in $H$.
        This completes the proof of the claim.
    \end{claimproof}
    \begin{claim}
        \label{2.21-4}
        For any hypergraph \(H\) we have $\fhtw(H) \le \frcoindno(H)+1$.
    \end{claim}
    \begin{claimproof}
        As before, let $S$ be co-independent in $H$ such that
        $\frcoindnoS{H}(S)=\frcoindno(H)$, and consider the following
        tree decomposition
        $(T,\scB)$.
        The tree $T$ is formed by joining the vertices $v \in V \setminus S$ in a path in
        any order and finally joining a new vertex \(w\notin V\) to the path.
        The family $\scB$ contains the set $B_v = S \cup \{v\}$ for each $v \in V\setminus S$ and the set \(B_w=S\).

        One may verify that  $(T,\scB)$ constitutes a tree decomposition of $H$.
        Moreover, each $B_v$ admits a fractional edge-cover of weight at most
        $\frcoindno(H)+1$, obtained by extending the optimal fractional
        edge-cover $\xi$
        of $S$ by assigning weight $1$ to any $e \in E$ that contains $v$.
    \end{claimproof}
    Combining \cref{2.21-3,2.21-4} yields $\fhtw(\scH') \le
    \frcoindno(\scH)+1 < \infty$, which implies the claim.
\end{proof}

Next, we use the following folklore corollary of a result of Grohe
and Marx (\cite[Theorem 3.5]{GroheM14}); we give a formal proof for
completeness.
\begin{theorem}\dglabel{2.21-2}($\homsprob(\scH)$ is $\ccFPT$
    whenever $\fhtw(\scH)<\infty$)
    For every hypergraph family $\scH$ with $\fhtw(\scH)<\infty$ the
    problem $\homsprob(\scH)$ is solvable in time~$f(\size{H}) \cdot
    \size{G}^{\fhtw(H)+O(1)}$ for some computable function $f$.
\end{theorem}
\begin{proof}
    Following the notation of~\cite{GroheM14}, let $I=(V,D,C)$ be the
    CSP (constraint satisfaction problem) instance defined as follows.
    \begin{itemize}
        \item $V = V(H)$.
        \item $D=V(G)$.
        \item $C$ contains for every $e \in E(H)$ the constraint $\left(\langle
            u_1,\ldots,u_r \rangle, R_e\right)$, where $u_1,\ldots,u_r$ are the
            vertices of $e$ in
            arbitrary order, and $R_e$ contains, for every $e_G \in E(G)$ with
            $|e_G| \le r$,
            every $r$-tuple from $(e_G)^r$ where each element of $e_G$ appears
            at least once.
    \end{itemize}
    The set of solutions of $I$ is precisely $\homs{H}{G}$.
    Indeed, a solution of $I$ is by definition a mapping $\phi : V \to
    D$ that for every
    $e \in E(H)$ satisfies $\langle\phi(u_1),\ldots,\phi(u_r)\rangle
    \in R_e$, which means
    precisely that the image of $e$ under $\phi$ is some $e_G \in E(G)$.
    Moreover one may compute $I$ in time $f(\size{H}) \cdot
    \size{G}^{O(1)}$, and that
    $\size{I} \le f(\size{H}) \cdot \size{G}$ for some computable $f$.
    Finally, by~\cite[Theorem 3.5]{GroheM14}, the solutions of $I$ can
    be enumerated in
    time $\size{I}^{\fhtw(H)+O(1)}$.
    This implies the claim, concluding the proof.
\end{proof}

\paragraph*{Proof of \cref{thm:classification_subsprob}: lower bounds}
Before delving into the proof, we need a technical lemma that
``extracts'' certain trimmed
quotients from hypergraphs with a large gap between adaptive width
and fractional
co-independent edge-cover number.

\begin{lemma}\label{lem:hat_H_subs}
    Let $H$ be a hypergraph, and suppose $t = \left\lceil {1}/{2} +
    {\frcoindno(H)}/{(4\cdot \aw(H))} \right\rceil$ is larger than
    some universal
    constant $t_0$.
    Then there exists $\hat H \in \WeakMinors(\scQ(H))$ with $|\hat H|
    \ge 2/3 \cdot t$ and
    $\rank(\hat H) \le 2$ such that $\tw(\hat H) \ge |\hat H| / 24 - 1$ and
    $\frcoindno(\hat H) \ge |\hat H| / 4$.
\end{lemma}
\begin{proof}
    Let $X \subseteq V$ be a minimum-size co-independent set of $H$,
    and let $I$ be a
    maximum independent set of $\trim{H}{X}$.
    Then \cref{lem:int_gap_aw} applied to $\trim{H}{X}$ yields
    \begin{align}\label{eq:sqrt-bound}
        |I| = \alpha(\trim{H}{X}) \ge \frac{1}{2} +
        \frac{\frindno(\trim{H}{X})}{4\cdot\aw(\trim{H}{X})}
         & \ge \frac{1}{2} + \frac{\frcoindno(H)}{4 \cdot \aw(H)}.
    \end{align}
    The second inequality holds as
    $\frindno(\trim{H}{X})=\fredgeco(\trim{H}{X})\ge
    \frcoindno(H)$ by \cref{lem:LP-duality} and the definition of
    $\frcoindno(H)$, and
    because $\aw(\trim{H}{X}) \le \aw(H)$ by \cref{lem:aw_monotone}.
    As $|I|$ is integral, it follows that $|I| \ge t$.

    Next, we use~$I$ to exhibit an almost bipartite graph~$B$, from
    which we build the trimmed quotient~$\hat H$.
    Let $R\coloneqq V\setminus X$; so $R$ is disjoint from $I$, and is
    independent in $H$
    since $X$ is co-independent. Thus, the induced trimmed
    sub-hypergraph~$B\coloneqq\trim{H}{I \cup R}$ has rank at most two,
    and (save for the singleton edges) is bipartite.

    \begin{claim}
        \label{claim:F-contains-large-matching}
        The hypergraph~$B$ contains a matching~$M$ of size at least $|I|$.
    \end{claim}
    \begin{claimproof}
        Without loss of generality we may assume $B$ has no singleton edges and is
        therefore a graph (the singletons are irrelevant to our argument).
        For an indirect proof, assume that the matching number of~$B$ is
        at most $|I|-1$.
        By König's theorem,~$B$ has a vertex cover $C \subseteq V(B)$ of
        size at most~$|I|-1$.
        Let $X' \coloneqq (X \setminus I) \cup C$.
        Since $I \subseteq X$ and $|C|<|I|$, we have $|X'| < |X|$.
        One can readily verify that $R' \coloneqq V \setminus X' = (I
        \cup R) \setminus
        C$.
        However, this implies that $R'$ is independent in $H$, as both $I$ and $R$ are
        independent in $H$ by construction and if $u \in I$ and $v\in R$
        are adjacent in
        $H$ then at least one of them is in $C$ as it is a vertex cover of~$B$.
        We conclude that $X'$ is co-independent in $H$ and $|X'|<|X|$,
        contradicting the
        minimality of $X$.
        Thus,~$B$ contains a matching of size at least $|I|$, as claimed.
    \end{claimproof}

    Next, we consider the trimmed sub-hypergraph of~$H$ induced by the
    vertices~$V(M)$ of
    the matching, and we use quotients to construct a hypergraph $\hat
    H$ with large
    $\tw(\hat H)$ and $\frcoindno(\hat H)$.
    \begin{claim}\label{clm:build-large-tw-graph}
        If $|I|$ is at least a large enough, global constant~$s$, then
        the hypergraph
        $\trim{H}{V(M)}$ has a quotient~$\hat H$ such that $|\hat H| =
        2/3 \cdot |I|$ and
        $\rank(\hat H) \le 2$ and $\tw(\hat H) \ge |\hat H| / 24-1$ and
        $\frcoindno(\hat H) \ge |\hat H|/4$.
    \end{claim}
    \begin{claimproof}
        For every $k \ge s$ where $s$ is a large enough constant, Dvořák
        and Norin show
        the existence of a $3$-regular graph $G_k$ with~$k$ edges,
        $2/3 \cdot k$ vertices,
        and treewidth at least $k/36 -1$; see \cite[Lemma~5 and
        Corollary~7]{DBLP:journals/siamdm/DvorakN16}.
        We now identify endpoints of~$M$ to construct~$G_{|I|}$. Formally, we apply
        \cref{def:hypergraph_quotient} as follows.
        Let $\tau$ be a partition of the vertices of~$M$ so that the
        quotient $M/\tau$ is isomorphic to~$G_{|I|}$.
        Let then $\hat H \coloneqq \trim{H}{V(M)}/\tau$ be the
        corresponding quotient of
        $\trim{H}{V(M)}$.
        By construction, $|\hat H| = 2/3 \cdot |I|$ and~$\rank(\hat H) \le 2$.
        Since $\hat H$ contains all edges of $M/\tau$, its treewidth is at least the
        treewidth of $G_{|I|}$, which in turn is at least $|I| / 36-1 =
        |\hat H| / 24-1$.
        Finally, any independent set of	$\hat H$ has size at most $|\hat H| / 2$, because,
        ignoring the singleton hyperedges, $\hat H$ is a regular graph.
        Thus, any co-independent set $X$ of $\hat H$ satisfies~$|X| \ge |\hat H| /
        2$.
        As $\rank(\hat H) \le 2$, this implies $\fredgecoS{\hat H}(X) \ge |X| / 2 \ge
        |\hat H| / 4$.
        The claim follows.
    \end{claimproof}

    The hypergraph $\hat H$ therefore satisfies the invariant
    inequalities of the statement.
    To see that~${\hat H \in \WeakMinors(\scQ(\scH))}$, note that we
    in fact obtain $\hat
    H$ by first taking a single quotient $\tau$ over $H$ (the one that
    identifies the edges of
    $M$ as described above) and then trimming the resulting hypergraph to
    $V(M / \tau)$.
\end{proof}

We now prove the lower bounds of \cref{thm:classification_subsprob}; which then in total
completes the proof of \cref{thm:classification_subsprob}.
\rstmthmone
\begin{proof}
    Let $\scH$ be a recursively enumerable hypergraph family with
    $\frcoindno(\scH)=\infty$.

    Suppose for some $g(n) = o(\sqrt[4]{n})$ there is an algorithm $\mathrm{A}$ that
    solves $\subsprob(\scH)$ in time $f(H) \cdot
    \size{G}^{g(\frcoindno(H))}$; we show ETH fails.
    First, the following chain of reductions holds by our results in \cref{sub:homsprob}.
    \begin{align*}
        \homsprob(\WeakMinors(\scQ(\scH))) & \fptredlin
        \cphomsprob(\WeakMinors(\scQ(\scH)))
        \tag{\text{\Cref{lem:homs_to_cphoms_new}}}
        \\
        ~                                  & \fptredlin
        \cphomsprob(\scQ(\scH)) \tag{\text{\Cref{lem:WM_fptred_new}}}
        \\
        ~                                  & \fptredlin
        \cfhomsprob(\scQ(\scH)) \tag{\text{\Cref{lem:cp_to_cf_new}}}
        \\
        ~                                  & \fptredlin \homsprob(\scQ(\scH))
        \tag{\text{\Cref{lem:cf_to_nocol_new}}}
        \\
        ~                                  & \fptredlin \subsprob(\scH)
        \tag{\text{\Cref{clm:explicit_sub_dedekind}}}
    \end{align*}
    Since $\fptredlin$ is transitive (see \cref{sec:prelim}), we conclude
    that $\homsprob(\WeakMinors(\scQ(\scH))) \fptredlin \subsprob(\scH)$.
    Thus, we solve $\homsprob(\WeakMinors(\scQ(\scH)))$ in
    time $f(H) \cdot \size{G}^{g(\frcoindno(H))}$, too.
    That is, there is an algorithm that, given an instance $(F,G)$ of
    $\homsprob(\WeakMinors(\scQ(\scH)))$, computes $\#\homs{F}{G}$ in
    time $f(H) \cdot \size{G}^{g(\frcoindno(H))}$ where $H$ is a
    hypergraph in $\scH$ such that $F$ is a trimmed induced subhypergraph
    of a quotient of $H$.
    As \(H\) is found by enumerating \(\scH\), its size depends only on \(F\).

    Now, fix any unbounded function $c$ such that $c(n) = o(n)$ and $g(n) = o\big(\sqrt[4]{c(n)}\big)$.
    We define the following two hypergraph families as
    \begin{align*}
        \scH_{\aw} & \coloneqq \left\{ H \in \scH : \aw(H) \ge
        c(\frcoindno(H)) \right\} \quad\text{and}
        \\
        \scH_{\tw} & \coloneqq \left\{  H \in \scH : \aw(H) <
        c(\frcoindno(H)) \right\}.
    \end{align*}
    As $\scH = \scH_{\aw} \, \cup \, \scH_{\tw}$, we have $\frcoindno(\scH_{\aw}) = \infty$ or
    $\frcoindno(\scH_{\tw})=\infty$.
    We consider the two cases separately, and we show that in any case ETH fails.

    \textbf{Case 1:} $\frcoindno(\scH_{\aw}) = \infty$.
    As $\scH_{\aw} \subseteq \scH \subseteq \WeakMinors(\scQ(\scH))$, we also have
    $\homsprob(\scH_{\aw}) \fptredlin \homsprob(\WeakMinors(\scQ(\scH)))$.
    Our assumptions above allow us to solve $\homsprob(\scH_{\aw})$ in time
    \[
        f(H) \cdot \size{G}^{g(\frcoindno(H))} \le f(H) \cdot
        \size{G}^{o\big(\sqrt[4]{c(\frcoindno(H))}\big)} \le f(H) \cdot
        \size{G}^{o\big(\sqrt[4]{\aw(H)}\big)},
    \]
    where the first inequality holds by the choice of $c$, and the second inequality holds
    by the definition of $\scH_{\aw}$.
    However, since $\frcoindno(\scH_{\aw})=\infty$, the fact that $c$ is
    unbounded and the
    definition of $\scH_{\aw}$ imply $\aw(\scH_{\aw})=\infty$.
    By the conditional lower bound of \cref{lem:unbounded_aw_LB}, the running time
    above implies that ETH fails.

    \textbf{Case 2:} $\frcoindno(\scH_{\tw}) = \infty$.
    Since $\frcoindno(\scH_{\tw}) = \infty$ and $c(n) = o(n)$, the ratio
    ${\frcoindno(H)}/{\aw(H)}$ diverges over~${H \in \scH_{\tw}}$.
    Define then the family $\hat \scH_{\tw} = \{\hat H: H \in
    \scH_{\tw}\}$, where for
    every $H \in \scH_{\tw}$ we let $\hat H$ be the hypergraph of
    \cref{lem:hat_H_subs}---note that $\hat H \in \WeakMinors(\scQ(\scH))$.
    By \cref{lem:hat_H_subs} and said divergence, the family $\hat
    \scH_{\tw}$ is infinite, satisfies~${\frcoindno(\hat \scH_{\tw})=\infty}$, has
    $\rank(\hat \scH_{\tw}) \le 2$, and moreover every $\hat H \in \hat\scH_{\tw}$ large
    enough satisfies~${\frcoindno(\hat H) \le 24(\tw(\hat H) +1)}$ since trivially
    $\frcoindno(\hat H) \le |\hat H|$ for every hypergraph $\hat H$.
    By the choice of $g$, this implies we solve
    $\homsprob(\hat{\scH}_{\tw})$ in time
    \[
        f(\hat H) \cdot \size{G}^{g(\frcoindno(\hat H))} \le f(\hat H) \cdot
        \size{G}^{g(24(\tw(\hat H)+1))} \le f(\hat H) \cdot
        \size{G}^{o(\tw(\hat H) / \ln
        \tw(\hat H))}.
    \]
    Note, however, that $\rank(\hat{\scH}_{\tw})\le 2$; by
    \cref{lem:unbounded_tw_LB},
    then, ETH fails.
    This concludes the proof.
\end{proof}

Finally, we prove \cref{rem:4-3-1}.

\rstmtonelemone
\begin{proof}
    For every pair of positive integers $n,k$ with $n \geq k$, let $H_{n,k}$ be the
    hypergraph with $V(H_{n,k})=[n]$ and $E(H_{n,k}) = {\binom{[k]}{2}} \cup \{[n]\}$.
    Further, set $\scH \coloneqq \{H_{n,k} : n \ge k \ge 1\}$.
    For all $n\geq k >0$, all vertices of $H_{n,k}$ are covered by the
    edge $[n]$; thus
    $\frcoindno(\scH) \le \fredgeco(\scH) = 1$.
    We give a polynomial-time reduction from~$\#\clique$ to $\subsprob(\scH)$.
    Let $(G,k)$ be the input to $\#\clique$.
    We construct a hypergraph $\hat G$ by adding the hyperedge~$V(G)$ to $G$.

    Observe that there is a bijection between the cliques in $G$ and the
    sub-hypergraphs
    of $\hat G$ isomorphic to $H_{n,k}$.
    Indeed, if $K=G[X]$ is a $k$-clique of $G$ for some $X \subseteq V$ then the
    edge-sub-hypergraph $\hat H$ of $\hat G$ defined by $V(\hat H)=V(\hat
    G)$ and $E(\hat
    H) = \{V(\hat G)\} \cup E(K)$ is isomorphic to $H_{n,k}$.
    Vice versa, if $\hat H$ is a sub-hypergraph of $\hat G$ isomorphic to $H_{n,k}$ then
    necessarily $V(\hat H)=V(\hat G)$ and moreover $E(\hat H) = \{V(\hat
    G)\}\cup E(K)$
    where $K$ is some $k$-clique of $G$.
    Thus, $\subs{K_k}{G} = \subs{H_{n,k}}{\hat G}$, concluding the proof.
\end{proof}

%% file: s6_indsub_counting.tex
\section{Complexity classification for
  \texorpdfstring{$\indsubsprob$}{Induced Sub-Hypergraph Counting}}
\label{sec:indsub}
Next, we prove
\cref{thm:classification_indsubsprob,lem:quasiP_indsubs,lem:GI_hard_indsubs}.

\rstmthmtwo
\begin{proof}
    \textbf{Upper bounds}

    Assuming $\fredgeco(\scH)< \infty$, we prove
    $\homsprob(\scQ(\scS(\scH)))$ is solvable
    in the prescribed time.
    The claim follows as~$\indsubsprob(\scH)\fptredlin\homsprob(\scQ(\scS(\scH)))$ by \cref{clm:explicit_sub_dedekind}.

    Fix any $Q \in \scQ(\scS(\scH))$.
    By definition, there exist hypergraphs $H \in \scH$ and $H' \in
    \scS(H)$ such that $Q
    \in \scQ(H')$.
    Note that $\fredgeco(Q) \le \fredgeco(H') \le \fredgeco(H)$.
    Indeed, every fractional edge-cover of $H$ exists in $H'$, too; and
    every fractional
    edge-cover in $H'$ yields a fractional edge-cover of the same weight
    in $Q$, by simply
    replacing each edge with its quotient.
    Therefore, we have$ \fredgeco(\scQ(\scS(\scH))) \le \fredgeco(\scH)$.
    Moreover, we obtain $\fhtw(\scQ(\scS(\scH))) \le \fredgeco(\scQ(\scS(\scH)))$
    by taking the
    trivial tree decompositions of the hypergraphs.
    We conclude that $\fhtw(\scQ(\scS(\scH))) \le \fredgeco(\scH) < \infty$.
    By \cref{2.21-2}, $\homsprob(\scQ(\scS(\scH)))$ is solvable in
    time $f(H) \cdot \size{G}^{\fredgeco(H)+O(1)}$, as claimed.

    \paragraph*{Lower bounds}
    By chaining \cref{lem:WM_fptred_new,clm:explicit_sub_dedekind}, and using transitivity
    of $\fptredlin$ as done for $\subsprob(\scH)$, we obtain
    \begin{equation*}
        \homsprob(\WeakMinors(\scS(\scH))) \fptredlin \indsubsprob(\scH).
    \end{equation*}
    Suppose then for some $g(n) = o(\sqrt[4]{n})$ there exists an
    algorithm $\mathrm{A}$
    that solves $\indsubsprob(\scH)$ in time $f(H) \cdot \size{G}^{g(\fredgeco(H))}$.
    Thus, we can solve $\homsprob(\WeakMinors(\scS(\scH)))$ in time $f(H) \cdot
    \size{G}^{g(\fredgeco(H))}$, too.
    Fix any unbounded function $c$ such that $c(n) = o(n)$ and $g(n) = o\big(\sqrt[4]{c(n)}\big)$.
    We define the following two hypergraph families as
    \begin{align*}
        \scH_{\aw} & \coloneqq \left\{ H \in \scH : \aw(H) \ge c(\fredgeco(H))
        \right\}
            \quad\text{and}
        \\
            \scH_{\tw} & \coloneqq \left\{  H \in \scH : \aw(H) < c(\fredgeco(H))
            \right\}.
    \end{align*}
    As $\scH = \scH_{\aw} \, \cup \, \scH_{\tw}$, we have
    $\fredgeco(\scH_{\aw}) = \infty$ or $\fredgeco(\scH_{\tw})=\infty$.
    We consider these two cases separately, and we show that in any case ETH fails.

    \textbf{Case 1:} $\fredgeco(\scH_{\aw}) = \infty$.
    As $\scH_{\aw} \subseteq \scH \subseteq \WeakMinors(\scS(\scH))$, we have
    $\homsprob(\scH_{\aw}) \fptredlin \homsprob(\WeakMinors(\scS(\scH)))$.
    By our assumptions above, we can solve $\homsprob(\scH_{\aw})$ in time
    \[
        f(H) \cdot \size{G}^{g(\fredgeco(H))} \le f(H) \cdot
        \size{G}^{o\big(\sqrt[4]{c(\fredgeco(H))}\big)} \le f(H) \cdot
        \size{G}^{o\big(\sqrt[4]{\aw(H)}\big)},
    \]
    where the first inequality holds by the choice of $c$, and the second one by the
    definition of $\scH_{\aw}$.
    However, since $\fredgeco(\scH_{\aw})=\infty$, the fact that $c$ is
    unbounded and the definition of $\scH_{\aw}$ imply $\aw(\scH_{\aw})=\infty$.
    By the conditional lower bound of \cref{lem:unbounded_aw_LB}, then,
    the running time
    above implies that ETH fails.

    \textbf{Case 2:} $\fredgeco(\scH_{\tw}) = \infty$.
    In this case we construct a family $\hat{\scH}_{\tw}$ as follows.
    For every $H \in \scH_{\tw}$ fix a maximum independent set $I
    \subseteq V(H)$, and let
    $\hat H$ be the hypergraph obtained from $H$ by (i) adding to $E(H)$ every edge
    $\{u,v\} \in {\binom{I}{2}}$, and (ii) trimming the resulting hypergraph to $I$.
    Let then $\hat{\scH}_{\tw} = \{\hat H : H \in \scH_{\tw}\}$.
    Note that, by construction, $\hat{\scH}_{\tw} \subseteq
    \WeakMinors(\scS(\scH))$.
    Therefore, by our assumptions, we can solve
    $\homsprob(\hat{\scH}_{\tw})$ in time
    $f(\hat H) \cdot \size{G}^{g(\fredgeco(\hat H))}$ for all $\hat H \in
    \hat{\scH}_{\tw}$.

    Now we show how ETH fails.
    For every $H \in \scH_{\tw}$,
    \cref{lem:int_gap_aw,lem:LP-duality} yield
    \begin{align}
        |\hat H| = \alpha(H) \ge \frac{1}{2} + \frac{\alpha^*(H)}{4\, \aw(H)} =
        \frac{1}{2} + \frac{\fredgeco(H)}{4\, \aw(H)} \ge \frac{1}{2} +
        \frac{\fredgeco(H)}{4\, c(\fredgeco(H))}.
        \label{eq:I_large}
    \end{align}
    As $c(\fredgeco(H)) = o(\fredgeco(H))$, this implies that
    $\hat{\scH}_{\tw}$ contains
    arbitrarily large hypergraphs.
    Moreover, since every $\hat H \in \hat{\scH}_{\tw}$ is a clique graph
    with (possibly)
    additional singleton edges, then $\tw(\hat H) = \tw(K_{|\hat H|}) =
    |\hat H|-1$ and
    $\fredgeco(\hat H) = {|\hat H|}/{2}$, and therefore $\fredgeco(\hat
    H)={(\tw(\hat H)+1)}/{2}$.
    Clearly, this also implies $\fredgeco(\hat{\scH}_{\tw})=\infty$.
    By the choice of $g$, our running time over $\hat{\scH}_{\tw}$ then satisfies
    \[
        f(\hat H) \cdot \size{G}^{g(\fredgeco(\hat H))} \le f(\hat H) \cdot
        \size{G}^{g((\tw(\hat H)+1)/2)} \le f(\hat H) \cdot
        \size{G}^{o(\tw(\hat H) / \ln
        \tw(\hat H))}.
    \]
    Note, however, that $\rank(\hat{\scH}_{\tw})\le 2$; by
    \cref{lem:unbounded_tw_LB},
    then, ETH fails.
\end{proof}

\rstmtlemone
\begin{proof}
    Let $H\in \scH$ and $G$ be the input to $\indsubsprob(\scH)$.
    Since $\fredgeco(\scH)<\infty$, we have that there is a subset
    $A\subseteq E(H)$ of
    size $\Oh(\ln \size{H})$ that covers all vertices of $H$---this
    follows by the $(\ln
    n)$-integrality gap for fractional edge-covers (cf.~\cite[Chapter
    13.1]{Vazirani01}).

    Hence, we can, in time $\size{G}^{\Oh(\ln \size{H})}$, enumerate all
    assignments from
    $A$ to $E(G)$.
    Let $X_1,\dots,X_\ell$ denote the list of distinct images of those
    assignments (note the list can easily be computed in time
    $\size{G}^{\Oh(\ln\size{H})}$ from the list of
    all assignments from $A$ to $E(G)$).

    Finally, we count for how many of the $X_i$ the induced
    sub-hypergraph $G[X_i]$ is
    isomorphic to $H$, and each (hypergraph) isomorphism check can be done in time
    $(\size{H}+\size{G})^{(\log \size{H})^{\Oh(1)}}$~\cite[Corollary
    1.2]{HyperNeuen22}.
    The total running time is thus bounded by
    \begin{equation*}
        \size{G}^{\Oh(\ln \size{H})} + \size{G}^{\Oh(\ln \size{H})} \cdot
        (\size{H}+\size{G})^{(\log \size{H})^{\Oh(1)}} \leq
        (\size{H}+\size{G})^{(\ln
        \size{H})^{\Oh(1)}}.
    \end{equation*}
    This concludes the proof.
\end{proof}

The quasi-polynomial running time in \cref{lem:quasiP_indsubs} has two causes.
We solve a quasi-polynomial number of hypergraph isomorphism
subproblems, each taking a quasi-polynomial amount of time.
The proof of \cref{lem:quasiP_indsubs} can be seen as a
quasi-polynomial time Turing reduction of $\indsubsprob$ to
$\textupsc{Graph Isomorphism}$.

\rstmtlemtwo
\begin{proof}
    For every graph $G=(V,E)$ set $\hat G \coloneqq (V, E \cup \{V\})$
    and $\scH\coloneqq
    \{\hat G \mid G \text{ is a graph}\}$.
    From the definition of a fractional edge-cover, we obtain $\fredgeco(\scH)=1$.
    Moreover, for every pair of graphs $F,G$ we have $\hat F \cong \hat
    G$ if and only if
    $F \cong G$.
\end{proof}

%% file: s7_trim.tex
\section{Trimmed homomorphisms and sub-hypergraphs}\label{sec:trimmed_sec}
In this section, we prove results on the complexity of counting
trimmed homomorphisms,
trimmed sub-hypergraphs, and induced trimmed sub-hypergraphs (see
\cref{sec:prelim}).
Under standard computational hardness assumptions, we show that the
complexity is not
monotone in the trimmed homomorphism basis and that there exists no
efficiently computable
associative operation on hypergraphs that is multiplicative with
respect to trimmed
homomorphism counts.

We start by formally defining the problems that we consider.%
\footnote{For our results, we do not need to define
  $\#\textupsc{TrimSub}$, the subgraph
version of the problem.}

\noindent
\begin{minipage}{1\linewidth}%
    \begin{problem}-[$\#\textupsc{TrimHom}$]{$\#\textupsc{TrimHom}(\scH)$}
        \label{prob:trimhoms}
        \PInput{A pair of hypergraphs $H, G$ with $H \in \scH$}
        \POutput{$\#\trimhoms{H}{G}$}
        \PParameter{$|H|$}
    \end{problem}

    \begin{problem}+=[$\#\textupsc{IndTrimSub}$]{$\#\textupsc{IndTrimSub}(\scH)$}
        \label{prob:trimindsub}
        \PInput{A pair of hypergraphs $H, G$ with $H \in \scH$}
        \POutput{$\#\indtrimsubs{H}{G}$}
        \PParameter{$|H|$}
    \end{problem}
\end{minipage}

As in the case of normal sub-hypergraphs and homomorphisms,
$\#\indtrimsubs{H}{\star}$ can
be represented as a linear combination of  $\#\trimhoms{H'}{\star}$.
However, we show that the complexity of this representation has undesirable
properties.
\begin{restatable*}[Complexity monotonicity fails for induced trimmed
    sub-hypergraphs]{theoremq}{rstri}
    \dglabel{thm:main_trimmed}
    There exists a hypergraph family $\scH$ with the following properties.
    \begin{enumerate}
        \item $\trimhomsprob(\scH)$ is $\#\W[2]$-hard and, assuming the
            Strong Exponential
            Time Hypothesis (SETH), cannot be computed
            in time $f(k)\cdot \Oh(|V(G)|^{k-\varepsilon})$ for any
            computable function
            $f$
            and $\varepsilon > 0$.
        \item $\indtrimsubsprob(\scH)$ can be computed in polynomial time.
        \item For every $H \in \scH$ the trimmed-homomorphism expansion of the
            function $\#\indtrimsubs{H}{\star}$ contains
            $\#\trimhoms{H}{\star}$ with a
            non-zero coefficient.
            \qedhere
    \end{enumerate}
\end{restatable*}

We may think of \cref{thm:main_trimmed} as a counterexample to the
complexity monotonicity
principle.
If we write $\#\indtrimsubs{H}{\star}$ as a linear
combination of trimmed
homomorphisms,
\begin{align}
    \#\indtrimsubs{H}{\star} = \sum_{F} \gamma(F) \cdot \#\trimhoms{F}{\star},
\end{align}
then the coefficient of $H$ itself does not vanish, that is, we have
$\gamma(H) \ne 0$.
Now, if the complexity monotonicity principle held, this would imply that
computing $\#\indtrimsubs{H}{\star}$ is at least as hard as computing
$\#\trimhoms{H}{\star}$---formally, $\trimhomsprob(\scH) \fptred
\indtrimsubsprob(\scH)$.%
\footnote{Recall that this is indeed what happens in
    \cref{sub:homsprob,sec:sub} for
establishing the hardness of $\subsprob$ and $\indsubsprob$.}

However, such a reduction is ruled out by \cref{thm:main_trimmed}:
$\trimhomsprob(\scH)$
is hard, but $\indtrimsubsprob(\scH)$ is easy.
Thus, \emph{somewhere}, the approach sketched above has to fail.
Our second contribution is to pinpoint what exactly breaks.

\begin{restatable*}[Trimmed hypergraph homomorphisms have no useful
    ``tensor product'']{theorem}{rstto}
    \dglabel{cor:trimmed_no_FPT_product}%
    [cor:trimmed_no_FPT_product,lem:trimmed_homs_hard,%
    lem:trimmed_indsubs_easy,lem:aut-omatic_lemma,%
    lem:trimmed-strembs_to_homs,lem:trimming-lovasz,%
    thm:dedekind_new]
    Assuming $\ccFPT\neq \#\W[2]$, there is no associative binary
    operator (``product'') $\boxplus$ over hypergraphs that is
    fixed-parameter tractable and for all hypergraphs $F,G,H$ satisfies
    \[\#\trimhoms{F}{G \boxplus H} = \#\trimhoms{F}{G} \cdot
    \#\trimhoms{F}{H}.\qedhere\]
\end{restatable*}
In light of \cref{thm:main_trimmed,cor:trimmed_no_FPT_product}, one
might wonder if the issue is with trimmed homomorphisms; perhaps, if
we use instead \emph{standard} hypergraph homomorphisms, we would
bypass the obstacles of
\cref{thm:main_trimmed,cor:trimmed_no_FPT_product} and obtain
complexity monotonicity.
As our final contribution we show that, using standard homomorphisms,
complexity monotonicity fails at an even
more fundamental level: $\#\indtrimsubs{H}{\star}$ cannot even be
\emph{expressed} as a linear combination of functions $\#\homs{F}{\star}$.

\begin{restatable*}[Trimmed hypergraphs are incompatible with
    standard homomorphisms]{theorem}{rstna}
    \dglabel{thm:trim_basis_indep}[eq:impossible-trim-homomorphism-count-equality]
    Let $H$ be a hypergraph that contains at least one edge and let
    $f\colon\G\to\N$ be a function from the
    set~$\{\#\trimhoms{H}{\star},\#\trimsubs{H}{\star},\#\indtrimsubs{H}{\star}\}$.
    Then, there is no finitely supported function $\gamma$ such that
    \begin{equation}
        \label{eq:impossible-trim-homomorphism-count-equality}
        f(G) = \sum_{F\in\G}\gamma(F) \cdot \#\homs{F}{G}
    \end{equation}
    holds for all $G\in\G$, where $\G$ is the class of all hypergraphs.
\end{restatable*}

Based on these observations, we see two avenues to analyse the complexity of
$\indtrimsubsprob(\scH)$.
Either we need to adopt a basis other than the trimmed homomorphism
basis; that is, we
need to express~$\#\indtrimsubs{H}{\star}$ as a linear combination of
functions other than~$\#\trimhoms{H'}{\star}$.
Or, we need a complexity analysis that is \emph{not} based on
complexity monotonicity and
Dedekind interpolation.
Either approach seems challenging and poised to yield interesting future work.
In the rest of this section, we
prove~\cref{thm:main_trimmed,cor:trimmed_no_FPT_product,thm:trim_basis_indep}.

\subsection{Basic results for trimmed sub-hypergraphs}
We now adapt \cref{eq:embs=auts/subs,eq:strembs=auts/indsubs} to
their trimmed counterparts.
\begin{lemma}[Rewriting (induced) subgraph counts as (strong)
    embedding counts]
    \dglabel{lem:aut-omatic_lemma}
    For any hypergraphs \(H\) and \(G\), we have
    \begin{align*}
        \#\trimembs{H}{G}    & = \#\mathsf{Aut}(H) \cdot \#\trimsubs{H}{G}, \quad
        \text{and} \\
        \#\trimstrembs{H}{G} & = \#\mathsf{Aut}(H) \cdot \#\indtrimsubs{H}{G}.
    \end{align*}
\end{lemma}
\begin{proof}
    The proof is almost identical to the non-trimmed version.
    We present the argument for trimmed embeddings and subgraphs; the
    proof is similar for
    strong trimmed embeddings and trimmed induced subgraphs.

    To this end, observe that $\#\mathsf{Aut}(H)$ acts on
    $\trimembs{H}{G}$ in the canonical way.
    For $a \in \mathsf{Aut}(H)$ and $\varphi \in \trimembs{H}{G}$ set
    $(a \triangleright \varphi)(v)\coloneqq \varphi(a(v))$.
    This action is free, that is, $a\triangleright \varphi = \varphi$
    if and only if $a$ is the trivial automorphism.
    In other words, the stabiliser of each $\varphi$ just contains the
    trivial automorphism.
    Thus, by the Orbit-Stabiliser-Theorem, we have $\#\trimembs{H}{G} =
    \#\mathsf{Aut}(H)
    \cdot \#\mathcal{O}(H,G)$, where $\mathcal{O}(H,G)$ is the set of
    all orbits of this
    action.

    It thus remains to show that $\#\mathcal{O}(H,G) = \#\trimsubs{H}{G}$.
    We construct a bijection $b$ explicitly.
    Let $[\varphi]$ be the orbit of $\varphi$.
    Set $b([\varphi])$ as the edge-sub-hypergraph of
    $\trim{G}{\mathsf{im}(\varphi)}$
    obtained by keeping only the edges $\varphi(e)$ for $e \in E(H)$.
    This mapping is well-defined since two distinct trimmed embeddings
    $\varphi_1,\varphi_2$ in the same orbit yield the same edge-sub-hypergraph of
    $\trim{G}{\mathsf{im}(\varphi_1)}=\trim{G}{\mathsf{im}(\varphi_2)}$.

    Moreover, for injectivity of $b$, assume that $b([\varphi])=b([\psi])$.
    Then $\mathsf{im}(\varphi)=\mathsf{im}(\psi)$.
    Moreover, $\varphi$ is an isomorphism from $H$ to $b([\varphi])$ and
    $\psi^{-1}|_{V(b([\psi]))}$ is an isomorphism from $b([\psi]) =
    b([\varphi])$ to $H$.
    Consequently, $a:= \varphi \circ \psi^{-1}|_{V(b([\psi]))}$ is an
    automorphism with $a
    \triangleright \psi = \varphi$ and thus  $[\varphi]=[\psi]$.

    Finally, for surjectivity, let $H' \in \trimsubs{H}{G}$, that is, $H'$ is an
    edge-sub-hypergraph of $\trim{G}{X}$ for some $X \subseteq V(G)$,
    and there is an
    isomorphism $\pi$ from $H$ to $H'$. Clearly, $H'=b([\pi])$.
\end{proof}

Next, we prove that the terms $\#\trimhoms{H_1}{\star}$ and
$\#\trimhoms{H_2}{\star}$ are
distinct functions whenever $H_1 \ncong H_2$.
\begin{lemma}[Trimmed hypergraph homomorphism indistinguishability]
    \dglabel{lem:trimming-lovasz}
    Let $H_1$ and $H_2$ be hypergraphs.
    Then $H_1 \cong H_2$ if and only if for all hypergraphs $G$ we have
    $\#\trimhoms{H_1}{G} = \#\trimhoms{H_2}{G}$.
\end{lemma}
\begin{proof}
    The only-if direction is immediate.
    Hence, assume that for all hypergraphs $G$ we have $\#\trimhoms{H_1}{G} =
    \#\trimhoms{H_2}{G}$.
    It remains to show that $H_1 \cong H_2$.
    To this end, let~$\mathsf{TrimSur}(H'\to G')$ denote the set of all
    surjective trimmed
    homomorphisms from $H'$ to $G'$. Moreover, recall that
    $\mathsf{Sur}(H' \to G')$
    denotes the set of all surjective (non-trimmed) homomorphisms from~$H'$ to $G'$.
    Observe that
    \begin{equation}\label{eq:trimmed_lovasz_helper}
        |V(H')| = |V(G')| ~\Rightarrow~ \mathsf{TrimSur}(H'\to G') =
        \mathsf{Sur}(H'\to
        G').
    \end{equation}
    Observe further that, by inclusion-exclusion, we have
    \begin{align*}
        \#\mathsf{TrimSur}(H'\to G')
        & = \sum_{A \subseteq V(G')} (-1)^{|A|} \cdot \#\{\varphi \in \trimhoms{H'}{G'} \mid
        \forall v \in A: v \notin \mathsf{im}(\varphi) \},\\
        & = \sum_{A \subseteq V(G')} (-1)^{|A|} \cdot \#\trimhoms{H'}{G'\setminus A}.
    \end{align*}
    Recall that $G'\setminus A$ denotes the hypergraph obtained from $G'$ by deleting all
    vertices in $A$ and trimming the edges accordingly, that is $G'\setminus A =
    \trim{G'}{V(G')\setminus A}$.
    Consequently, we have
    \begin{align*}
        \#\mathsf{TrimSur}(H_1\to H_2) & = \sum_{A \subseteq V(H_2)}
        (-1)^{|A|} \cdot
        \#\trimhoms{H_1}{H_2\setminus A}
        \\
        ~& =\sum_{A \subseteq V(H_2)} (-1)^{|A|} \cdot
        \#\trimhoms{H_2}{H_2\setminus A} =
        \#\mathsf{TrimSur}(H_2\to H_2) > 0.
    \end{align*}
    Symmetrically, we also have $\#\mathsf{TrimSur}(H_2\to H_1)>0$.

    As a first consequence, this implies that $|V(H_1)|=|V(H_2)|$.
    Moreover, in combination with \eqref{eq:trimmed_lovasz_helper}, we obtain
    $\#\mathsf{Sur}(H_1\to H_2)>0$ and $\#\mathsf{Sur}(H_2\to H_1)>0$,
    that is, there are
    surjective (non-trimmed) homomorphisms from $H_1$ to $H_2$ and from
    $H_2$ to $H_1$,
    implying that $H_1 \cong H_2$.
\end{proof}

Similarly to \cref{lem:embs_to_homs,lem:strembs_to_homs}, we obtain analogous
transformations for the trimmed case.

\begin{lemmaq}[Rewriting trimmed embedding counts as trimmed
    homomorphism counts]
    \dglabel{lem:trimmed-embs_to_homs}
    For all hypergraphs $H$ and $G$ we have
    \begin{align*}
        \#\trimembs{H}{G} = \sum_{F\in\scQ(H)} \gamma(F) \cdot \#\trimhoms{F}{G},
    \end{align*}
    where $\gamma(F) \ne 0$ for all terms in the sum.
\end{lemmaq}

\begin{lemmaq}[Rewriting trimmed strong embedding counts as trimmed
    homomorphism counts]
    \dglabel{lem:trimmed-strembs_to_homs}
    Let $\scS(H)$ be the set of all edge-super-hypergraphs of $H$.
    For all hypergraphs $H$ and $G$ we have
    \begin{align*}
        \#\trimstrembs{H}{G} = \sum_{F \in \scQ(\scS(H))} \gamma(F) \cdot
        \#\trimhoms{F}{G},
    \end{align*}
    where $\gamma(F) \ne 0$ for every $F \in \scS(H)$.
\end{lemmaq}

\subsection{A counterexample to complexity monotonicity}

We proceed to the proof of \cref{thm:main_trimmed}.

\rstri

Consider the following family of hypergraphs.

\begin{definition}
    Let $B_k$ be the hypergraph with vertex set $[k]$ and with edge set
    $\{[k]\}$; that is, the hypergraph has only one edge that contains all vertices.
    Moreover, set $\scB=\{B_k \mid k \in \mathbb{N}\}$.
\end{definition}

We intend to show that \(\scB\) fulfils all claims of
\cref{thm:main_trimmed}; we show
each claim separately.

\paragraph*{\(\#\W[2]\)-hardness of \(\trimhomsprob(\scB)\)}
We prove a more complete claim which includes a fine-grained lower
bound based on the
Strong Exponential Time Hypothesis (SETH,~\cite{ImpagliazzoPZ01}).
To this end, we reduce from the problem $\CfCommonNeighProb$, which
we define next, after
introducing the required notation.

Let $G$ be a bipartite graph with $V(G) = Y \dotcup X$, and let $X=
X_1 \dotcup \dots
\dotcup X_k$ be a partition of $X$.
For any $y \in Y$, let $N(y) \subseteq X$ be the neighbourhood of $y$ in $G$.
A subset $S \subseteq X$ is a {\em colourful $k$-neighbourhood} of
$G$ if $S \subseteq
N(y)$ for some $y \in Y$ and $|S \cap X_i| = 1$ for every $i \in [k]$.
Set $\scN(G) \coloneqq \{S \subseteq X : S \text{ is a colourful
}k\text{-neighbourhood of }G\}$.

\begin{problem}[$\#\textupsc{CfCommonNeighbours}$]%
    {$\#\textupsc{CfCommonNeighbours}$}
    \label{prob:common_neigh}
    \PInput{Bipartite graph $G=(Y \dotcup X, E)$, and partition $X_1
    \dotcup \ldots \dotcup X_k$ of $X$}
    \POutput{$|\scN(G)|$}
    \PParameter{$k$}
\end{problem}

It is known that $\CfCommonNeighProb$ is $\#\W[2]$-hard and that,
assuming SETH, for every
fixed $k$ it cannot be solved in time $|V(G)|^{k-\varepsilon}$ for
any $\varepsilon > 0$.
This was shown in Lemma 19 in the full version of~\cite{DellRW19},
and more concisely
in~\cite[Section~3]{Mengel21}.%
\footnote{Observe that both~\cite{DellRW19,Mengel21} model
  $\CfCommonNeighProb$ via
counting answers to a star-shaped conjunctive query.}

\begin{lemma}
    \dglabel{lem:trimmed_homs_hard}[sub:homs_and_colhoms]($\trimhomsprob(\scB)$
    is $\#\W[2]$-hard)
    $\trimhomsprob(\scB)$ is $\#\W[2]$-hard and, assuming SETH, cannot
    be solved in time
    $f(k)\cdot \Oh(|V(G)|^{k-\varepsilon})$ for any computable function $f$ and
    $\varepsilon> 0$.
\end{lemma}
\begin{proof}
    We prove
    \(
        \CfCommonNeighProb \fptred \trimhomsprob(\scB).
    \)

    The reduction goes via the colour-prescribed and colourful versions of
    $\trimhomsprob(\scB)$, denoted with $\#\cptrimhomsprob(\scB)$ and
    $\#\cftrimhomsprob(\scB)$; which are defined as expected (consult
    \cref{sub:homs_and_colhoms}), with the exception that, here, the
    colouring $c : V(G)
    \to V(H)$ of the host hypergraph $G$ does \emph{not need} to be a
    homomorphism from
    $G$ to $H$.

    We proceed to show
    \(
        \#\cptrimhomsprob(\scB) \fptred \#\cftrimhomsprob(\scB) \fptred
        \trimhomsprob(\scB).
    \)

    For the first reduction, we observe
    \(
        \#\cftrimhoms{B_k}{G} = k! \cdot \#\cptrimhoms{B_k}{(G,c)}.
    \)

    For the second reduction, we use standard inclusion-exclusion
    arguments on the subset
    of colours spanned by the image of the homomorphisms.
    Formally, we have
    \begin{align*}
        \#\cftrimhoms{B_k}{G} = \sum_{A \subseteq [k]} (-1)^{|A|}
        \#\trimhoms{B_k}{G \setminus c^{-1}(A)}.
    \end{align*}
    Finally, we prove $\CfCommonNeighProb \fptred \#\cftrimhomsprob(\scB)$.

    To this end, let the input to $\CfCommonNeighProb$ be $G \coloneqq
    (Y \dotcup X, E)$
    and the partition $X_1 \dotcup\ldots\dotcup X_k$ of $X$.
    Consider $H=B_k$, and construct a hypergraph $\hat G=(\hat V, \hat
    E)$ as well as a
    colouring $c$ as follows.
    The vertex set is $\hat V = X$. The edge set is $\hat E = \{e_y : y
    \in Y\}$, where
    $e_y = N_G(y)$ for all $y \in Y$.
    Finally, construct the colouring $c : \hat V \to [k]$ by choosing
    $c(x)$ to be the
    unique index $i \in [k]$ such that $x \in X_i$.

    \begin{claim}
        There is a bijection between $\scN(G)$ and $\cptrimhoms{H}{(\hat G,c)}$.
        Further, $\hat G$ can be constructed in time $\poly(\size{H})
        \cdot \size{G}$.
    \end{claim}
    \begin{claimproof}
        Set $S \coloneqq \{x_1,\ldots,x_k\} \in \scN(G)$ where $x_i \in
        X_i$ for each $i =1,\ldots,k$.
        Consider the map $\varphi_S : V(H) \to \hat V$ defined by
        $\varphi_S(i)=x_i$.
        Since $S \in \scN(G)$, by construction there exists $e_y \in \hat
        E$ such that $S\subseteq e_y$.
        Therefore, $\varphi_S([k]) = S = \img(\varphi_S) \cap e_y$.
        Thus, $\varphi_S \in \trimhoms{H}{\hat G}$.
        Moreover, we have $c(\varphi_S(i))=i$ for all $i=1,\ldots,k$.
        Hence, we also have $\varphi_S \in \cptrimhoms{H}{(\hat G,c)}$.

        For the other direction, consider $\varphi \in
        \cptrimhoms{H}{(\hat G,c)}$, and set $S \coloneqq \img(\varphi)$.
        Clearly the colour-prescribedness implies $S =
        \{x_1,\ldots,x_k\}$ where $x_i \in
        X_i$ for each $i=1,\ldots,k$.
        Moreover,~$\varphi \in \trimhoms{H}{\hat G}$ implies the
        existence of $e_y \in \hat E$ such that $\varphi([k]) = \img(\varphi) \cap e_y=S\cap e_y$, and by
        construction this means that $S \subseteq N_G(y)$.
        Thus, $S \in \scN(G)$.

        We conclude that $|\scN(\hat G)| = \#\cptrimhoms{H}{(\hat G, c)}$.
        Finally, our construction implies that $\hat G$ can be constructed in time $\poly(\size{H}) \cdot \size{G}$.
    \end{claimproof}
    Therefore, $\CfCommonNeighProb \fptred \#\cptrimhomsprob(\scB)$.
    This proves that $\#\cptrimhomsprob(\scB)$ is $\ccSharpW{2}$-hard, as claimed.

    For the conditional lower bounds, suppose that
    $\#\cptrimhomsprob(\scB)$ admits an
    algorithm with running time $f(k)\cdot |V(\hat G)|^{k-\varepsilon}$
    for some $0
    <\varepsilon < 1$.
    For $k \ge 3$, the reduction above yields $\hat G$ such that $|V(\hat G)| \le
    |V(G)|+|k| = \Oh(|V(G)|)$.
    Moreover, $\size{G} = \Oh(|V(G)|^2)$ as $G$ is a graph, hence the
    reduction runs in
    time $\poly(k) \cdot \size{G} = \Oh(|V(G)|^2)$.
    We conclude that, for some $k \ge 3$, we can solve
    $\CfCommonNeighProb$ in time
        \(\Oh(|V(G)|^2) + |V(\hat G)|^{k-\varepsilon} = \Oh(|V(G)|^{k-\varepsilon}).\)
    As discussed above, this contradicts SETH.
\end{proof}

\paragraph*{A polynomial-time algorithm for $\indtrimsubsprob(\scB)$}

We proceed with the second item of \cref{thm:main_trimmed}.
\begin{lemma}
    \dglabel{lem:trimmed_indsubs_easy}
    ($\indtrimsubsprob(\mathcal{B})$ is solvable in polynomial time)
    $\indtrimsubsprob(\mathcal{B})$ is solvable in polynomial time.
\end{lemma}
\begin{proof}
    Suppose we are given $B_k$ and $G$ and write
    $\{e_1,e_2,\dots,e_m\}$ for the edges of
    $G$.
    For each vertex \(v\) of \(G\), we set \(t(v) \subseteq [m]\) to be
    the set of all
    indices of edges of \(G\) that \(v\) is part of; formally
    \[
        t(v) \coloneqq \{ i \mid v \in e_i \}.
    \]
    We write \(\scT\) for the set of all types of \(G\).
    For each \(t \in \scT\), we also write \(n_t\) for the number of
    vertices in \(G\)
    that have type \(t\).

    Now, we proceed as follows.
    First, we compute \(t(v)\) for each vertex \(v\) of \(G\).
    Next, we compute and return
    \( \sum_{t \in \scT}\binom{n_t}{k}.\)

    The key observation is that a non-empty subset $X$ of $V(G)$ satisfies
    $\trimsub{G}{X}\cong B_k$ if and only if all $v\in X$ have the same type.
    Hence, we indeed have
    \( \#\indtrimsubs{B_k}{G} =  \sum_{t \in \scT}\binom{n_t}{k}.\)

    Finally, observe that all computations can be performed in time
    polynomial in the
    input.
\end{proof}

\paragraph*{The trimmed-homomorphism expansion of
\(\#\indtrimsubs{B_k}{\star}\)}

Observe that every hypergraph is trivially an edge-super-hypergraph of itself.
Thus, the final item of \cref{thm:main_trimmed} follows from~\cref{lem:trimmed-strembs_to_homs}.

\subsection{On the non-existence of associative hypergraph products}
Next, we prove  \cref{cor:trimmed_no_FPT_product}.
We require the following definition.
\begin{definition}[Associative Hypergraph Product]
    An \emph{associative hypergraph product} is a mapping~$\boxplus$ from pairs of
    hypergraphs to hypergraphs such that, for all $F,G,H$, we have
    \( F \boxplus (G \boxplus H) \cong (F \boxplus G) \boxplus H.\)

    We say that $\boxplus$ is fixed-parameter tractable, if there is a
    computable function
    $f$ and an algorithm $\mathbb{A}$ that computes $G \boxplus H$ in
    time $f(|H|) \cdot
    |G|^{\Oh(1)}$.
\end{definition}

\rstto
\begin{proof}
    Assume for contradiction that there is a fixed-parameter tractable associative
    hypergraph product~$\boxplus$ satisfying the criterion of
    \cref{cor:trimmed_no_FPT_product}.
    Let $\mathbb{A}$ be the algorithm for computing $\boxplus$ in \ccFPT time.
    We claim that the existence of $\mathbb{A}$ implies the reduction
    \(\trimhomsprob(\mathcal{B}) \fptred \indtrimsubsprob(\mathcal{B}),\)
    which yields the contradiction by
    \cref{lem:trimmed_homs_hard,lem:trimmed_indsubs_easy}.
    For constructing the reduction, let $B_k$ and $G$ be the input to
    $\trimhomsprob(\mathcal{B})$.
    By \cref{lem:aut-omatic_lemma,lem:trimmed-strembs_to_homs}, we have
    \begin{align*}
        \#\indtrimsubs{B_k}{G} = \frac{1}{k!}\cdot \sum_{F \in
        \scQ(\scS(B_k))} \gamma(F) \cdot \#\trimhoms{F}{G},
    \end{align*}
    where $\gamma(F) \ne 0$ for every $F \in \scS(B_k)$.
    Observe that $k!=\#\mathsf{Aut}(B_k)$.
    Clearly, $B_k$ itself is contained in $\scQ(\scS(B_k))$, so the above linear
    combination includes $\#\trimhoms{B_k}{G}$ with a non-zero coefficient.

    Consequently, using $\mathbb{A}$ for the computation of $\boxplus$,
    and observing that
    \cref{lem:trimming-lovasz} guarantees~$\#\trimhoms{H_1}{\star}\neq\#\trimhoms{H_2}{\star}$ for $H_1 \neq H_2$.
    By \cref{thm:dedekind_new}, we can efficiently isolate the term
    $\#\trimhoms{B_k}{G}$ by
    querying the oracle for $\indtrimsubsprob(\mathcal{B})$.
\end{proof}

\subsection{Homomorphisms are not a good basis to count trimmed sub-hypergraphs}
\label{sub:trim_basis_indep}

Finally, we prove \cref{thm:trim_basis_indep}.

\rstna
\begin{proof}
    Assume for contradiction that \cref{eq:impossible-trim-homomorphism-count-equality}
    holds for all $G\in\G$.
    Then, $\gamma(F)=0$ for all $F\in\G$ with $E(F)=\emptyset$.
    To see this, let $G$ be a hypergraph consisting of $n\in\N$ isolated vertices.
    Applying \cref{eq:impossible-trim-homomorphism-count-equality} yields
    \[
        0=f(G)
        =\sum_{F\in\G}\gamma(F)\cdot\#\homs{F}{G}
        =\sum_{F\in\G\,:\,E(F)=\emptyset}\gamma(F)\cdot\#\homs{F}{G}
        =\sum_{F\in\G\,:\,E(F)=\emptyset}\gamma(F)\cdot n^{|F|},
    \]
    so $\sum_{F\in\G\,:\,E(F)=\emptyset}\gamma(F)\cdot n^{|F|}$ is a
    polynomial in $n$ with
    infinitely many roots.
    Thus, $\gamma(F)$ must be zero for all hypergraphs~$F\in\G$ that do
    not contain any edges.

    Set $N\coloneqq\max\{|F|\colon\gamma(F)\neq0\}$.
    Now, we construct the hypergraph~$G$ from~$H$ as follows.
    For each edge~$e$ of~$H$, we add~$N$ fresh and otherwise unused
    vertices to~$e$.
    This way, we have $H=\trim{G}{V(H)}$ and thus $f(G)>0$.
    Moreover, we have $\#\homs{F}{G}=0$ for all $F$ with
    $\gamma(F)\neq0$, because~$F$ only
    contains edges of size between~$1$ and $N$, whereas all edges of
    $G$ have size at
    least $N+1$.
    Applying \cref{eq:impossible-trim-homomorphism-count-equality} yields
    \[
        0=\sum_{F\in\G}\gamma(F)\cdot\#\homs{F}{G}
        =f(G)
        >0,
    \]
    which is a contradiction.
\end{proof}

%% file: sa_dedekind.tex
\section{On Dedekind interpolation}

In this appendix, we give a self-contained proof of the Dedekind interpolation result
used in \cref{sub:hyper_motifs}. We first formulate a circuit model tailored to the
interpolation argument, then prove a single-output version, and finally derive the full
statement of \cref{thm:dedekind_new}.

\ddcirc

\def\circin#1{\ensuremath\mathrm{in}[#1]}
\def\circinn#1#2{\ensuremath\mathrm{in}_{#1}[#2]}
\def\circout#1{\ensuremath\mathrm{out}_{#1}}
\def\alga{\mathbb{A}}

\begin{lemma}
    \dglabel{lem:better1}
    Write $(\mathrm{G},\ast)$ for a computable semigroup and let
    ${\varphi_1,\dots,\varphi_k\colon \mathrm{G} \to \Q}$ be \(k \ge
    1\) pairwise distinct and
    computable semigroup homomorphisms from $(\mathrm{G},\ast)$ to $(\Q,\cdot)$.
    For each subset \(S = \{s_1<\dots<s_t\} \subseteq\{2,\dots,k\}\), we  write
    \[m(S)\coloneqq g_{s_t}\ast\dots\ast g_{s_1}\ast g_1;\]
    and for \(S=\emptyset\) we write \(m(\emptyset)\coloneqq g_1\).

    There is an algorithm \(\alga\) that
    on input \(g_1, \dots, g_k \in \mathrm{G}\) with \(\varphi_1( g_1 )
    \neq 0\) and
    \(\varphi_1(g_i) \neq \varphi_i(g_i)\)
    for all \(1 < i \le k\),
    computes a circuit \(D(\varphi_1, \dots, \varphi_k)\) such that
    \begin{itemize}
        \item
            the circuit has an input gate
            \(\circinn{D}{m(S)}\)
            for each subset \(S\subseteq\{2,\dots,k\}\),
            as well as an input gate \(\circinn{D}{m(\emptyset)}\);
        \item
            the circuit has a single output gate \(\circout{D}\);
        \item
            \(D(\varphi_1, \dots, \varphi_k)\) has a depth of \(\Oh(k)\);
            each gate has fan-out \(1\) and fan-in at most \(2\);
        \item
            for each rational linear combination
            \(F = a_1 \varphi_1 + \sum_{i = 2}^k a_i \varphi_i\), the circuit has the
            following property.
            Assigning the value \(F(m(S))\) to the gate \(\circinn{D}{m(S)}\) makes the output
            gate \(\circout{D}\) evaluate to
            \(a_1\).
    \end{itemize}
    The algorithm \(\alga\) runs in time \(\Oh(2^k)\).
\end{lemma}
\begin{proof}
    For each \(i\in[k]\), let \(\mathcal P_i\coloneqq 2^{\{2,\dots,i\}}\).
    Note that for every \(S\in\mathcal P_{i-1}\), we have
    \(m(S\cup\{i\})=g_i\ast m(S)\).

    The algorithm \(\mathbb{A}\) iteratively constructs circuits
    \(D_1(\varphi_1), D_2(\varphi_1, \varphi_2), \dots, D_k(\varphi_1,\dots,\varphi_k)\).
    Consult \cref{fig:circ} for illustrations of \(D_1\), \(D_2\), and \(D_3\).

    For \(i=1\), the algorithm defines the specified circuit via
    \[D_1(\varphi_1) \coloneqq \{\circout{D_1} \gets \circinn{D_1}{m(\emptyset)}
    \cdot (1/\varphi_1(g_1))\}.\]
    In words, \(D_1\) has a single input gate indexed by \(m(\emptyset)\), which corresponds
    to the element \(m(\emptyset)=g_1\), and the circuit simply rescales its input by
    \(1/\varphi_1(g_1)\).

    For \(i>1\), assuming that \(D_{i-1}\) has already been constructed,
    the algorithm \(\mathbb{A}\) obtains \(D_i\) from \(D_{i-1}\) by replacing,
    for each \(S\in\mathcal P_{i-1}\), the input gate \(\circinn{D_{i-1}}{m(S)}\)
    with the circuit
    \[\circinn{D_i}{m(S\cup\{i\})} - \varphi_i(g_i)\cdot \circinn{D_i}{m(S)},\]
    and then multiplying the resulting output by
    \(1/(\varphi_1(g_i)-\varphi_i(g_i))\).

    We next verify by induction on \(i\) that, for every rational linear combination
    \(F \coloneqq a_1\varphi_1+\dots+a_i\varphi_i\), if each input gate
    \(\circinn{D_i}{m(S)}\) with \(S\in\mathcal P_i\) is assigned the value \(F(m(S))\),
    then the output gate \(\circout{D_i}\) has the value \(a_1\).

    For \(i=1\) and
    \(F=a_1\varphi_1\), we observe that assigning the value \(F(m(\emptyset))=F(g_1)\) to
    \(\circinn{D_1}{m(\emptyset)}\) yields
    \[
        \circout{D_1} \gets F(g_1)\cdot (1/\varphi_1(g_1))
        = a_1\varphi_1(g_1)/\varphi_1(g_1)=a_1,
    \]
    where we use the assumption \(\varphi_1(g_1)\neq 0\).

    For the inductive step, assume \(i>1\) and that the claim holds for \(D_{i-1}\).
    Fix any \(S\in\mathcal P_{i-1}\).
    By construction, the former input gate \(\circinn{D_{i-1}}{m(S)}\) of \(D_{i-1}\)
    is assigned the value
    \begin{align*}
        \circinn{D_{i-1}}{m(S)} \; & \!\gets
        \circinn{D_i}{m(S\cup\{i\})} - \varphi_i(g_i)\cdot \circinn{D_i}{m(S)}
        \\
                                & = F(m(S\cup\{i\})) - \varphi_i(g_i)\cdot F(m(S))                \\
                                & = F(g_i\ast m(S)) - \varphi_i(g_i)\cdot F(m(S))
                                \\
                                & = \sum_{j=1}^i a_j \cdot \varphi_j(g_i\ast m(S)) -
                                \varphi_i(g_i)\sum_{j=1}^i a_j \cdot \varphi_j(m(S)).
                                \intertext{Exploiting that each \(\varphi_j\) is a semigroup homomorphism, we rewrite this as}
                                & = \sum_{j=1}^i a_j \cdot \varphi_j(g_i) \cdot \varphi_j(m(S)) -
                                \varphi_i(g_i)\sum_{j=1}^i a_j \cdot \varphi_j(m(S))
                                \\
                                & = \sum_{j=1}^{i-1} a_j(\varphi_j(g_i) - \varphi_i(g_i)) \cdot
                                \varphi_j(m(S)).
    \end{align*}
    Thus, if we set \({a}'_j\coloneqq a_j(\varphi_j(g_i)-\varphi_i(g_i))\) and
    \(F'\coloneqq \sum_{j=1}^{i-1} {a}'_j\varphi_j\), then every input gate
    \(\circinn{D_{i-1}}{m(S)}\) is assigned the value \(F'(m(S))\).
    By the induction hypothesis, the output of \(D_{i-1}\) is therefore
    \({a}'_1 = a_1(\varphi_1(g_i)-\varphi_i(g_i))\).
    Finally, the last multiplication by \(1/(\varphi_1(g_i)-\varphi_i(g_i))\) is
    well-defined by assumption and yields \(a_1\) at the output gate \(\circout{D_i}\).

    In each recursive step, the depth increases by a constant, so the total depth is
    \(\Oh(k)\).
    Moreover, all gates have fan-in at most \(2\) and fan-out \(1\) by construction.
    Finally, \(D_i\) has exactly \(|\mathcal P_i| = 2^{i-1}\) input gates.
    If \(s_i\) denotes the total number of gates of \(D_i\), then
    \(s_1=2\).
    For \(i>1\), the circuit \(D_{i-1}\) has \(|\mathcal P_{i-1}|=2^{i-2}\) input gates,
    one for each subset \(S\subseteq\{2,\dots,i-1\}\).
    Each of these input gates is replaced by a constant-size gadget, and one final
    multiplication gate is added.
    Hence \(s_i \le s_{i-1}+c\cdot 2^{i-2}+c\) for some constant \(c\).
    Hence \(s_i=\Oh(2^i)\), and in particular \(D=D_k\) has size \(\Oh(2^k)\).

    Finally, observe that \(\alga\) runs in the time needed to write down the
    circuit, which is of size \(\Oh(2^k)\).
\end{proof}

\begin{figure}[tp]
    \renewcommand\tabularxcolumn[1]{m{#1}}
    \begin{subfigure}[t]{.3\linewidth}
        \centering
        \includegraphics[scale=1.7,page=1]{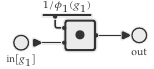}
        \caption{The circuit $D_1$.}
    \end{subfigure}
    \begin{subfigure}[t]{.67\linewidth}
        \centering
        \includegraphics[scale=1.7,page=2]{figs/2026-01_g01}
        \caption{The circuit $D_2$.}
    \end{subfigure}

    \bigskip

    \begin{subfigure}[t]{\linewidth}
        \centering
        \includegraphics[scale=1.75,page=3]{figs/2026-01_g01}
        \caption{The circuit $D_3$.}
    \end{subfigure}
    \caption{Illustration of the recursive circuit construction in
    \cref{lem:better1}.}
    \label{fig:circ}
\end{figure}

We now use \cref{lem:better1} as a building block for the full Dedekind interpolation
statement. The idea is to recover one coefficient at a time by rotating the role of
\(\varphi_1\), and then to assemble the resulting single-output circuits into one circuit
with \(k\) designated outputs.

\rstdi
\begin{proof}
    We use \(k\) calls to the algorithm of \cref{lem:better1}, where we
    shift the role of
    \(\varphi_1\) in each call. Finally, we combine the obtained
    circuits into a single
    Dedekind circuit.
\end{proof}